\documentclass[
notitlepage
,floatfix,
aps,
pra,
reprint,
superscriptaddress,
]{revtex4-2}

\usepackage{graphicx}
\usepackage[dvipsnames,svgnames]{xcolor}
\usepackage{bm,bbm}%
\usepackage{braket}
\usepackage{amsthm,amsmath,amssymb}
\usepackage{enumerate}
\usepackage{subfig,caption}
\usepackage{booktabs,multirow}
\usepackage{mathtools,bbding}
\usepackage[percent]{overpic}
\usepackage{microtype}
\usepackage{array,adjustbox,afterpage}
 \usepackage[T1]{fontenc}
 \usepackage[utf8]{inputenc}
 \usepackage{tikz}
 \usetikzlibrary{calc}
 
 \usepackage[english]{babel}

\usepackage{newtxtext,newtxmath}

\newtheoremstyle{red}{}{}{\itshape}{}{\color{red!80!black}\bfseries}{.}{ }{}

\definecolor{darkred}{rgb}{0.57,0,0.12}
\usepackage[colorlinks=true, linkcolor=MidnightBlue, urlcolor=darkred!75!black, citecolor=MidnightBlue]{hyperref}

\usepackage{cleveref}
\usepackage[most,breakable]{tcolorbox}

\AtBeginDocument{
\heavyrulewidth=.08em
\lightrulewidth=.05em
\cmidrulewidth=.03em
\belowrulesep=.65ex
\belowbottomsep=0pt
\aboverulesep=.4ex
\abovetopsep=0pt
\cmidrulesep=\doublerulesep
\cmidrulekern=.5em
\defaultaddspace=.5em
}

\let\nc\newcommand

\hypersetup{
    bookmarksnumbered=true, 
    breaklinks=true,
    unicode=false, 
    pdfstartview={FitH}, 
    pdftitle={}, 
    pdfnewwindow=true, 
    colorlinks=true, 
    allcolors=MidnightBlue!70!black!70!TealBlue
}

\DeclareMathOperator{\Tr}{Tr}
\DeclareMathOperator{\tr}{tr}

\DeclareMathOperator{\supp}{supp}

\let\Re\relax
\DeclareMathOperator{\Re}{Re}

\renewcommand{\H}{\mathcal{H}}

\newcommand{\proj}[1]{\ket{#1}\!\bra{#1}}

\renewcommand{\S}{\mathcal{S}}

\newcommand{\T}{\mathcal{T}}

\newcommand{\E}{\mathcal{E}}

\newcommand{\X}{\mathcal{X}}

\renewcommand{\L}{\mathcal{L}}

\renewcommand{\P}{\mathcal{P}}
\newcommand{\M}{\mathcal{M}}
\newcommand{\N}{\mathcal{N}}

\newcommand{\STAB}{{\FF_{\mathrm{STAB}}}}

\renewcommand{\bar}{\;\rule{0pt}{9.5pt}\right|\;}
\let\sbar\bar

\newcommand{\lset}{\left\{\left.}
\newcommand{\rset}{\right\}}

\DeclareMathOperator{\cone}{cone}

\newcommand{\RR}{\mathbb{R}}

\renewcommand{\SS}{\mathbb{S}}

\newcommand{\DD}{\mathbb{D}}

\newcommand{\OO}{\mathbb{O}}

\newcommand{\FF}{\mathbb{F}}

\newcommand{\mfR}{\mathfrak{R}}

\newcommand{\id}{\mathbbm{1}}
\newcommand{\idc}{\mathrm{id}}

\newenvironment{boxxed}[1]%
  {\expandafter\ifstrequal\expandafter{#1}{red}{\begin{tcolorbox}[colback=orange!3,colframe=orange!15]}{\begin{tcolorbox}[colback=white,colframe=gray!10,breakable,enhanced]}}%
  {\end{tcolorbox}}

\let\textsc\uppercase

{\begin{boxxed}{orange}\vspace{-\baselineskip}\begin{equation}}%
{\end{equation}\end{boxxed}}

\RenewDocumentEnvironment{boxed}{}%
{\begin{boxxed}{white}}%
{\end{boxxed}}

\newtheorem{theorem}{Theorem}

\newtheorem{proposition}[theorem]{Proposition}
\newtheorem{corollary}[theorem]{Corollary}

\newtheorem{lemma}[theorem]{Lemma}

\theoremstyle{definition}
\newtheorem*{remark}{Remark}

\theoremstyle{red}

\let\oldproofname\proofname
\renewcommand{\proofname}{\rm\bf{\oldproofname}}

  \nc{\MIO}{{\text{\rm MIO}}}
\nc{\DIO}{{\text{\rm DIO}}}
\nc{\SIO}{{\text{\rm SIO}}}
\nc{\IO}{{\text{\rm IO}}}
\nc{\CRNG}{{\text{\rm CRNG}}}

\nc{\lsetr}{\left\{\,}
\nc{\rsetr}{\right.\right\}}
\nc{\barr}{\,\rule{0pt}{9.5pt}\left|\;}
\nc{\ketbra}[2]{\ket{#1}\!\bra{#2}}
\let\sbarr\barr
\nc{\TT}{\mathbb{T}}
\nc{\CT}{\Upsilon}

\nc{\wt}{\widetilde}

\nc{\dia}{\!\!\Diamond}
\nc{\bdia}{\!\!\vardiamond}

\DeclareMathOperator{\aff}{aff}

\DeclarePairedDelimiter\floor{\lfloor}{\rfloor}
\DeclarePairedDelimiter\ceil{\lceil}{\rceil}

\newcommand{\bal}{\begin{equation}\begin{aligned}}
\newcommand{\eal}{\end{aligned}\end{equation}}
\newcommand{\bann}{\begin{equation*}\begin{aligned}}
\newcommand{\eann}{\end{aligned}\end{equation*}}
\newcommand{\dm}[1]{\ketbra{#1}{#1}}

\makeatletter
\def\renewtheorem#1{%
  \expandafter\let\csname#1\endcsname\relax
  \expandafter\let\csname c@#1\endcsname\relax
  \gdef\renewtheorem@envname{#1}
  \renewtheorem@secpar
}
\def\renewtheorem@secpar{\@ifnextchar[{\renewtheorem@numberedlike}{\renewtheorem@nonumberedlike}}
\def\renewtheorem@numberedlike[#1]#2{\newtheorem{\renewtheorem@envname}[#1]{#2}}
\def\renewtheorem@nonumberedlike#1{  
\def\renewtheorem@caption{#1}
\edef\renewtheorem@nowithin{\noexpand\newtheorem{\renewtheorem@envname}{\renewtheorem@caption}}
\renewtheorem@thirdpar
}
\def\renewtheorem@thirdpar{\@ifnextchar[{\renewtheorem@within}{\renewtheorem@nowithin}}
\def\renewtheorem@within[#1]{\renewtheorem@nowithin[#1]}
\makeatother

\nc{\phig}{{\phi_{\rm gold}}}

\newcommand{\texteq}[1]{\stackrel{\mathclap{\scriptsize \mbox{#1}}}{=}}
\newcommand{\textleq}[1]{\stackrel{\mathclap{\scriptsize \mbox{#1}}}{\leq}}

\allowdisplaybreaks

\begin{document}

\title{One-Shot Yield--Cost Relations in General Quantum Resource Theories}

\author{Ryuji Takagi}
\email{ryuji.takagi@ntu.edu.sg}
\affiliation{Nanyang Quantum Hub, School of Physical and Mathematical Sciences, Nanyang Technological University, 637371, Singapore}

\author{Bartosz Regula}
\email{bartosz.regula@gmail.com}
\affiliation{Department of Physics, Graduate School of Science, The University of Tokyo, Bunkyo-ku, Tokyo 113-0033, Japan}
\affiliation{Nanyang Quantum Hub, School of Physical and Mathematical Sciences, Nanyang Technological University, 637371, Singapore}

\author{Mark M. Wilde}
\email{mwilde@lsu.edu}
\affiliation{Hearne Institute for Theoretical Physics, Department of Physics and Astronomy, and Center for Computation and Technology, Louisiana State University, Baton Rouge, Louisiana 70803, USA}

\begin{abstract}
Although it is well known that the amount of resources that can be asymptotically distilled from a quantum state or channel does not exceed the resource cost needed to produce it, the corresponding relation in the non-asymptotic regime hitherto has not been well understood. Here, we establish a quantitative relation between the one-shot distillable resource yield and dilution cost in terms of transformation errors involved in these processes. Notably, our bound is applicable to quantum state and channel manipulation with respect to any type of quantum resource and any class of free transformations thereof, encompassing broad types of settings, including entanglement, quantum thermodynamics, and quantum communication. We also show that our techniques provide strong converse bounds relating the distillable resource and resource dilution cost in the asymptotic regime. Moreover, we introduce a class of channels that generalize twirling maps encountered in many resource theories, and by directly connecting it with resource quantification, we compute analytically several smoothed resource measures and improve our one-shot yield--cost bound in relevant theories. We use these operational insights to exactly evaluate important measures for various resource states in the resource theory of magic states. 
\end{abstract}

\maketitle

\section{Introduction}

Resource distillation and dilution are two fundamental resource manipulation tasks concerned, respectively, with the purification of noisy quantum resources and with the synthesis of general resources from pure ones. There have been many efforts to investigate the optimal rates of these tasks, known respectively as \emph{distillable resource yield} and \emph{resource dilution cost}, under various physical settings in the asymptotic regime with vanishing error~\cite{Bennett1996purification,Bennett1996concentrating,Brandao2013resource,Horodecki2013quantumness,Winter2016operational,Marvian2020coherence} and the non-asymptotic regime with non-zero error~\cite{Brandao2011one-shot,Bravyi2005magic,Zhao2018one-shot,Regula2018one-shot,Faist2018fundamental,Liu2019one-shot}.
In the asymptotic regime, an operational argument paralleling the second law of thermodynamics implies that the distillable resource is  bounded from above by the dilution cost, because otherwise one could produce an unbounded amount of resources by repeating the distillation-dilution cycle \textit{ad infinitum}~\cite{Bennett1996concentrating,Horodecki2002are,Donald2002uniqueness,Brandao2010reversible,Brandao2015reversible,Hayashi2016quantum,Winter2016operational,Watrous2018theory,Kuroiwa2020general}. 
(See also Refs.~\cite{Kumagai2013entanglement,Ito2015asymptotic,Kumagai2017secondorder} for related discussions of second-order asymptotics in the context of entanglement theory.)
However, in the non-asymptotic regime, the relation between the distillable resource and the resource cost is more subtle, especially when errors incurred in the transformations are also taken into account.

The recent work of Ref.~\cite{Wilde2021second} established a quantitative relation in the non-asymptotic regime for entanglement transformations that use local operations and classical communication, and it recovered the well-known asymptotic relation as a limit of the one-shot bound.
A natural question following these findings is to what extent we can strengthen the relations obtained there and establish the corresponding bounds in more general settings beyond the distillation and dilution of entangled states.

One possibility is to extend the result for entanglement manipulation to the manipulation of other types of quantum resources, such as quantum superposition~\cite{Baumgratz2014quantifying,Streltsov2017quantum,Regula2018one-shot,Zhao2018one-shot} and thermal non-equilibrium~\cite{Horodecki2013fundamental,Brandao2013resource}. 
The ultimate form of this extension is to consider \emph{general resource theories}~\cite{Chitambar2019quantum}, which encompass diverse types of physical resources in a single framework.
Results obtained for general resource theories are not only readily applicable to practical settings of interest but also provide foundational insights into what properties are universally shared by wide classes of quantum resources, which might seem very different from each other on the surface.   
Despite the apparent difficulty of getting meaningful results out of this abstract setup, recent years have seen significant development in the understanding of operational aspects of general resource theories, including the settings of resource erasure~\cite{Liu2017resource,Anshu2018quantifying}, discrimination tasks~\cite{Takagi2019operational,Takagi2019general,Uola2019quantifying,Uola2020all,Ducuara2020operational,Regula2021operational}, resource quantification~\cite{Regula2018convex,Chitambar2019quantum,gonda_2019}, and resource transformations~\cite{Horodecki2013quantumness,Brandao2015reversible,Gour2017quantum,Takagi2019general,Liu2019one-shot,Zhou2020general,Vijayan2020simple,Regula2020benchmarking,Kuroiwa2020general,Fang2020nogo-state}.

Another possible extension is to the setting of quantum channel manipulation.
Although research in resource theories has largely focused on quantum states, there are various situations in which quantum dynamics, generally represented by quantum channels, play the role of the main objects of study.
As a framework to accommodate these settings, \emph{resource theories of quantum channels} have recently been under rapid development~\cite{Bendana2017resource,Kaur2017amortized,Pirandola2017fundamental,Diaz2018using,Rosset2018resource,Bauml2019resource,Gour2020dynamical,Saxena2020dynamical,Faist2018fundamental,Theurer2019quantifying,Theurer2020dynamical,Yuan2020universal,Wang2019quantifying,Wang2019asymmetric_channel,Takagi2019general,Liu2019resource,Liu2020operational,Regula2021fundamental,Fang2020no-go,Gour2019how,Takagi2020application,Takagi2021optimal,Regula2020oneshot,Gour2021entropy,Rosset2020type,Schmid2020typeindependent,Yuan2020one-shot,Gour2020dynamicalresources,Salzmann2021symmetric}.
These dynamical resources include resource theories of states as special cases, because quantum states can be equivalently represented as quantum channels that prepare them. 
Thus, channel theories are strictly more general than state theories, and results that hold for arbitrary quantum channel resources are automatically carried over to state theories. 
Importantly, we can impose additional structures to channels under study so that the resulting resource theory can describe various settings involving quantum channels with pre-determined structures, such as the scenario of Bell non-locality represented by no-signalling boxes~\cite{de_Vicente2014nonlocality,Geller2014quantifying,Wolfe2020quantifyingbell}, and measurement incompatibility~\cite{Heinosaari2016invitation}, where a channel describes a set of POVMs~\cite{Regula2020oneshot}. 
Although the performance of channel distillation and one-shot channel manipulation in general resource theories has recently been studied~\cite{Regula2021fundamental,Fang2020no-go,Regula2020oneshot,Yuan2020one-shot,Gour2020dynamicalresources}, the relation between the one-shot distillable resource and resource cost under arbitrary free channel transformations has still been elusive.

In this work, we establish a fundamental quantitative relation between the one-shot distillable resource and dilution cost that is applicable to general resource theories of quantum channels, accomplishing the two extensions in a single framework.
We further extend the one-shot bound to the asymptotic regime and obtain strong converse bounds between transformation rates, revealing also new relations in the asymptotic manipulation of quantum states.
Our framework not only accommodates well-known quantum resources that have been previously studied using related approaches --- such as quantum entanglement and coherence --- but also other settings that involve dynamical resources as central objects of study, e.g., quantum communication.
Notably, a major part of our main results does not even assume the convexity or closedness of the set of free channels, unlike most of the works dealing with general resource theories.  
Moreover, our bounds hold for all chosen sets of free operations satisfying only minimal requirements, thus accommodating all physically motivated choices of free operations, including ones that may be equipped with complex structures.

An important technical aspect of our result is that, for the bounds to be most effective, they require the collapse of a class of resource measures based on R\'{e}nyi divergences. Focusing on resource theories of states, we provide an operational interpretation for this and related conditions in terms of a special class of free operations that reflects the essential features of the ``twirling'' operation commonly encountered in various theories~\cite{Werner1989quantum,Horodecki1999reduction}.
We use this class of operations to obtain analytical expressions for a number of smoothed resource measures, leading also to an improved yield--cost bound for the case of state manipulation. 
Finally, to demonstrate the applicability of our results, we use these operational insights in the resource theory of magic~\cite{Veitch2014resource,Howard2017application,Wang2019quantifying} and compute resource measures for several special states of interest.

We focus on discussing the main results and their implications in the main text, while deferring all proofs and technical details to the Appendices.

\section{Resource theories of states}\label{sec:theory of states}

We begin by introducing the framework of resource theories of quantum states, in which quantum states are the central objects under study~\footnote{In this manuscript, we restrict our attention to the cases in which underlying Hilbert spaces are finite dimensional.}.
Quantum information processing necessarily involves the manipulation of quantum states but commonly is performed under some physical restrictions, which limits the accessible set of quantum states and operations. 
A primary example is a scenario in which two parties, Alice and Bob, are physically separated.
Unless they are connected by a quantum communication channel, they can only apply local quantum operations and exchange classical messages. 
This restriction limits the accessible set of operations to the class known as local operations and classical communication (LOCC). 
Importantly, LOCC channels cannot create entanglement for free but can only prepare states that are separable~\cite{Horodecki2009quantum}. 

The framework of resource theories allows us to  deal formally with such restrictions~\cite{Chitambar2019quantum}.
It provides a recipe to quantify the amount of  resource that cannot be created due to the restriction and characterize feasible resource manipulation using only accessible operations.
Each resource theory comes with a set $\FF$ of \emph{free states}  and a set $\OO$ of \emph{free operations}.
They are subsets of quantum states and quantum channels (completely positive trace-preserving maps), respectively, and are considered to be accessible for free under a given setting.   
The minimal requirement for free operations is that they should not create resourceful states out of free states, i.e., $\Lambda\in\OO\Rightarrow\Lambda(\sigma)\in\FF,\forall \sigma\in\FF$. 
Physical considerations usually impose more constraints, but it is useful to consider the set of free operations consisting of all the operations satisfying this minimal requirement; we call this the maximal set of free operations and write it as~$\OO_{\max}$.
Consequently, any valid set of free operations is a subset of~$\OO_{\max}$.

A major strength of resource theories is that they allow us to evaluate the resourcefulness of a given state quantitatively. 
It is made possible by introducing \emph{resource measures}, which are functions from quantum states to real numbers that represent the amount of resource contained in the state. 
The requirement for resource measures is that they output the same value for free states, and they do not increase under the application of free operations. 

Here, we consider two resource measures that can be defined for an arbitrary set $\FF$ of free states and an arbitrary set $\OO$ of free operations. 
The first one is known as the \emph{generalized robustness}, or alternatively, \emph{max-relative entropy measure}~\cite{Vidal1999robustness,Steiner2003generalized,Datta2009min,Datta2009max} defined as
\bal
D_{\max,\FF}(\rho)&\coloneqq\inf\lset \log(1+s) \sbar \frac{\rho+s\tau}{1+s}\in\FF,\,\tau\in\DD\rset\\
&= \inf\lset\lambda\sbar\rho\leq 2^\lambda\sigma,\,\sigma\in\FF\rset
\label{eq:robustness def},
\eal
where $\DD$ is the set of all quantum states.
The other type of resource measure relevant to this work is \emph{min-relative entropy measure}~\cite{Datta2009min}, defined as 
\bal
 D_{\min,\FF}(\rho)&\coloneqq \inf_{\sigma\in\FF} D_{\min}(\rho\|\sigma),
\eal
where $D_{\min}(\rho\|\sigma)\coloneqq -\log\Tr[\Pi_\rho \sigma]$, with $\Pi_\rho$ denoting the projector onto the support of $\rho$.

These quantities lay out a platform for establishing relations between resource manipulation and resource quantification.
Resource distillation and dilution particularly stand out as important subclasses of resource manipulation tasks. 
Resource distillation is a protocol to transform a given state to a state in the family $\TT$ of \emph{reference states}  using free operations, while resource dilution is the opposite task in which a reference state is to be transformed to  the desired state.
The optimal performance of these tasks is characterized by the one-shot \emph{distillable resource} and \emph{dilution cost}, defined respectively as 
\bal
  d^\epsilon_{\OO} (\rho)\! &\coloneqq \sup \lset \mathfrak{R}_{\FF}(\Phi) \!\bar\! F(\Lambda(\rho), \Phi) \geq 1\!-\!\epsilon,\; \Phi\! \in \TT,\; \Lambda \in \OO \rset,\\
  c^\epsilon_{\OO} (\rho)\! &\coloneqq \inf \lset \mathfrak{R}_{\FF}(\Phi) \!\bar\! F(\rho, \Lambda(\Phi)) \geq 1\!-\!\epsilon,\; \Phi\! \in \TT,\; \Lambda \in \OO \rset,
\eal
where $F(\rho,\sigma)\coloneqq \|\sqrt{\rho}\sqrt{\vphantom{\rho}\sigma}\|_1^2$ is the fidelity \cite{Uhlmann1976transition}, and $\mathfrak{R}_\FF$ refers to an arbitrary resource monotone.  
For simplicity, we take $\mathfrak{R}_\FF=D_{\min,\FF}$ in the above definition throughout this manuscript.

To provide more insight into these definitions and connect them with notions of distillation and dilution familiar from commonly encountered resource theories, let us consider the case of quantum entanglement. Here, $\FF$ denotes all separable quantum states, and the family $\TT$ of reference states  can be taken simply as copies of the maximally entangled qubit Bell state:
\begin{equation}
    \TT = \{ \Phi_2^{\otimes n} \,|\; n \in \mathbb{N} \},
\end{equation}
where $\ket{\Phi_2} \coloneqq  \frac{1}{\sqrt{2}} (\ket{00} + \ket{11})$ and we use the shorthand notation $\Phi_2 = \proj{\Phi_2}$. Combined with the fact that~\cite{Shimony1995degree}
\begin{equation}
    D_{\min,\FF}(\Phi_2^{\otimes n}) = n,
\end{equation}
the quantities $d^\epsilon_{\OO}$ and $c^\epsilon_{\OO}$ can be alternatively understood as asking, respectively, how many copies of $\ket{\Phi_2}$ can be obtained from a given state, or how many copies of $\ket{\Phi_2}$ are necessary to produce a given state. We use more general definitions of distillation yield and resource cost to accommodate potentially more complex families of reference states and resource theories where a simple choice of reference states $\TT$ might not be available.


\section{Resource theories of channels}\label{sec:theory of channels}

We now introduce resource theories of quantum channels, extending the framework described in the previous section. 
To take into account the various structures that may be imposed on channels, we employ the approach of Ref.~\cite{Regula2020oneshot} and consider a subset $\OO_{\rm all}$ of quantum channels as the set of channels of interest. 
Each situation designates a set $\OO\subset \OO_{\rm all}$ of \emph{free channels} that does not contain any resource considered precious under the given setting.
We also need to introduce the set of free operations that are accessible to manipulate channels. 
General channel transformations are described by quantum \emph{superchannels}~\cite{Chiribella2008transforming,Chiribella2008quantum}. 
Superchannels are linear maps that map quantum channels to quantum channels, and they are physically realizable by sandwiching an input channel with a pre-processing channel and a post-processing channel, both of which can also be connected by a quantum memory.    
Letting $\OO_{\rm all}'$ denote the set of output channels of interest, the set of relevant superchannels is specified as 
\begin{equation}
\SS_{\rm all}\coloneqq \lset \Xi \sbar \Xi(\E)\in \OO_{\rm all}',\,\forall \E\in\OO_{\rm all}\rset.    
\end{equation}
Then, one can consider a subset $\SS\subset\SS_{\rm all}$ of superchannels  as \emph{free superchannels}, serving as free operations.  
The minimal requirement for free operations forces $\Theta\in\SS\Rightarrow\Theta(\M)\in\OO',\,\forall \M\in\OO$.  
Analogous to the case of state theories, we also define the maximal set of free superchannels as the set of superchannels satisfying the above minimal condition and call it $\SS_{\max}$.

We can analogously introduce resource quantifiers that do not increase under any set of free superchannels.
The relevant channel resource measures that we need are the max-relative entropy measure~\cite{Diaz2018using,Takagi2019general,Liu2019resource,Liu2020operational}, defined for every $\E\in\OO_{\rm all}$ as
\begin{equation}
 D_{\max,\OO}(\E) \coloneqq \inf_{\M\in\OO} \sup_{\psi} D_{\max}(\idc\otimes\E(\psi)\|\idc\otimes\M(\psi)),
\end{equation}
and the min-relative entropy measure~\cite{Liu2019resource,Liu2020operational}, defined as 
\begin{equation}
 D_{\min,\OO}(\E) \coloneqq \inf_{\M\in\OO} \sup_{\psi} D_{\min}(\idc\otimes\E(\psi)\|\idc\otimes\M(\psi)),   
\end{equation}
where the optimization is restricted to every pure input state $\psi$, without loss of generality~\cite{Gilchrist2005distance}. 

Then, we can formalize the tasks of resource distillation and dilution.
Resource distillation is a task that transforms a given channel $\E\in\OO_{\rm all}$ to a channel in the set $\TT\subset\OO_{\rm all}'$ of reference channels  using free superchannels $\SS$, while resource dilution transforms a reference channel to the desired channel~$\E$.
The optimal performance of these tasks is characterized by the one-shot distillable resource and dilution cost, defined respectively as 
\bal
  d^\epsilon_{\SS} (\E)\! &\coloneqq \sup \lset \mathfrak{R}_{\OO'}(\T) \!\bar\! F(\Theta(\E), \T) \geq 1\!-\!\epsilon,\; \T\! \in \TT,\; \Theta \in \SS \rset,\\
  c^\epsilon_{\SS} (\E)\! &\coloneqq \inf \lset \mathfrak{R}_{\OO'}(\T) \!\bar\! F(\E, \Theta(\T)) \geq 1\!-\!\epsilon,\; \T\! \in \TT,\; \Theta \in \SS \rset,
\eal
where
\begin{equation}
F(\E_1,\E_2)\coloneqq \min_\psi F(\idc\otimes\E_1(\psi),\idc\otimes\E_2(\psi)) 
\label{eq:fidelity def}
\end{equation}
is the worst-case fidelity~\cite{Gilchrist2005distance} (also called channel fidelity), and  $\mathfrak{R}_\OO$ refers to an arbitrary resource monotone, which we take as $\mathfrak{R}_\OO=D_{\min,\OO}$ in the above definition for simplicity.

In our discussion below, we will make the natural assumption that the reference channel $\T$ is a pure target channel, in the sense that $\idc\otimes\T(\psi)$ is pure for every pure state $\psi$. This covers general types of distillation processes encountered in physical resource theories, where $\T$ plays the role of a noiseless resource. If the input and output systems are identical, then the condition implies that $\T$ is unitary, while if the input and output systems have different dimensions, then it implies that $\T$ is an isometry. 
When the input space is trivial --- which is the case, e.g.,\ in the special case of resource theories of states --- then $\T$ reduces to a pure-state preparation channel.


\section{Yield--cost relations}

Our first main result is a relation between the one-shot distillable resource and dilution cost, which we dub a yield--cost relation, applicable to (1) dynamical resource theories of channels, (2) any set of free channels, and (3) any set of free superchannels defined for the given free channels.

\begin{theorem}
\label{thm:second law general}
     Suppose that for every channel $\T$ in the chosen reference set $\TT$, $\idc \otimes \T(\psi)$ is pure for every pure state $\psi$, and $\T$ satisfies $D_{\min,\OO}(\T) = D_{\max,\OO}(\T)$.
     Then, for all $\epsilon_1\in[0,1)$ and $\epsilon_2\in[0,1-\epsilon_1]$, for every quantum channel $\E$, and for every set $\SS\subseteq \SS_{\max}$ of free superchannels, the following inequality holds
 \bal
  d_{\SS}^{\epsilon_1}(\E)\leq c_{\SS}^{\epsilon_2}(\E) + \log f(\epsilon_1,\epsilon_2),
 \eal
 where $f(\epsilon_1,\epsilon_2)$ is a function defined as 
 \bal
  f(\epsilon_1,\epsilon_2)\coloneqq \min\left\{(1-\epsilon_1-\sqrt{\epsilon_2})^{-1},(\sqrt{1-\epsilon_2}-\sqrt{\epsilon_1})^{-2}\right\}
 \eal
 for $\epsilon_1+\sqrt{\epsilon_2}<1$ and 
  \bal
  f(\epsilon_1,\epsilon_2)\coloneqq (\sqrt{1-\epsilon_2}-\sqrt{\epsilon_1})^{-2}
 \eal
 otherwise. 
\end{theorem}
For $\epsilon_1+\sqrt{\epsilon_2}<1$, each quantity can be tighter than the other for certain error regions.
Indeed, direct calculation reveals that $(1-\epsilon_1-\sqrt{\epsilon_2})^{-1}\leq (\sqrt{1-\epsilon_2}-\sqrt{\epsilon_1})^{-2}$ if and only if 
\bal
 \frac{1}{2}(1-\sqrt{1-\epsilon_2})(1-\sqrt{\epsilon_2})\leq \epsilon_1\leq \frac{1}{2}(1+\sqrt{1-\epsilon_2})(1-\sqrt{\epsilon_2}).
\eal

This gives a fundamental relation between distillable resource and dilution cost in the one-shot regime under any chosen set of free superchannels.
In particular, when transformation errors are taken into account, the one-shot cost could be smaller than the one-shot yield.
Our result establishes a quantitative trade-off relation between a potential yield--cost gap and the transformation inaccuracy. 

The only assumption required for the result to hold is the choice of reference channels which satisfy $D_{\min,\OO}(\T) = D_{\max,\OO}(\T)$. The collapse of the two measures to the same value is a common property of maximally resourceful states or channels~\cite{Liu2019one-shot,Regula2020benchmarking,Regula2020oneshot}, which are precisely the most appropriate references employed in distillation and dilution protocols in practice.
We review examples of reference states and channels in several physical settings in Appendix~\ref{app:reference}.
We will also return to this condition in Section~\ref{sec:operational} to provide more intuition behind the assumption and give an operationally motivated understanding of it.

Theorem~\ref{thm:second law general} is a direct consequence of  the following lemma.

\begin{lemma}
 \label{lem:second law}
 Let $\E$ be an arbitrary quantum channel, and let $\epsilon_1,\epsilon_2$ be arbitrary real numbers such that $\epsilon_1\in[0,1)$, $\epsilon_2\in[0,1-\epsilon_1]$. Also, let $\T_1$ be a channel for which $\idc \otimes \T_1(\psi)$ is pure for every pure state $\psi$ and there exists $\Theta_1 \in \SS$ such that $F(\Theta_1(\E),\T_1)\geq 1-\epsilon_1$, and let $\T_2$ be an arbitrary channel for which there exists $\Theta_2 \in \SS$ such that $F(\Theta_2(\T_2),\E)\geq 1-\epsilon_2$. Then
    \begin{equation}\begin{aligned}
      D_{\min,\OO} (\T_1) \leq D_{\max,\OO} (\T_2) + \log f(\epsilon_1,\epsilon_2).
    \end{aligned}\end{equation}
where $f(\cdot,\cdot)$ is the function introduced in Theorem~\ref{thm:second law general}.
\end{lemma}

A natural question following these results is whether they smoothly connect to an asymptotic yield--cost relation.
In asymptotic distillation (and analogously for dilution), the goal of the tasks is commonly to obtain as many copies of a fixed reference channel as possible from multiple copies of the given channel, in such a way that the transformation error vanishes in the limit of infinite copies. The figure of merit is then the ratio of the number of obtained copies of the reference channel to the number of used copies of the given channel.
The following result provides a relation between the optimal rates for the asymptotic distillation and dilution in a more general setting, in which the errors do not necessarily approach zero in the asymptotic limit.

\begin{theorem}\label{thm:asymptotic}
Let $\E$ be an arbitrary input channel, and let $\T$ be some target reference channel for which $\idc \otimes \T(\psi)$ is pure for every pure state $\psi$. 
Let $d$ be any rate of distillation such that there exists a sequence $\{\Theta_n\}_n$ of free superchannels with
\begin{equation}\begin{aligned}
  1 - F\left(\Theta_n(\E^{\otimes n}), \T^{\otimes \floor{d n}}\right) \eqqcolon \delta_n.
\end{aligned}\end{equation}
Also, let $c$ be any rate of dilution such that there exists a sequence $\{\Theta_n\}_n$ of free superchannels with
\begin{equation}\begin{aligned}
  1 - F\left(\Theta_n(\T^{\otimes \ceil{c n}}), \E^{\otimes n}\right) \eqqcolon \epsilon_n.
\end{aligned}\end{equation}

Suppose that the following conditions are satisfied:
\begin{enumerate}[(i)]
\item It holds that $\liminf_{n \to \infty} \epsilon_n + \delta_n < 1$; 
\item The resource theory is closed under tensor products; that is, $\M_1, \M_2 \in \OO \Rightarrow \M_1 \otimes \M_2 \in \OO$.
\end{enumerate}
Then the following inequality holds
\begin{equation}\begin{aligned}
  d \cdot D^\infty_{\min,\OO}(\T) \leq
  c\cdot D^\infty_{\max,\OO}(\T),
\end{aligned}\end{equation}
where
\begin{equation}\begin{aligned}
  D^\infty_{\min,\OO}(\T) &\coloneqq \lim_{n\to\infty}\frac{1}{n} D_{\min}\left(\T^{\otimes n}\right),\\
  D^\infty_{\max,\OO}(\T) &\coloneqq \lim_{n\to\infty}\frac{1}{n} D_{\max}\left(\T^{\otimes n}\right).
\end{aligned}\end{equation}
\end{theorem}

Let us briefly discuss the assumptions of Theorem~\ref{thm:asymptotic}. 
Condition (i) simply means that the errors of distillation and dilution do not get too large at the same time; it is satisfied, for instance, when we take $\lim_{n\to\infty} \epsilon_n = 0$ and $\liminf_{n \to\infty} \delta_n < 1$ or vice versa. Condition (ii) is a basic property obeyed by virtually all theories encountered in practice.

Although the result of Theorem~\ref{thm:asymptotic} is appealing since it follows straightforwardly from our one-shot relation in Theorem~\ref{thm:second law general}, one can ask whether it is possible to derive asymptotic relations that have better dependence on quantities such as $D^\infty_{\max,\OO}(\T)$ and $D^\infty_{\min,\OO}(\T)$, as well as are free from the constraint on the target channel $\T$, that is, the constraint that $\idc\otimes\T(\psi)$ is pure for every pure state $\psi$. 
We discuss several variations on this idea in Appendix~\ref{app:alternative asymptotic}, where we present other bounds that are potentially tighter and do not impose any condition on $\T$ but come with an additional condition on the sequences $\{\epsilon_n\}_n$ and $\{\delta_n\}_n$ of achievable errors  or require the computation of more complicated regularized quantities. For our purposes, it will be sufficient to use the condition of Theorem~\ref{thm:asymptotic}, and indeed we will see that the reliance on the max-/min-relative entropies is not a problem in most theories of practical interest.

The result of Theorem~\ref{thm:asymptotic} implies in particular the strong converse bounds given in Corollary~\ref{cor:strong converse} below, for the rates of resource manipulation tasks. Before stating this result, let us recall the definitions of the basic asymptotic quantities involved.
We define the asymptotic distillable resource $d_\SS^\infty$ as the largest achievable rate at which copies of $\T$ can be extracted from a given channel, and, analogously, the asymptotic resource cost $c_\SS^\infty$ as the least rate at which copies of $\T$ are needed to produce a given channel. Imposing that the error in such transformations vanishes in the asymptotic limit, we have
\bal\label{eq:def_distcost}
d_\SS^\infty(\E,\T)&\!\coloneqq\!\sup\lsetr\! d \sbarr\lim_{n\to\infty}\sup_{\Theta_n\in\SS}F\left(\Theta_n(\E^{\otimes n}),\T^{\otimes\lfloor dn \rfloor}\right)\!=\!1\!\rsetr \\
c_\SS^\infty(\E,\T)&\!\coloneqq\!\inf\lsetr\! c \sbarr\lim_{n\to\infty}\sup_{\Theta_n\in\SS}F\left(\Theta_n(\T^{\otimes \lceil cn \rceil}),\E^{\otimes n}\right)\!=\!1\!\rsetr.
\eal
As counterparts to the above asymptotic quantities, we define also the corresponding strong converse quantities as~\cite{Hayashi2016quantum,Khatri2020principles}
\bal
{\tilde d}_\SS^\infty(\E,\T)&\!\coloneqq\!\inf\lsetr\! d \sbarr\lim_{n\to\infty}\sup_{\Theta_n\in\SS}F\left(\Theta_n(\E^{\otimes n}),\T^{\otimes\lfloor dn \rfloor}\right)\!=\!0\!\rsetr \\
{\tilde c}_\SS^\infty(\E,\T)&\!\coloneqq\!\sup\lsetr\! c \sbarr\lim_{n\to\infty}\sup_{\Theta_n\in\SS}F\left(\Theta_n(\T^{\otimes \lceil cn \rceil}),\E^{\otimes n}\right)\!=\!0\!\rsetr ,
\eal
with the interpretation of the strong converse distillable resource as the smallest rate at which the error is guaranteed to converge to one in the asymptotic limit, and that for the strong converse dilution cost as the largest rate at which the error is guaranteed to converge to one.
The following equivalent expressions for the strong converse quantities hold:
\bal
{\tilde d}_\SS^\infty(\E,\T)  &\!=\!\sup\lsetr\! d \sbarr\limsup_{n\to\infty}\sup_{\Theta_n\in\SS}F\left(\Theta_n(\E^{\otimes n}),\T^{\otimes\lfloor dn \rfloor}\right)\!>\!0\!\rsetr \\
{\tilde c}_\SS^\infty(\E,\T) &\!=\!\inf\lsetr\! c \sbarr\limsup_{n\to\infty}\sup_{\Theta_n\in\SS}F\left(\Theta_n(\T^{\otimes \lceil cn \rceil}),\E^{\otimes n}\right)\!>\!0\!\rsetr,
\eal
which can alternatively be understood as weakening the requirements imposed on the transformation error in Eq.~\eqref{eq:def_distcost}.
Here the interpretation is as follows: unlike the quantities in \eqref{eq:def_distcost}, the strong converse quantities no longer ask that the error converges to zero, but represent the best rates at which the error does not converge to one --- they therefore provide a general threshold on achievable asymptotic transformations even when perfect conversion is not required.
Note that $d_\SS^\infty(\E,\T)\leq {\tilde d}_\SS^\infty(\E,\T)$ and $c_\SS^\infty(\E,\T)\geq {\tilde c}_\SS^\infty(\E,\T)$ by definition. 

Theorem~\ref{thm:asymptotic} implies the following strong converse bounds. 

\begin{corollary}\label{cor:strong converse}
Let $\E$ be an arbitrary input channel, and let $\T$ be some target reference channel for which $\idc \otimes \T(\psi)$ is pure for every pure state $\psi$.
Suppose that the following conditions are satisfied:
\begin{enumerate}[(i)]
\item The resource theory is closed under tensor products; that is, $\M_1, \M_2 \in \OO \Rightarrow \M_1 \otimes \M_2 \in \OO$;
\item The target channel satisfies $D_{
\min,\OO}(\T)=D_{\max,\OO}(\T)$ and $D_{\min,\OO}(\T^{\otimes n}) = n D_{\min,\OO}(\T)$ for every $n$.
\end{enumerate}
Then the following inequalities hold:
\begin{equation}\begin{aligned}
 d_\SS^\infty(\E,\T)\leq{\tilde c}_\SS^\infty(\E,\T),\quad {\tilde d}_\SS^\infty(\E,\T)\leq c_\SS^\infty(\E,\T).
 \label{eq:str-conv-bounds-yc}
\end{aligned}\end{equation}
\end{corollary}

We remark here that the definitions of strong converse quantities and Corollary~\ref{cor:strong converse} imply the following fundamental inequality:
\bal
 d_\SS^\infty(\E,\T)\leq c_\SS^\infty(\E,\T).
\eal

Here, condition (i) is generally satisfied in physical theories, but condition (ii) imposes non-trivial requirements on the choice of a suitable target channel $\T$. In particular, in addition to the collapse of the max- and min-relative entropy measures to the same value, it requires that the min-relative entropy measure is additive, i.e.\ $D_{\min,\OO}(\T^{\otimes n}) = n D_{\min,\OO}(\T)$. Although this constitutes a seemingly restrictive condition, the existence of channels or states satisfying this requirement is actually a common property of important resource theories, obeyed, e.g., by quantum entanglement, quantum thermodynamics, magic, and quantum communication (cf.~\cite[Table~1]{Regula2021fundamental}). Intuitively, the condition can be thought of as characterizing how well the given reference $\T$ serves as an intermediary channel in distillation and dilution protocols.

We stress that all the results in this section automatically provide bounds for state theories as special cases; when one is interested in a resource theory of quantum states with free states $\FF$ and free operations $\OO$, the same results hold by replacing free superchannels $\SS$ with free operations $\OO$ and free channels $\OO$ with free states $\FF$.  
Furthermore, in the case of state transformations, the generalized quantum Stein's lemma~\cite{Brandao2010reversible,Brandao2010generalization} implies that the asymptotic distillable resource coincides with its strong converse rate when the set of free states and the reference state satisfy some mild assumptions --- see Appendix~\ref{app:strong converse} for details.
When the generalized quantum Stein's lemma holds, for the asymptotic distillation and dilution for an arbitrary state $\rho$ with respect to a reference state $\Phi$ under free operations $\OO$, we can improve the bound in Corollary~\ref{cor:strong converse} to
\bal
 {\tilde d}_\OO^\infty(\rho,\Phi)\leq{\tilde c}_\OO^\infty(\rho,\Phi).
\label{eq:two-sided strong converse}
\eal
This applies in particular to the resource theory of quantum entanglement, which solves an open problem posed in Ref.~\cite{Wilde2021second}. The implication of \eqref{eq:two-sided strong converse} is that the yield--cost relation still holds even when neither of the errors is required to vanish in the asymptotic limit, which gives a significant strengthening of the inequalities in \eqref{eq:str-conv-bounds-yc}.

We remark that another approach to asymptotic transformation rates of states and the consequent relations between asymptotic resource yield and cost in general resource theories was considered in Ref.~\cite{Kuroiwa2020general}. 

\section{Operational account for reference states}\label{sec:operational}

We saw that the relation between $D_{\min,\OO}$ and $D_{\max,\OO}$ plays a major role in establishing the yield--cost relation in general resource theories.  
Here, we give operational insights into the properties of these quantifiers and in particular the condition $D_{\min,\OO}(\T) = D_{\max,\OO}(\T)$ that we have encountered previously. 
By studying the interplay between these and other related measures, we will be able to quantify a number of smoothed resource monotones for target reference states in general resource theories.
In this section, we restrict our attention to theories of quantum states, and we use $\FF$ and $\OO$ to specify the set of free states and free operations, respectively. 
In what follows, we also assume that $\FF$ is a convex and closed set.

In addition to $D_{\min,\FF}$ and $D_{\max,\FF}$, we also consider another type of robustness measure known as the \emph{standard robustness}~\cite{Vidal1999robustness}, defined as
\bal
D_{s,\FF}(\rho)\coloneqq\inf\lset \log(1+s) \sbar \frac{\rho+s\tau}{1+s}\in\FF,\tau\in\FF\rset,
\label{eq:standard robustness def}
\eal

We also define two smooth robustness quantities as~\cite{Datta2009min,Brandao2010generalization}
\begin{align}
D_{s,\FF}^\epsilon(\rho)& \coloneqq \inf\lset D_{s,\FF}(\rho')\sbar F(\rho',\rho)\geq 1-\epsilon\rset,\\
D_{\max,\FF}^\epsilon(\rho)& \coloneqq  \inf\lset D_{\max,\FF}(\rho')\sbar F(\rho',\rho)\geq 1-\epsilon\rset,
\end{align}
and two types of hypothesis-testing relative entropy measures~\cite{Brandao2010generalization,Regula2018one-shot,Regula2020benchmarking}:
\bal
D_{H,\FF}^\epsilon(\rho)&\coloneqq \inf_{\sigma\in\FF}D_H^\epsilon(\rho\|\sigma),\\ D_{H,\aff(\FF)}^\epsilon(\rho)&\coloneqq \inf_{\sigma\in\aff(\FF)}D_H^\epsilon(\rho\|\sigma),
\eal
where
\begin{equation}
D_H^\epsilon(\rho\Vert\sigma)\coloneqq  \sup_{\substack{0\leq P \leq \id\\\Tr(P\rho)\geq 1-\epsilon}}\log \Tr(P\sigma)^{-1}
\label{eq:hypothesis-testing def}
\end{equation}
is the hypothesis-testing relative entropy \cite{Buscemi2010quantum,Brandao2011one-shot,Wang2012one-shot}.
We used $\aff(\FF)$ to denote the affine hull of $\FF$, which is the smallest affine subspace that contains $\FF$.
The min-relative entropy is obtained as a special case of the hypothesis-testing relative entropy measure as $D_{\min,\FF}(\rho)=D_{H,\FF}^{\epsilon=0}(\rho)$.
We use this correspondence to define the affine min-relative entropy measure~\cite{Regula2020benchmarking,Regula2020oneshot} as 
\bal
 D_{\min,\aff(\FF)}(\rho)& \coloneqq D_{H,\aff(\FF)}^{\epsilon=0}(\rho)\\
 & = \inf_{\sigma\in\aff(\FF)}D_{H}^{\epsilon=0}(\rho\|\sigma) .
\eal
Then, the following ordering holds for an arbitrary state $\rho$:  
\bal
 D_{\min,\aff(\FF)}(\rho)\leq D_{\min,\FF}(\rho)\leq D_{\max,\FF}(\rho)\leq D_{s,\FF}(\rho),
\eal
where the first inequality follows because $\FF\subseteq\aff(\FF)$, the second inequality  because $D_{\min}(\rho\|\sigma)\leq D_{\max}(\rho\|\sigma)$ for all states $\rho$ and $\sigma$, and the third inequality  from the definitions of the generalized and standard robustness measures.
Depending on the structure of $\FF$, some of these measures may exhibit drastic behavior. 
$D_{\min,\FF}$ and $D_{\max,\FF}$ take finite values for every state $\rho$ as long as $\FF$ contains at least one full-rank state, which is satisfied by most of the relevant theories. 
However, $D_{s,\FF}$ may diverge if $\FF$ is \emph{reduced dimensional}~\cite{Regula2020benchmarking}, meaning that $\FF$ has zero volume in the set of all states and ${\rm span}(\FF)$ is not the whole space of self-adjoint operators. 
Examples of reduced-dimensional theories include the theories of coherence and thermal non-equilibrium. 
On the other hand, for \emph{full-dimensional} theories, which are not reduced dimensional and include the theories of entanglement and magic as examples, 
$D_{s,\FF}$ remains finite but $D_{\min,\aff(\FF)}$ takes the value zero for every state $\rho$.
Therefore, $D_{s,\FF}$ is usually a relevant resource quantifier in full-dimensional theories while $D_{\min,\aff(\FF)}$ is relevant in reduced-dimensional theories~\cite{Regula2020benchmarking,Regula2020oneshot}. 

Although these measures have been introduced in a rather abstract manner, they play crucial roles in the quantitative characterization of distillation and dilution --- under certain assumptions on the target state $\Phi$, the value of the one-shot resource yield $d^\epsilon_{\OO}(\rho)$ under the maximal set of free operations is directly related to $D^\epsilon_{H,\FF}(\rho)$ or $D^\epsilon_{H,\aff(\FF)}(\rho)$, while the value of the resource cost $c^\epsilon_{\OO}(\rho)$ corresponds to the value of $D^\epsilon_{s,\FF}(\rho)$ or $D^\epsilon_{\max,\FF}(\rho)$~\cite{Liu2019one-shot,Regula2020benchmarking}.
A necessary requirement for a precise description of distillation or dilution to be possible was that, when evaluated on the reference state $\Phi$, the resource measures all collapse to the same value; specifically, in full-dimensional theories one needs that
\bal\label{eq:golden_req_full}
 D_{\min,\FF}(\Phi) = D_{s,\FF}(\Phi),
\eal
while in reduced-dimensional theories one instead requires
\bal\label{eq:golden_req_reduced}
 D_{\min,\aff(\FF)}(\Phi) = D_{\max,\FF}(\Phi).
\eal
In some theories, such as entanglement or coherence, the existence of such states is natural: the maximally entangled or coherent states always satisfy the requirement. It is rather remarkable that maximally resourceful states satisfy $D_{\min,\FF}(\Phi) = D_{\max,\FF}(\Phi)$ in \emph{any} convex resource theory~\cite{Regula2020benchmarking}, but this is still not sufficient to guarantee that Eqs.~\eqref{eq:golden_req_full}--\eqref{eq:golden_req_reduced} hold in general --- we expect this to be a resource-dependent property which needs to be verified explicitly in each specific setting.

Here, we find that Eqs.~\eqref{eq:golden_req_full}--\eqref{eq:golden_req_reduced} can have operational implications, which then help evaluate other resource measures introduced in this section.
In particular, we  provide an understanding of the conditions \eqref{eq:golden_req_full}--\eqref{eq:golden_req_reduced} in terms of a twirling-like free operation.

\begin{lemma}\label{lem:resource to twirling}
If a state $\Phi$ satisfies $D_{\min,\FF}(\Phi)=D_{s,\FF}(\Phi)$ or $D_{\min,\aff(\FF)}(\Phi)=D_{\max,\FF}(\Phi)$, there exists a free operation $\Lambda\in\OO_{\max}$ defined by an operator $0\leq P^\star \leq \id$ with $\Tr[P^\star\Phi]=1$ of the form 
\bal
 \Lambda(\cdot)=\Tr[P^\star \cdot ]\Phi + \Tr[(\id-P^\star) \cdot ]\sigma^\star,
 \label{eq:resource to twirling}
\eal
such that $\sigma^\star$ is a state orthogonal to $\Phi$, satisfying $\Tr[P\sigma^\star]=\Tr[\Phi \sigma^\star]=0$.
When $D_{\min,\FF}(\Phi)=D_{s,\FF}(\Phi)$, one can further take $\sigma^\star$ to be a free state. 
\end{lemma}

This free operation possesses a property similar to well-known group twirling operations such as the isotropic twirling $\int U \otimes U^* (\cdot) U^\dagger \otimes U^{* \dagger} \,\mathrm{d}U$ in entanglement theory~\cite{Horodecki1999reduction}, in that it maps states to the given reference state or its complement, while stabilizing a specific state $\Phi$.
Such twirling operations were found to be useful to evaluate several entanglement measures for states invariant under them~\cite{Vollbrecht2001entanglement}.
In particular, the existence of such an operation allows us to obtain exact expressions for smoothed resource measures, which \textit{a priori} require a non-trivial optimization over all states within an error $\epsilon$.

\begin{proposition}\label{pro:smooth measures}
If a state $\Phi$ satisfies $D_{\min,\FF}(\Phi)=D_{s,\FF}(\Phi)=:r$, then for every $\epsilon \in [0,1)$ it holds that
\bal
 D_{H,\FF}^\epsilon(\Phi) = r + \log\frac{1}{1-\epsilon}
 \eal
 and
 \bal
 D_{\max,\FF}^\epsilon(\Phi) = D_{s,\FF}^\epsilon(\Phi) = \max\left\{r - \log\frac{1}{1-\epsilon},0\right\}.
 \eal
Similarly, if a state $\Phi$ satisfies $D_{\min,\aff(\FF)}(\Phi)=D_{\max,\FF}(\Phi)=:r$, then 
\bal
 D_{H,\aff(\FF)}^\epsilon(\Phi)=D_{H,\FF}^\epsilon(\Phi) = r + \log\frac{1}{1-\epsilon}
 \eal
 and
 \bal
 D_{\max,\FF}^\epsilon(\Phi) = \max\left\{r - \log\frac{1}{1-\epsilon},0\right\}.
 \eal 
 \end{proposition}

A crucial property which makes this result possible is that optimal states in the optimization of the smoothed resource measures can always be taken of the form $\Phi_\kappa=\kappa\Phi+(1-\kappa)\sigma^\star$, as obtained through the operation in Lemma~\ref{lem:resource to twirling}. 
Such states share many useful properties with isotropic states of entanglement theory, and indeed several quantitative insights in the description of isotropic states in works such as Ref.~\cite{Vollbrecht2001entanglement} can be extended to general resource theories. In particular, the entropic resource quantifiers can be computed exactly for such states.
\begin{proposition}\label{pro:isotropic exact}
Suppose a state $\Phi$ satisfies $D_{\min,\FF}(\Phi)=D_{s,\FF}(\Phi)=:r$, and let $\sigma^\star$ be the state in \eqref{eq:resource to twirling}. 
Then,   
\bal
 D_{\min,\FF}(\Phi_\kappa)=\begin{cases}
 0 & 0\leq\kappa<1 \\
 r & \kappa=1,
 \end{cases}
\eal
and
\bal
 D_{\max,\FF}(\Phi_\kappa)&=D_{s,\FF}(\Phi_\kappa)\\
 &=\max\left\{r - \log\frac{1}{\kappa},0\right\}
\eal
for every $\Phi_\kappa=\kappa\Phi+(1-\kappa)\sigma^\star$ with $\kappa\in[0,1]$. 
Similarly, if a state $\Phi$ satisfies $D_{\min,\aff(\FF)}(\Phi)=D_{\max,\FF}(\Phi)=:r$, then 
\bal
 D_{\min,\aff(\FF)}(\Phi_\kappa)&=D_{\min,\FF}(\Phi_\kappa)\\
 &=\begin{cases}
 0 & 0<\kappa<1\\
\log\frac{1}{1-2^{-r}} &\kappa=0\\
 r& \kappa=1,
 \end{cases}
\eal
and
\bal
D_{\max,\FF}(\Phi_\kappa)&=\max\left\{r-\log\frac{1}{\kappa},\, \log\frac{1-\kappa}{1-2^{-r}} \right\}
\eal
for every $\Phi_\kappa=\kappa\Phi+(1-\kappa)\sigma^\star$ with $\kappa\in[0,1]$.
\end{proposition}
We present a more complete discussion of the quantitative properties of the isotropic-like states $\Phi_\kappa$ in Appendix~\ref{app:isotropic}, where we also show how the smoothed entropic measures can be computed for this class of states.

These exact evaluations of the resource measures allow us to employ the argument in Ref.~\cite{Wilde2021second} to obtain an alternative bound to that given in Theorem~\ref{thm:second law general}.

\begin{theorem} \label{thm:second law general state}
 If the chosen reference set obeys $D_{\min,\FF}(\Phi) = D_{s,\FF}(\Phi) \; \forall \Phi \in \TT$ or $D_{\min,\aff(\FF)}(\Phi) = D_{\max,\FF}(\Phi) \; \forall \Phi \in \TT$, then for every set $\OO\subseteq\OO_{\max}$ of free operations and all $\epsilon_1,\epsilon_2\geq 0$ satisfying $\epsilon_1+\epsilon_2<1$, the following inequality holds
 \bal
  d_{\OO}^{\epsilon_1}(\rho)\leq c_{\OO}^{\epsilon_2}(\rho)+\log\frac{1}{1-\epsilon'},
 \eal
 where $\epsilon'\coloneqq \left(\sqrt{\epsilon_1(1-\epsilon_2)}+\sqrt{\epsilon_2(1-\epsilon_1)}\right)^2$.
\end{theorem}

This improves the bound in Theorem~\ref{thm:second law general}, as it gives a tighter upper bound for all regions of $(\epsilon_1,\epsilon_2)$ with $\epsilon_1+\epsilon_2<1$, as we prove in Appendix~\ref{app:comparison bounds}. 
We also remark that this bound becomes essentially tight when $\OO=\OO_{\max}$, $\rho\in\TT$, and either $\epsilon_1$ or $\epsilon_2$ is 0. 
This can be shown by using $d_{\OO_{\max}}^{\epsilon_1}(\rho)=D_{H,\FF}^{\epsilon_1}(\rho)$ and $c_{\OO_{\max}}^{\epsilon_2}(\rho)=D_{s,\FF}^{\epsilon_2}(\rho)$ when $D_{\min,\FF}(\rho)=D_{s,\FF}(\rho)$, and $d_{\OO_{\max}}^{\epsilon_1}(\rho)=D_{H,\aff(\FF)}^{\epsilon_1}(\rho)$ and $c_{\OO_{\max}}^{\epsilon_2}(\rho)=D_{\max,\FF}^{\epsilon_2}(\rho)$ when $D_{\min,\aff(\FF)}(\rho)=D_{\max,\FF}(\rho)$ (up to some floor or ceiling due to a discrete structure of $\TT$)~\cite{Regula2020oneshot}, as well as explicitly evaluating these measures by Proposition~\ref{pro:smooth measures}.

So far we have employed the property of several resource measures evaluated for a reference state to show the existence of a twirling-like operation \eqref{eq:resource to twirling} and presented its applications. 
We now also investigate the opposite direction, asking whether the existence of a certain type of free operation provides insights into the relation between different resource measures. 
The following result addresses this question.

\begin{lemma}\label{lem:twirling to resource}
Let $\S(\Phi)$ be a set of channels stabilizing $\Phi$, defined as 
\bal
 \S(\Phi)\coloneqq \lset \Tr[P\cdot]\Phi + \tilde\Lambda(\cdot) \sbar 0\leq P \leq \id, \Tr[\Phi P]=1, \tilde\Lambda\in {\rm CP}\rset
 \label{eq:stabilizers}
\eal
where {\rm CP} is the set of completely positive maps. 
Then, if there exists a free operation $\Lambda\in\OO_{\max}\cap\S(\Phi)$, the equality
\bal
D_{\min,\FF}(\Phi)=D_{\max,\FF}(\Phi)
\eal
holds.
In addition, if $\Lambda$ is completely free, i.e., $\idc\otimes\Lambda\in\OO_{\max}$ where $\idc$ is the identity map on an arbitrary ancillary system, then
\bal
 D_{\min,\FF}(\Phi^{\otimes n})=D_{\max,\FF}(\Phi^{\otimes n}),\ \forall n.
\eal

Moreover, if $\tilde\Lambda$ can be taken as a free subchannel in \eqref{eq:stabilizers}, i.e., $\tilde\Lambda(\sigma)\propto\cone(\FF),\forall \sigma\in\FF$, then 
\bal
 D_{\min,\FF}(\Phi)=D_{s,\FF}(\Phi),
\eal
where $\cone(X)\coloneqq\lset\lambda x\sbar \lambda\geq 0,\ x\in X\rset$.
\end{lemma}

Together with Lemma~\ref{lem:resource to twirling}, we conclude a general operational correspondence between the collapse of resource measures and the existence of twirling-like free operations: we have that $D_{\min,\FF}(\Phi)=D_{s,\FF}(\Phi)$ if and only if there exists a map of the form \eqref{eq:resource to twirling} with a free state $\sigma^\star$.

However, it is not always easy to find a channel $\Lambda\in\OO_{\max}\cap\S(\Phi)$ to apply Lemma~\ref{lem:twirling to resource}. 
The following result helps to find such a channel by formally relating the group-theoretic twirling operations inspired by the original LOCC twirlings~\cite{Werner1989quantum,Horodecki1999reduction} to the form of \eqref{eq:stabilizers} necessary to apply Lemma~\ref{lem:twirling to resource}.  

\begin{proposition}\label{pro:group twirling}
Let $\Phi$ be a pure state.  
Suppose that there exists a convex set $\OO\subseteq\OO_{\max}$ of free operations and a finite or a compact Lie group $G$ with a unitary representation $\{U(g)\}_{g\in G}$ that satisfies $U(g)\cdot U(g)^\dagger\in\OO$ and $U(g)\ket{\Phi}=e^{i\phi_g}\ket{\Phi}\ \forall g\in G$ for some set of eigenvalues $\{e^{i\phi_g}\}_g$. 
If $\ket{\Phi}$ is the unique simultaneous eigenvector of all $U(g)$'s with eigenvalues $\{e^{i\phi_g}\}_g$, then $D_{\min,\FF}(\Phi)=D_{\max,\FF}(\Phi)$ holds.
Moreover, if $\OO$ is completely free, then $D_{\min,\FF}(\Phi^{\otimes n})=D_{\max,\FF}(\Phi^{\otimes n})$ holds for every positive integer~$n$.
\end{proposition}

This result allows us to find an appropriate free operation to ensure the condition required in Theorems~\ref{thm:second law general} and \ref{thm:second law general state}.
We indeed find that Lemma~\ref{lem:twirling to resource} and Proposition~\ref{pro:group twirling} are helpful to obtain new insights into specific resource theories, which we now demonstrate.


\section{Application: Theory of magic}\label{sec:magic}

As we emphasized before, our results immediately hold in general types of quantum resources of both states and channels, and we direct the interested reader, e.g.,\ to the recent Refs.~\cite{Liu2019one-shot,Regula2021fundamental,Fang2020no-go,Regula2020oneshot,Yuan2020one-shot} as well as to Appendix~\ref{app:reference} for a discussion of how such general methods can be applied in specific theories such as quantum communication. Here, we discuss a non-trivial example in which the quantification of resource measures for many important states has still not been understood --- the theory of magic (non-stabilizerness)~\cite{Veitch2014resource,Howard2017application}. Applying the results of our work, we show how they can be used to provide new quantitative insights and reveal broad and useful relations for this resource.

To realize scalable fault-tolerant quantum computation, it is essential to encode the whole computation in a higher-dimensional space using an error-correcting code.
Clifford gates particularly stand out as a subset of quantum gates that admits efficient fault-tolerant encoding on many error-correcting codes.
However, to form a universal gate set, we also need to implement a non-Clifford gate. 
It is usually accomplished by the gate teleportation technique~\cite{Gottesman1999demonstrating,Zhou2000methodology}, which simulates the action of a non-Clifford gate by combining a Clifford operation and a \emph{magic state}, which cannot be created solely by Clifford operations.  
Since magic states are hard to produce in a fault-tolerant manner, they are precious resources under this setting, motivating us to consider the quantification and manipulation of them using a resource-theoretic formalism~\cite{Veitch2014resource,Howard2017application}. 

In particular, the optimal performance of magic state distillation and dilution has been a central question in the field, as it is the most resource-demanding part to realize fault-tolerant universal quantum computation. 
Our results establish a relation between the optimal performance of these two, under the assumption that the reference resource states satisfy the aforementioned conditions on resource measures.  
Here, we investigate these conditions for some well-known states that can be good candidates for reference states in distillation and dilution protocols.
We find that our approach, particularly Lemma~\ref{lem:twirling to resource} and Proposition~\ref{pro:group twirling}, provides new operational insights into the evaluation of resource measures, which may be of independent interest.

\emph{Stabilizer states} are those that can be represented by a probabilistic mixture of eigenstates of Pauli operators. 
In the resource theory of magic, stabilizer states form the set $\FF_{\rm STAB}$ of free states.
Although it is often assumed that quantum computation is carried out in multi-qubit systems, higher-dimensional qu\emph{d}it systems also stand as potential candidates for a large-scale quantum computing architecture~\cite{Campbell2012magic-state,Howard2012qudit}.
Resource theories have been developed for both settings~\cite{Veitch2014resource,Howard2017application,Wang2020efficiently}, and we will consider several standard magic states defined in these scenarios.  

Let us start with multi-qubit systems. 
For single-qubit states, the farthest states from the set of stabilizer states are positioned at $\frac{1}{\sqrt{3}}(\pm 1, \pm 1, \pm 1)$ in Bloch coordinates, which are connected to each other by Clifford unitaries.
We call one of them the \emph{face state}, which is written as 
\bal
 \dm{F} = \frac{\id+(X+Y+Z)/\sqrt{3}}{2},
\eal
where $X,Y,Z$ are qubit Pauli operators. 
Then, Lemma~\ref{lem:twirling to resource} and Proposition~\ref{pro:group twirling} allow us to obtain the following relation:

\begin{proposition}\label{pro:face state}
The qubit face state satisfies
 \bal
 D_{\min,\STAB}(F^{\otimes n})=D_{\max,\STAB}(F^{\otimes n}).
 \label{eq:face state}
 \eal
for every integer $n\geq 1$.
\end{proposition}

Since the face state is a maximizer of $D_{\min,\FF_{\rm STAB}}$, called a golden state~\cite{Liu2019one-shot,Regula2020benchmarking}, the case of $n=1$ recovers the result in Ref.~\cite{Liu2019one-shot}. 
Also, noting that the stabilizer extent introduced in Ref.~\cite{Bravyi2019simulation} is identical to $D_{\max,\FF_{\rm STAB}}$~\cite{Regula2018convex}, Eq.~\eqref{eq:face state} can be shown by using the additivity of $D_{\min,\FF_{\rm STAB}}(F^{\otimes n})$ and $D_{\max,\FF_{\rm STAB}}(F^{\otimes n})$, as well as the equivalence of $D_{\min,\FF_{\rm STAB}}(F)$ and $D_{\max,\FF_{\rm STAB}}(F)$ shown in 
Ref.~\cite{Bravyi2019simulation}.
Our result provides a different operational approach to this relation based on a strategy with a broad potential applicability.

We also get an analogous result for a special class of magic states that includes the $T$ state and the Toffoli state, which are usually used for gate teleportation protocols.

\begin{proposition} \label{pro:Clifford magic}
 Let $V$ be a unitary in the third level of the Clifford hierarchy, and let $\ket{\phi}$ be a stabilizer state defined on a multi-qudit system with local dimension $d$, where $d=2$ or $d$ is an odd prime.  Then, the state $\ket{\psi}=V\ket{\phi}$ satisfies $D_{\min,\STAB}(\psi^{\otimes n})=D_{\max,\STAB}(\psi^{\otimes n})$ for every integer $n\geq 1$.
\end{proposition}

This provides an alternative proof of the result in Ref.~\cite{Bravyi2019simulation} and extends it to qudit systems, including the qudit generalization of the $T$ state~\cite{Howard2012qudit}. Note that $D_{\max,\STAB}(\psi^{\otimes n}) = n D_{\max,\STAB}(\psi)$ holds for every $\psi$ of up to three qubits~\cite{Bravyi2019simulation}, and indeed also for single-qubit mixed states~\cite{Seddon2021quantifying}.

Although $D_{\min,\FF_{\rm STAB}}$ and $D_{\max,\FF_{\rm STAB}}$ coincide for the above cases, one can actually show that for every single-qubit state, $D_{\max,\FF_{\rm STAB}}$ is strictly less than $D_{s,\FF_{\rm STAB}}$.  
The natural question is whether there exists a state for which the three quantities collapse to the same value, allowing us to use both Theorems~\ref{thm:second law general} and~\ref{thm:second law general state}.
No such state has been previously known in the theory of multi-qubit magic, which prevented previously known bounds for one-shot distillation and dilution from being tight~\cite{Liu2019one-shot,Regula2020benchmarking}. To show that a suitable choice does indeed exist, let us consider the three-qubit fiducial Hoggar state~\cite{Andersson2015states,Howard2017application,stacey2016geometric,Stacey2019invariant,Zhu2012quantum} defined as
\begin{equation}\begin{aligned}
  \ket{{\rm Hog}} = \frac{1}{\sqrt{6}} (1+i,\, 0,\, -1,\, 1,\, -i,\, -1,\, 0,\, 0)^T.
\end{aligned}\end{equation}
The Hoggar state serves as one of the fiducial states from which group covariant SIC-POVMs can be generated~\cite{Zhu2012quantum}.
We show that the three measures collapse for the Hoggar state.

\begin{proposition}\label{pro:Hoggar state}
The three-qubit Hoggar state satisfies
\bal
 D_{\min,\STAB}({\rm Hog})=D_{\max,\STAB}({\rm Hog})&=D_{s,\STAB}({\rm Hog}) \\&= \log\frac{12}{5},
\eal
and 
\bal
 D_{\min,\STAB}({\rm Hog}^{\otimes n})=D_{\max,\STAB}({\rm Hog}^{\otimes n})
\eal
for every integer $n\geq 1$.
\end{proposition}
We note that the value of $D_{s,\FF_{\rm STAB}}({\rm Hog})$ was reported in Ref.~\cite{Howard2017application}.

Let us now turn our attention to qutrit states.
In the qutrit magic theory, two classes of states were identified to have the maximum sum negativity of the discrete Wigner function~\cite{Veitch2014resource}. 
The first class is represented by the Strange state~\cite{Appleby2005symmmetric,Zhu2010SIC,Veitch2014resource,Wang2020efficiently}, defined as 
\bal
 \ket{S}\coloneqq \frac{1}{\sqrt{2}}(\ket{1}-\ket{2}).
\eal
We can use Lemma~\ref{lem:twirling to resource} to show the collapse of the three measures for this state.

\begin{proposition}\label{pro:strange state}
The qutrit Strange state satisfies
 \bal
 D_{\min,\STAB}(S)=D_{\max,\STAB}(S)=D_{s,\STAB}(S)=1
 \eal
 and 
 \bal
 D_{\min,\STAB}(S^{\otimes n})=D_{\max,\STAB}(S^{\otimes n})
 \eal
 for every integer $n\geq 1$.
\end{proposition}

Our results complement the findings in Ref.~\cite{Wang2020efficiently}, which considered $D_{\min,\FF'}$ and $D_{\max,\FF'}$ with respect to a larger set of free operators (not necessarily normalized quantum states) based on the negativity of the discrete Wigner function and found that $D_{\min,\FF'}(S)=D_{\max,\FF'}(S)=\log(5/3)$.

Another maximizer of the sum negativity is the Norrell state~\cite{Veitch2014resource} defined as 
\bal
\ket{N}\coloneqq \frac{1}{\sqrt{6}}(-\ket{0}+2\ket{1}-\ket{2}).
\eal

We can get a similar collapse for the measures, but to a different value. 
\begin{proposition}\label{pro:Norrell state}
The qutrit Norrell state satisfies
\bal
 D_{\min,\STAB}(N)=D_{\max,\STAB}(N)&=D_{s,\STAB}(N)\\&=\log\frac{3}{2}.
\eal
\end{proposition}

We note that $D_{\min,\STAB}(N)=D_{\max,\STAB}(N)=\log \frac{3}{2}$ was originally reported in Ref.~\cite{Wang2020efficiently}.

These results may make one wonder whether there is a general characterization for when the three measures take the same value. 
Although we still do not have a definitive answer for this question, we can make the following observation.
For a given state $\psi$, Lemma~\ref{lem:resource to twirling} implies that if $D_{\min,\STAB}(\psi)=D_{\max,\STAB}(\psi)$, we must have $D_{\max,\STAB}(\psi)<D_{s,\STAB}(\psi)$ unless there exists a free state acting on $\supp\!\left(\frac{\id-\psi}{d-1}\right)$.
This is equivalent to the condition that the \emph{weight resource measure}~\cite{Lewenstein1998separability,Regula2021fundamental}, defined for the set $\FF$ of free states  as 
\bal
 W_\FF (\rho)\coloneqq \sup\lset w \sbar \rho=w \sigma + (1-w)\tau,\ \sigma\in\FF \rset
\eal
satisfies
\begin{equation}
W_{\STAB}\!\left(\frac{\id-\psi}{d-1}\right)>0.    
\end{equation}
This condition can be explicitly verified to hold for the Strange state and Norrell state, for which we have seen the three measures to coincide. 
On the other hand, for the qutrit $T$ state~\cite{Howard2012qudit}, defined as 
\begin{equation}
\ket{T}\coloneqq\frac{1}{\sqrt{3}}(e^{2\pi i/9}\ket{0}+\ket{1}+e^{-2\pi i/9}\ket{2}),    
\end{equation}
one can check that $W_{\STAB}\left(\frac{\id-T}{2}\right)=0$. 
Combining it with Proposition~\ref{pro:Clifford magic}, we recover the fact that
\bal
 D_{\min,\FF_{\rm STAB}}(T)=D_{\max,\FF_{\rm STAB}}(T)<D_{s,\FF_{\rm STAB}}(T).
\eal

Interestingly, numerical investigations suggest that, for most qutrit states, the weight measure for the complement is equal to  zero, indicating that the non-zero gap between $D_{\max,\FF_{\rm STAB}}$ and $D_{s,\FF_{\rm STAB}}$ is a generic feature shared by the majority of qutrit states.  
We leave a thorough investigation on the relationship between the weight resource measure and the robustness measures for future work.

\section{Conclusions}

We established a quantitative relation between the one-shot distillable resource and resource cost for general quantum resource theories, including both state-based resources as well as dynamical resources of quantum channels.
We also showed the corresponding bounds in the asymptotic regime and recovered the familiar relation between the distillable resource and resource cost in the form of strong converse bounds. 
We investigated the conditions that appear in the yield--cost relation and related it to a class of free operations that have similar properties to twirling operations.
We employed such operations to obtain analytical expressions for several smoothed resource measures for general resource theories of states and tightened the yield--cost relation.
We then applied our operational technique to evaluate resource measures for several standard resource states in the resource theory of magic, recovering previous results with different techniques and presenting new evaluations of measures for some magic states.  

Outstanding questions include the extension of the results that were only shown for state theories in this work to channel theories, such as the two-sided strong converse bound \eqref{eq:two-sided strong converse} and the relation to the twirling-like operation. The difficulty in characterizing the manipulation of quantum channels, and in particular their asymptotic properties~\cite{Gour2019how}, makes such questions non-trivial to answer.
Another interesting direction is to use the operational techniques introduced here to shed light on settings other than magic theory. 
On the other hand, it will also be beneficial to gain a deeper understanding of the structure of magic theory, for which additional operational insights might be helpful. 

\begin{acknowledgments}
We thank an anonymous referee of Ref.~\cite{Wilde2021second} for  suggesting to improve a result in Ref.~\cite{Wilde2021second} by using a tighter triangle inequality of the purified distance from Ref.~\cite{Tomamichel2015quantum},  which eventually allowed us to arrive at the bound in Theorem~\ref{thm:second law general state}.
R.T.\ acknowledges the support from National Research Foundation (NRF) Singapore, under its NRFF Fellow programme (Award No. NRF-NRFF2016-02), the Singapore Ministry of Education Tier 1 Grant 2019-T1-002-015, and the Lee Kuan Yew Postdoctoral Fellowship at Nanyang Technological University Singapore. Any opinions, findings and conclusions or recommendations expressed in this material are those of the author(s) and do not reflect the views of National Research Foundation, Singapore.
B.R.\ acknowledges support of the Japan Society for the Promotion of Science (JSPS) \mbox{KAKENHI} Grant No.\ 21F21015, the JSPS Postdoctoral Fellowship for Research in Japan, and the Presidential Postdoctoral Fellowship from Nanyang Technological University, Singapore.
M.M.W.\ acknowledges support from the NSF under Grant No.~1907615.
\end{acknowledgments}


\appendix

\section{Proofs of Theorem~\ref{thm:second law general}, Lemma~\ref{lem:second law}, and Theorem~\ref{thm:asymptotic}}\label{app:second law general proof}

Let us first define the smoothed channel divergences~\cite{Cooney2015strong,Diaz2018using}:
\bal
 D_H^\epsilon(\E\|\M)\coloneqq \sup_\psi D_H^\epsilon(\idc\otimes\E(\psi)\|\idc\otimes\M(\psi))\\
 D_{\max}^\epsilon(\E\|\M)\coloneqq \inf_{F(\E',\E)\geq 1-\epsilon} D_{\max}(\E'\|\M),
\eal
which construct the following resource monotones~\cite{Diaz2018using,Liu2019resource,Liu2020operational}:
\bal
 D_{H,\OO}^\epsilon(\E)&\coloneqq \inf_{\M\in\OO} D_H^\epsilon(\E\|\M)\\
 D_{\max,\OO}^\epsilon(\E)&\coloneqq \inf_{\M\in\OO} D_{\max}^\epsilon(\E\|\M).
\eal

We begin by showing some useful lemmas. 
The first two are channel extensions of the relations shown in Refs.~\cite{Datta2013smooth,Anshu2019minimax} between the hypothesis-testing relative entropy and the max-relative entropy of quantum states. They have an interpretation in the resource theory of asymmetric distinguishability as a one-shot yield--cost relation for bits of asymmetric distinguishability \cite{Wang2019asymmetric_channel}. 

\begin{lemma}\label{lem:hypothesis max channel square}

For all $\epsilon_1,\epsilon_2\geq 0$ with $\epsilon_1+\sqrt{\epsilon_2}<1$, two arbitrary channels $\E,\M$ satisfy 
\bal
 D_{H}^{\epsilon_1}(\E\|\M)\leq D_{\max}^{\epsilon_2}(\E\|\M) +  \log\frac{1}{1-\epsilon_1-\sqrt{\epsilon_2}}.
\eal
\end{lemma}

\begin{proof}
Ref.~\cite{Datta2013smooth} showed that 
\bal
 D_H^{\epsilon_1}(\rho\|\sigma)\leq D_{\max}^{\epsilon_2}(\rho\|\sigma)+\log\frac{1}{1-\epsilon_1-\sqrt{\epsilon_2}}
 \label{eq:hypothesis-testing and max}
\eal
for all $\epsilon_1,\epsilon_2\geq 0$, $\epsilon_1+\sqrt{\epsilon_2}<1$ and all states $\rho,\sigma$.
Then, we get
\begin{align}\begin{aligned}
 &D_H^{\epsilon_1}(\E\|\M)\\&=\sup_\psi D_H^{\epsilon_1}(\idc\otimes\E(\psi)\|\idc\otimes\M(\psi))\\
 &\textleq{(1)} \sup_\psi D_{\max}^{\epsilon_2}(\idc\otimes\E(\psi)\|\idc\otimes\M(\psi)) + \log\frac{1}{1-\epsilon_1-\sqrt{\epsilon_2}}\\
 &=\sup_\psi \inf_{F(\phi,\idc\otimes\E(\psi))\geq 1-\epsilon_2} D_{\max}(\phi\|\idc\otimes\M(\psi))\\
 & \qquad +\log\frac{1}{1-\epsilon_1-\sqrt{\epsilon_2}}\\
 &\textleq{(2)}\sup_\psi \inf_{F(\E',\E)\geq 1-\epsilon_2} D_{\max}(\idc\otimes\E'(\psi)\|\idc\otimes\M(\psi))\\
 & \qquad +\log\frac{1}{1-\epsilon_1-\sqrt{\epsilon_2}}\\
 &\textleq{(3)}  \inf_{F(\E',\E)\geq 1-\epsilon_2} \sup_\psi D_{\max}(\idc\otimes\E'(\psi)\|\idc\otimes\M(\psi))\\
 & \qquad +\log\frac{1}{1-\epsilon_1-\sqrt{\epsilon_2}}\\
 &= D_{\max}^{\epsilon_2}(\E\|\M)+\log\frac{1}{1-\epsilon_1-\sqrt{\epsilon_2}},
\end{aligned}\end{align}
where we used (1): Eq.~\eqref{eq:hypothesis-testing and max},\ (2): the restriction of the optimization over $\phi$ to the form $\phi=\idc\otimes\E'(\psi)$ with $F(\E',\E)\geq 1-\epsilon_2$, which is justified by   $F(\idc\otimes\E'(\psi),\idc\otimes\E(\psi))\geq F(\E',\E)\geq 1-\epsilon_2$ for every $\psi$,\ (3): max-min inequality. 
\end{proof}

We note that alternative bounds with trace-distance and diamond-distance smoothing were shown in Refs.~\cite{Wang2019asymmetric_state,Wang2019asymmetric_channel}. 

We also have the following alternative bound. 

\begin{lemma}\label{lem:hypothesis max channel minimax}
For all $\epsilon_1,\epsilon_2\geq 0$ with $\epsilon_1+\epsilon_2<1$, two arbitrary channels $\E,\M$ satisfy
\bal
 D_{H}^{\epsilon_1}(\E\|\M)\leq D_{\max}^{\epsilon_2}(\E\|\M)+\log\frac{1}{ \left(\sqrt{1-\epsilon_2}-\sqrt{\epsilon_1}\right)^2}.
\eal
\end{lemma}

\begin{proof}

The proof of \cite[Theorem~4]{Anshu2019minimax} established that 
\bal
 D_H^{\epsilon_1}(\rho\|\sigma)\leq D_{\max}^{\epsilon_2}(\rho\|\sigma)+\log\frac{1}{\left(\sqrt{1-\epsilon_2}-\sqrt{\epsilon_1}\right)^{2}}
\eal
for all $\epsilon_1,\epsilon_2\geq 0$, $\epsilon_1+\epsilon_2<1$ and all states $\rho,\sigma$.
The statement for the channels follows by employing the same argument in the proof of Lemma~\ref{lem:hypothesis max channel square}.
\end{proof}

We also recall monotonicity properties under one-shot channel transformations involving smooth measures.

\begin{lemma}[Theorems 1 and 3 in Ref.~\cite{Regula2020oneshot}]\label{lem:converses}
If there exists a free superchannel $\Theta\in\SS_{\max}$ such that $F(\Theta(\E), \N) \geq 1-\epsilon$, then for every resource monotone $\mathfrak{R}_\OO$, 
\begin{equation}\begin{aligned}
  \mathfrak{R}_{\OO} (\E) \geq \mathfrak{R}_{\OO}^{\epsilon} (\N)
\end{aligned}\end{equation}
where 
\bal
 \mathfrak{R}_\OO^\epsilon(\E)\coloneqq \inf_{F(\E',\E)\geq 1-\epsilon}\mathfrak{R}_\OO(\E').
\eal
In particular, 
\begin{equation}\begin{aligned}
  D_{\max,\OO} (\E) \geq D_{\max,\OO}^{\epsilon} (\N).
\end{aligned}\end{equation}

Now let $\N$ be a channel such that $\idc \otimes \N(\psi)$ is pure for every pure state $\psi$. If there exists a free superchannel $\Theta$ such that $F(\Theta(\E), \N) \geq 1-\delta$, then
\begin{equation}\begin{aligned}
  D_{H,\OO}^\delta (\E) \geq D_{\min,\OO}(\N).
\end{aligned}\end{equation}
\end{lemma}

We are now in a position to prove Theorem~\ref{thm:second law general} and Lemma~\ref{lem:second law}.
\medskip

\noindent
\textit{Proof of Theorem~\ref{thm:second law general} and Lemma~\ref{lem:second law}.}
Theorem~\ref{thm:second law general} is obtained from Lemma~\ref{lem:second law} by maximizing the left-hand side over $\T_1\in\TT$ and minimizing the right-hand side over $\T_2\in\TT$ for fixed errors $\epsilon_1,\epsilon_2$. 
Thus, it suffices to show Lemma~\ref{lem:second law}.
Let $\T_1$ be an arbitrary channel for which there exists $\Theta_1 \in \SS$ such that $F(\Theta_1(\E),\T_1)\geq 1-\epsilon_1$, and let $\T_2$ be an arbitrary channel for which there exists $\Theta_2 \in \SS$ such that $F(\Theta_2(\T_2),\E)\geq 1-\epsilon_2$. Then,
\begin{equation}\begin{aligned}
  D_{\min,\OO}(\T_1)&\leq D_{H,\OO}^{\epsilon_1}(\E)\\
  &\leq D_{\max,\OO}^{\epsilon_2}(\E) + \log f(\epsilon_1, \epsilon_2)\\
  &\leq D_{\max,\OO}(\T_2) + \log f(\epsilon_1, \epsilon_2),
\end{aligned}\end{equation}
where in the first and the third inequalities we used Lemma~\ref{lem:converses}, and in the second inequality we used Lemmas~\ref{lem:hypothesis max channel square} and \ref{lem:hypothesis max channel minimax}.
\qed
\medskip

Theorem~\ref{thm:asymptotic} can be shown similarly.

\medskip

\noindent
\textit{Proof of Theorem~\ref{thm:asymptotic}.}
Let us first observe that the regularized quantities $D^\infty_{\min,\OO}$ and $D^\infty_{\max,\OO}$ are both well defined. This follows since $D_{\max,\OO}$ and $D_{\min,\OO}$ are subadditive under assumption (ii) --- as can be seen from their definitions using the additivity of $D_{\max}$ and $D_{\min}$ on tensor-product arguments --- and hence the limit in the definition of the regularized quantities exists by Fekete's lemma.
Now, by assumption (i), there exists a subsequence $\{n_k\}_k$ of indices such that $\lim_{k\to\infty} \epsilon_{n_k} + \delta_{n_k} < 1$.
Then
\begin{equation}\begin{aligned}
  & \!\!\!\!d \, D^\infty_{\min,\OO}(\T) \\
  &\texteq{(1)} \lim_{k\to\infty} \frac{\floor{d n_k}}{n_k} \frac{1}{\floor{d n_k}} D_{\min,\OO}(\T^{\otimes \floor{d n_{k}}})\\
  &\textleq{(2)} \lim_{k\to\infty} \frac{1}{n_k} D_{H,\OO}^{\delta_{n_{k}}}(\E^{\otimes {n_{k}}})\\
  &\textleq{(3)} \lim_{k \to \infty} \frac{1}{n_k} \left( D_{\max,\OO}^{\epsilon_{n_{k}}}(\E^{\otimes {n_{k}}}) + \log f(\epsilon_{n_{k}}, \delta_{n_{k}}) \right)\\
  &\textleq{(4)} \lim_{k \to \infty} \frac{1}{n_k} \left( D_{\max,\OO}(\T^{\otimes \ceil{c{n_{k}}} }) + \log f(\epsilon_{n_{k}}, \delta_{n_{k}}) \right)\\
    &\texteq{(5)}  \lim_{k\to\infty} \frac{\ceil{c n_k}}{n_k} \frac{1}{\ceil{c n_k}} D_{\max,\OO}(\T^{\otimes \ceil{c{n_{k}}} })\\
    &\texteq{(6)} c \, D^\infty_{\max,\OO}(\T)
\end{aligned}\end{equation}
where we used: (1) the fact that $\lim_{k\to\infty} \frac{\floor{d n_k}}{n_k} = d$ and that
\begin{equation}\begin{aligned}
  \lim_{k\to\infty} \frac{1}{\floor{d n_k}} D_{\min,\OO}(\T^{\otimes \floor{d n_{k}}}) = D^\infty_{\min,\OO}(\T)
\end{aligned}\end{equation}
by definition; (2) Lemma~\ref{lem:converses}; (3) Lemmas~\ref{lem:hypothesis max channel square}--\ref{lem:hypothesis max channel minimax} and that $\lim_{k\to\infty} \epsilon_{n_k} + \delta_{n_k} < 1$; (4) Lemma~\ref{lem:converses}; (5) $\lim_{k\to\infty}\frac{1}{n_k}f(\epsilon_{n_k},\delta_{n_k})=0$; (6) the same argument as in (1).

\qed
\medskip

We remark that the subadditivity of $D_{\max,\OO}$ under assumption (ii) implies that $D^\infty_{\max,\OO}(\T) \leq D_{\max,\OO}(\T)$, giving also the general bound
\begin{equation}\begin{aligned}
 d\,D^\infty_{\min,\OO}(\T)\leq c\,D_{\max,\OO}(\T).
\end{aligned}\end{equation}

\section{Applicability to specific settings}\label{app:reference}
We briefly review several examples of physical settings, as well as a reference channel $\T$, that satisfy the conditions of Theorem~\ref{thm:second law general}. 
Examples discussed here are either state preparation channels with one-dimensional input or unitaries with input and output of the same dimension, ensuring that $\idc\otimes\T(\psi)$ is pure for every pure state $\psi$.
They further satisfy $D_{\min,\OO}(\T)=D_{\max,\OO}(\T)$ and thus serve as appropriate references for distillation and dilution protocols that meet the conditions in Theorem~\ref{thm:second law general}.

Let us first consider state theories, where $\OO$ is a set of channels preparing free states (denoted by $\FF$ in Section~\ref{sec:theory of states}). 
The theory of bipartite entanglement~\cite{Horodecki2009quantum}, in which separable states serve as free states, takes a maximally entangled state $\ket{\Phi^+_m} = \sum_{i=1}^{m} m^{-1/2} \ket{ii}$ as a reference state and obeys $D_{\min,\OO}(\Phi^+_m)=D_{\max,\OO}(\Phi^+_m)=\log m$.
In the case of the theory of coherence~\cite{Streltsov2017quantum}, where the free states are the diagonal states with respect to a given preferred basis, a maximally coherent state, i.e., a uniform superposition of the basis states $\ket{\phi^+_m} = \sum_{i=1}^{m} m^{-1/2} \ket{i}$, meets the condition $D_{\min,\OO}(\phi^+_m)=D_{\max,\OO}(\phi^+_m)=\log m$.
The theory of thermal non-equilibrium~\cite{Brandao2013resource} is defined by fixed temperature and Hamiltonian. 
The ``work bit'' represented by the eigenstate of the Hamiltonian with the highest energy is considered as a standard unit resource, and it also satisfies the conditions of Theorem~\ref{thm:second law general}.
In fact, the above observation can be generalized to an arbitrary state theory; every theory with an arbitrary convex set $\FF$ of free states is equipped with a \emph{golden state} $\Phi_{\rm gold}$~\cite{Liu2019one-shot} such that $D_{\min,\FF}(\Phi_{\rm gold})=D_{\max,\FF}(\Phi_{\rm gold})$~\cite{Regula2020benchmarking}.
See also Section~\ref{sec:magic} in the main text for further discussion on the theory of magic states. 

Many important dynamical resource theories are relevant to the setting of quantum communication.  
A central purpose of quantum communication is to transmit a quantum state from one party to another.  
In such a scenario, the identity channel connecting the two parties is considered as the most useful channel.  
It is then natural to take the $m$-dimensional identity channel $\idc_m$ as a reference in distillation and dilution protocols, and indeed, the identity channel satisfies the conditions in Theorem~\ref{thm:second law general} under several coding schemes that are specified by different choices of $\OO$~\cite{Regula2021fundamental,Regula2020oneshot}. For instance, in the theory of no-signalling--assisted communication, we get $D_{\min,\OO}(\idc_m)=D_{\max,\OO}(\idc_m)=2 \log m$, while in quantum communication assisted by LOCC, separable, or positive partial transpose codes we get $D_{\min,\OO}(\idc_m)=D_{\max,\OO}(\idc_m)=\log m$.
A related setting is a framework that quantifies how much quantum memory a given channel can preserve~\cite{Rosset2018resource}. 
This theory takes the set of entanglement-breaking channels as free channels, and the identity channel again serves as an appropriate reference channel~\cite{Yuan2020universal}. 
In relation to Section~\ref{sec:magic}, dynamical resource theories of magic~\cite{Seddon2019quantifying,Wang2019quantifying} admit several unitary gates as reference channels satisfying the conditions in Theorem~\ref{thm:second law general}~\cite{Regula2021fundamental}.

\section{Alternative asymptotic bounds}\label{app:alternative asymptotic}

Besides Theorem~\ref{thm:asymptotic}, we can also show alternative asymptotic bounds, which are tighter but have additional conditions on the achievable errors. 
Another advantage of our alternative bounds is that they do not require the reference channel $\T$ to be pure, i.e., $\idc\otimes\T(\psi)$ does not need to be pure for every pure state $\psi$.  

Let us begin by presenting some useful lemmas.
We first introduce the smoothed hypothesis-testing relative entropy measure as
\bal
 D_{H,\OO}^{\epsilon,\delta}(\E)\coloneqq \inf_{F(\E',\E)\geq 1-\delta} D_{H,\OO}^{\epsilon}(\E').
 \label{eq:smoothed hypothesis def}
\eal
Then, we can relate this smoothed measure and the standard hypothesis-testing measure as follows.

\begin{lemma}\label{lem:smoothed hypothesis}
For an arbitrary channel $\E$ and all $\epsilon,\delta\geq 0$ satisfying $\epsilon+\sqrt{\delta}\leq 1$,  
 \bal
 D_{H,\OO}^{\epsilon}(\E)\leq D_{H,\OO}^{\epsilon+\sqrt{\delta},\delta}(\E).
 \eal
\end{lemma}
\begin{proof}
The hypothesis-testing measure can explicitly be written by
 \begin{equation}\begin{aligned}
  D_{H,\OO}^{\epsilon}(\E)&=-\log \max_{\M\in\OO}\min_\psi\min_{\substack{0\leq P \leq \id\\\Tr[P\idc\otimes\E(\psi)]\geq 1-\epsilon}}\Tr[P\,\idc\otimes\M(\psi)],
 \end{aligned}\end{equation}
where we define $\log 0\coloneqq-\infty$ here and throughout the manuscript.

 Note that for an arbitrary channel $\E'$ satisfying $F(\E',\E)\geq 1-\delta$, an arbitrary positive semidefinite operator $P$ satisfying $0\leq P \leq \id$, and an arbitrary state $\psi$, 
 \begin{equation}\begin{aligned}
 &\left|\Tr[P\,\idc\otimes\E(\psi)]-\Tr[P\,\idc\otimes\E'(\psi)]\right|\\
 &\leq \frac{1}{2}\|\idc\otimes\E(\psi)-\idc\otimes\E'(\psi)\|_1\\
 &\leq \sqrt{1-F(\E',\E)},
 \end{aligned}\end{equation}
where the first inequality is because of the optimization form of the trace distance $\frac{1}{2}\|\rho-\sigma\|_1=\max_{0\leq P \leq \id}\Tr[P(\rho-\sigma)]$ satisfied for arbitrary two states $\rho$ and $\sigma$~\cite{Nielsen2011quantum}, and the second inequality is because 
\begin{align}
 & \frac{1}{2}\|\idc\otimes\E(\psi)-\idc\otimes\E'(\psi)\|_1\notag \\
 &\leq \sqrt{1-F(\idc\otimes\E(\psi),\idc\otimes\E'(\psi))}\notag \\
 &\leq \sqrt{1-F(\E,\E')},
\end{align}
where we used the relation between trace distance and fidelity~\cite{Fuchs1999cryptographic} in the first inequality and the definition of channel fidelity~\eqref{eq:fidelity def} in the second inequality.
This implies that as long as $F(\E',\E)\geq 1-\delta$, every $P$ with $0\leq P \leq \id$ and state $\psi$ satisfying $\Tr[P\,\idc\otimes\E(\psi)]\geq 1-\epsilon$ also satisfy $\Tr[P\,\idc\otimes\E'(\psi)]\geq 1-\epsilon-\sqrt{\delta}$.
This gives 
 
\begin{equation}\begin{aligned}
  &\max_{\M\in\OO}\min_\psi\!\min_{\substack{0\leq P \leq \id\\\Tr[P\idc\otimes\E(\psi)]\geq 1-\epsilon}}\Tr[P\,\idc\otimes\M(\psi)]\\
  &\geq \max_{\M\in\OO}\min_\psi\max_{F(\E',\E)\geq 1-\delta}\!\min_{\substack{0\leq P \leq \id\\\Tr[P\idc\otimes\E'(\psi)]\geq 1-\epsilon-\sqrt{\delta}}}\Tr[P\,\idc\otimes\M(\psi)]\\
  &\geq\max_{F(\E',\E)\geq 1-\delta}\max_{\M\in\OO}\min_\psi\!\min_{\substack{0\leq P \leq \id\\\Tr[P\idc\otimes\E'(\psi)]\geq 1-\epsilon-\sqrt{\delta}}}\Tr[P\,\idc\otimes\M(\psi)]
 \end{aligned}\end{equation}
 where in the last line we used the max-min inequality. 
Taking $-\log$ on both sides concludes the proof. 
\end{proof}

The following result, which is a variant of Theorem~1 in Ref.~\cite{Regula2020oneshot}, is also useful.

\begin{lemma}\label{lem:converse with smoothing}
If there exists a free superchannel $\Theta\in\SS_{\max}$ such that $F(\Theta(\E), \N) \geq 1-\epsilon$, then for an arbitrary resource monotone $\mfR_\OO$,  
\bal
\mfR_\OO^\delta(\E)\geq \mfR_{\OO}^{(\sqrt{\delta}+\sqrt{\epsilon})^2}(\N)
\eal
for every $0\leq\delta\leq 1$, where
\bal
 \mfR_\OO^\delta(\E)\coloneqq\inf_{F(\E',\E)\geq 1-\delta} \mfR_\OO(\E').
\eal
\end{lemma}
\begin{proof}
Let $P(\E,\N)\coloneqq \sqrt{1-F(\E,\N)}$ be the sine distance of quantum channels (also known as purified distance), which satisfies the triangle inequality 
\bal
 P(\E,\T) \leq P(\E,\N) + P(\N,\T) 
 \label{eq:triangle purified distance channel}
\eal
for all channels $\E$, $\N$, $\T$~\cite{Gilchrist2005distance}.
Noting that $\mathfrak{R}_\OO^\delta$ is a resource monotone for all $\delta$, we can use Lemma~\ref{lem:converses} to get 

\bal
 \mathfrak{R}_\OO^\delta(\E)&\geq\inf_{P(\N',\N)\leq \sqrt{\epsilon}} \mathfrak{R}_\OO^\delta(\N')\\
 &=\inf_{P(\N',\N)\leq \sqrt{\epsilon}} \inf_{P(\N'',\N')\leq \sqrt{\delta}}\mathfrak{R}_\OO(\N'').
\eal 
Using the triangle inequality of the purified distance, we get $P(\N'',\N)\leq \sqrt{\epsilon} +\sqrt{\delta}$ for all channels $\N$, $\N'$, $\N''$ that satisfy $P(\N',\N)\leq \sqrt{\epsilon}$ and $P(\N'',\N')\leq \sqrt{\delta}$.
Noting also that $P(\N'',\N)\leq \sqrt{\epsilon}+\sqrt{\delta}$ is equivalent to $F(\N'',\N)\geq 1- (\sqrt{\epsilon}+\sqrt{\delta})^2$, we get

\begin{equation}\begin{aligned}
 \mathfrak{R}_\OO^\delta(\E)&\geq  \inf_{F(\N'',\N)\geq 1-(\sqrt{\epsilon}+\sqrt{\delta})^2}\mathfrak{R}_\OO(\N'')\\
 &=\mathfrak{R}_\OO^{(\sqrt{\delta}+\sqrt{\epsilon})^2}(\N).
\end{aligned}\end{equation}
\end{proof}

\begin{remark}
We can improve the triangle inequality \eqref{eq:triangle purified distance channel}. 
Define the following ``best-case'' fidelity 
\begin{align}
 F_{\max}(\E,\N)&\coloneqq\sup_{\rho_{RA}} F(\idc_R\otimes\E(\rho_{RA}),\idc_R\otimes\N(\rho_{RA}))
 \notag \\
 &= \max_{\rho} F(\E(\rho),\N(\rho)),
 \label{eq:max-out-fidelity}
\end{align}
where the supremum in the first line is taken over every possible ancillary system $R$, and the second equality is due to the data-processing inequality, ensuring that tracing out the ancillary system does not decrease the fidelity. 
The last form particularly allows for an efficient computation via the following semidefinite program, which follows from \eqref{eq:max-out-fidelity} and the known semi-definite program for fidelity~\cite{watrous_2013}:
\begin{multline}
    \sqrt{F_{\max}(\E,\N)} =\\ \max_{\substack{\rho\geq 0, \\X\in\text{L}(\H)}} \lset  \Re[\Tr[X]] \sbar \begin{bmatrix}
    \E(\rho) & X^\dag\\
    X & \N(\rho)
    \end{bmatrix} \geq 0 ,\, \Tr[\rho]=1\rset,
\end{multline}
where $\text{L}(\H)$ denotes the set of linear operators acting on the output Hilbert space $\H$ for the channels $\E$ and $\N$.
Then, we have the following inequality. 
\begin{lemma} \label{lem:triangle purified distance improved}
Every set of channels $\E$, $\N$, and $\T$ with $P(\E,\N)^2+P(\N,\T)^2\leq 1$ satisfies 
\bal
 P(\E,\T)\leq P(\E,\N)\sqrt{F_{\max}(\N,\T)}+P(\N,\T)\sqrt{F_{\max}(\E,\N)}.
 \label{eq:triangle purified channel improved}
\eal
\end{lemma}
\begin{proof}
Recall that the purified distance $P(\rho,\sigma)\coloneqq \sqrt{1-F(\rho,\sigma)}$ defined for arbitrary two states $\rho$ and $\sigma$ satisfies the following triangle inequality~\cite{Tomamichel2015quantum} 
\bal
 P(\rho,\tau)\leq P(\rho,\sigma)\sqrt{F(\sigma,\tau)}+P(\sigma,\tau)\sqrt{F(\rho,\sigma)}
\eal
for every set of three states $\rho,\sigma,\tau$ such that $P(\rho,\sigma)^2+P(\sigma,\tau)^2\leq 1$.
We can employ this to obtain 
\begin{equation}\begin{aligned}
 P(\E,\T)&=\max_\psi P(\idc\otimes\E(\psi),\idc\otimes\T(\psi))\\
 &\leq \max_\psi P(\idc\otimes\E(\psi),\idc\otimes\N(\psi))\sqrt{F_{\max}(\N,\T)}\\
 &\quad+\max_\psi P(\idc\otimes\N(\psi),\idc\otimes\T(\psi))\sqrt{F_{\max}(\E,\N)}\\
 &= P(\E,\N)\sqrt{F_{\max}(\N,\T)}+P(\N,\T)\sqrt{F_{\max}(\E,\N)}.
\end{aligned}\end{equation}
\end{proof}

Eq.~\eqref{eq:triangle purified channel improved} is tighter than \eqref{eq:triangle purified distance channel} in general because $F_{\max}(\E,\N)\leq 1$ for every $\E$ and $\N$. 
Although \eqref{eq:triangle purified distance channel} and Lemma~\ref{lem:converse with smoothing} suffice to establish the forthcoming results, Lemma~\ref{lem:triangle purified distance improved} may find use in other different settings. 
\end{remark}

We are ready to show the first alternative asymptotic bound, which is tighter than Theorem~\ref{thm:asymptotic} and does not assume anything on $\T$, while having less flexibility in the achievable errors.

\begin{proposition}
 Let $\E$ be an arbitrary input channel and let $\T$ be some target reference channel. 
 Let $d$ be any rate of distillation such that there exists a sequence $\{\Theta_n\}_n$ of free superchannels with
\begin{equation}\begin{aligned}
  1 - F\left(\Theta_n(\E^{\otimes n}), \T^{\otimes \floor{d n}}\right) \eqqcolon \delta_n.
\end{aligned}\end{equation}
Also, let $c$ be any rate of dilution such that there exists a sequence $\{\Theta_n\}_n$ of free superchannels with
\begin{equation}\begin{aligned}
  1 - F\left(\Theta_n(\T^{\otimes \ceil{c n}}), \E^{\otimes n}\right) \eqqcolon \epsilon_n.
\end{aligned}\end{equation}
Suppose that 
\bal
\lim_{n\to\infty} \delta_{n} = 0,\quad\liminf_{n\to\infty} \epsilon_{n} < 1. 
\eal
Then, the following inequality holds
\begin{equation}\begin{aligned}
   d\cdot\tilde D^\infty_{H,\OO}(\T)\leq c\cdot\tilde D^\infty_{\max,\OO}(\T),
\end{aligned}\end{equation}
where 
\bal
 \tilde D_{H,\OO}^\infty(\T)&\coloneqq \lim_{\epsilon\to 0}\liminf_{n\to \infty} \frac{1}{n} D_{H,\OO}^\epsilon(\T^{\otimes n})\\
 \tilde D_{\max,\OO}^\infty(\E)&\coloneqq \lim_{\epsilon\to 0}\limsup_{n\to\infty}\frac{1}{n} D_{\max,\OO}^\epsilon(\E^{\otimes n}).
\eal

\end{proposition}

\begin{proof}

Let $\{n_k\}_k$ be a subsequence of indices such that $\epsilon_\infty\coloneqq\lim_{k\to\infty} \epsilon_{n_k} < 1$.
Then, for every $\xi$ with $0<\xi<1-\epsilon_\infty$, we get
\begin{equation}\begin{aligned}
  & \!\!\!\!c \, \tilde D^\infty_{\max,\OO}(\T) \\
  &\geq  \lim_{\eta\to 0}\liminf_{k\to\infty} \frac{\ceil{c n_k}}{n_k} \frac{1}{\ceil{c n_k}} D_{\max,\OO}^\eta(\T^{\otimes \ceil{c{n_{k}}} })\\
  &\geq \lim_{\eta\to 0}\liminf_{k\to\infty} \frac{1}{n_k} D_{\max,\OO}^{(\sqrt{\eta}+\sqrt{\epsilon_{n_{k}}})^2}(\E^{\otimes {n_{k}}})\\
  &\geq \lim_{\eta\to 0}\liminf_{k\to\infty} \frac{1}{n_k} \left[D_{H,\OO}^{\xi}(\E^{\otimes {n_{k}}})-\log f\left(\xi,(\sqrt{\eta}+\sqrt{\epsilon_{n_{k}}})^2\right)\right]\\
  &\geq \lim_{\eta\to 0}\liminf_{k\to\infty}\\
  &\qquad\frac{1}{n_k}\left[ D_{H,\OO}^{\xi,\delta_{n_k}}(\T^{\otimes \floor{d n_{k}}})-\log f\left(\xi,(\sqrt{\eta}+\sqrt{\epsilon_{n_{k}}})^2\right)\right]\\
  &\geq \lim_{\eta\to 0}\liminf_{k\to\infty}\\
  &\qquad\frac{1}{n_k} \left[D_{H,\OO}^{\xi-\sqrt{\delta_{n_k}}}(\T^{\otimes \floor{d n_{k}}})-\log f\left(\xi,(\sqrt{\eta}+\sqrt{\epsilon_{n_{k}}})^2\right)\right].
  \end{aligned}\end{equation}
The second inequality follows from Lemma~\ref{lem:converse with smoothing}, the third inequality from Lemmas~\ref{lem:hypothesis max channel square}--\ref{lem:hypothesis max channel minimax} and the fact that $\xi+(\sqrt{\eta}+\sqrt{\epsilon_{n_k}})^2<1$ holds for $\eta$ sufficiently close to $0$ and for sufficiently large $k$, the fourth inequality from Lemma~\ref{lem:converses} noting that $D_{H,\OO}^\xi$ is a resource monotone for a fixed $\xi$, and the fifth inequality from Lemma~\ref{lem:smoothed hypothesis} together with that $\xi-\sqrt{\delta_{n_k}}\in[0,1]$ for sufficiently large $k$.

The second term of the last line, $\log f\left(\xi,(\sqrt{\eta}+\sqrt{\epsilon_{n_{k}}})^2\right)/n_k$, vanishes at the limit of $k\to\infty$, which also removes the $\eta$ dependence. 
Also, for every $\delta'>0$, we have $\sqrt{\delta_k}<\delta'$ for sufficiently large $k$. Noting that $D_{H,\OO}^\epsilon$ is nondecreasing with respect to $\epsilon$, we can bound the last line as
\begin{equation}\begin{aligned}
&\geq \lim_{\delta'\to 0}\liminf_{k\to\infty} \frac{1}{n_k} D_{H,\OO}^{\xi-\delta'}(\T^{\otimes \floor{d n_{k}}})\\
&= \lim_{\delta'\to 0}\liminf_{k\to\infty} \frac{\floor{d n_k}}{n_k} \frac{1}{\floor{d n_k}}D_{H,\OO}^{\xi-\delta'}(\T^{\otimes \floor{d n_{k}}}).
\end{aligned}\end{equation}

Since this holds for every $\xi$ with $0<\xi<1-\epsilon_\infty$, we can further take $\lim_{\xi\to 0}$ and use 
\bal
\lim_{\xi\to 0}\lim_{\delta'\to 0}\liminf_{k\to\infty} \frac{1}{\floor{d n_k}}D_{H,\OO}^{\xi-\delta'}(\T^{\otimes \floor{d n_{k}}})\geq \tilde D_{H,\OO}^\infty(\T)
\eal
to get 
\bal
c\,\tilde D_{\max,\OO}^\infty(\T)\geq d\, \tilde D_{H,\OO}^\infty(\T).
\eal

\end{proof}

In particular, this gives a strong converse inequality 
\bal
 \tilde c_\SS^\infty(\E,\T)\,\tilde D_{\max,\OO}^\infty(\T)\geq d_\SS^\infty(\E,\T)\,\tilde D_{H,\OO}^\infty(\T).
\eal
In the case of state transformations with mild assumptions (c.f., Appendix~\ref{app:strong converse}), the asymptotic equipartition property $\tilde D_{\max,\FF}^\infty(\Phi)=\tilde D_{H,\FF}^\infty(\Phi)$ holds~\cite{Brandao2010generalization}, and consequently we get 
\bal
 d_\OO^\infty(\E,\Phi)\leq \tilde c_\OO^\infty(\E,\Phi).
\eal

In addition, if $\Phi$ satisfies $D_{\min,\FF}(\Phi^{\otimes m})=D_{s,\FF}(\Phi^{\otimes m})=m D_{\min,\FF}(\Phi)\ \forall m$, this leads to the double-sided strong converse inequality 
\bal 
\tilde d_\OO^\infty(\rho,\Phi)\leq\tilde c_\OO^\infty(\rho,\Phi)
\eal
as we discuss in Appendix~\ref{app:strong converse}.

Finally, if we further impose a stronger condition to the achievable errors, we can obtain an even tighter bound.

\begin{proposition}\label{pro:strong converse smooth}
Let us consider the setting in which we instead have the condition $\lim_{k\to\infty} \delta_{n_k} = \lim_{k\to\infty} \epsilon_{n_k} = 0$ for some subsequence $\{n_k\}_k$. 
For every resource monotone $\mathfrak{R}_\OO$, define 
\bal
 \overline{\mathfrak{R}}_{\OO}^\infty(\E)&\coloneqq \lim_{\epsilon\to 0}\limsup_{n\to\infty}\frac{1}{n} \mathfrak{R}_{\OO}^\epsilon(\E^{\otimes n})\\
 \underline{\mathfrak{R}}_{\OO}^\infty(\E)&\coloneqq \lim_{\epsilon\to 0}\liminf_{n\to\infty}\frac{1}{n} \mathfrak{R}_{\OO}^\epsilon(\E^{\otimes n}).
\eal
Then, 
\begin{equation}\begin{aligned}
  d\cdot \underline{\mathfrak{R}}^\infty_{\OO}(\T) \leq c\cdot \overline{\mathfrak{R}}^\infty_{\OO}(\T).
\end{aligned}
\label{eq:cost-dist-asymp-with-R_O^infty}\end{equation}
\end{proposition}

\begin{proof}

\begin{equation}\begin{aligned}
  & \!\!\!\!c \, \overline{\mfR}^\infty_{\OO}(\T) \\
  &\geq  \lim_{\eta\to 0}\liminf_{k\to\infty} \frac{\ceil{c n_k}}{n_k} \frac{1}{\ceil{c n_k}} \mfR_{\OO}^\eta(\T^{\otimes \ceil{c{n_k}} })\\
  &\geq \lim_{\eta\to 0}\liminf_{k\to\infty} \frac{1}{n_k} \mfR_{\OO}^{(\sqrt{\eta}+\sqrt{\epsilon_{n_k}})^2}(\E^{\otimes {n_k}})\\
  &\geq \lim_{\eta\to 0}\liminf_{k\to\infty} \frac{1}{n_k}  \mfR_{\OO}^{(\sqrt{\eta}+\sqrt{\epsilon_{n_k}}+\sqrt{\delta_{n_k}})^2}(\T^{\otimes \floor{d{n_k}} }),
  \label{eq:bound regularized}
  \end{aligned}\end{equation}
where the second and third lines follow from Lemma~\ref{lem:converse with smoothing}.

Since $\lim_{k\to\infty} \epsilon_{n_k}=\lim_{k\to\infty} \delta_{n_k}=0$, for arbitrary constant $\xi>0$, it is ensured that $\epsilon_{n_k}<\xi$ and $\delta_{n_k}<\xi$ for sufficiently large $k$.
Since $\mfR_\OO^\epsilon$ is nonincreasing with respect to $\epsilon$, we can bound the last expression as 
\begin{equation}\begin{aligned}
&\geq\lim_{\eta\to 0}\liminf_{k\to\infty}\frac{1}{n_k}\mfR_{\OO}^{(\sqrt{\eta}+2\sqrt{\xi})^2}(\T^{\otimes \floor{dn_k}}). 
\end{aligned}\end{equation}

Since this holds for every $\xi>0$, we can bound the expression in \eqref{eq:bound regularized} by taking the limit $\xi\to 0$ as
\begin{equation}\begin{aligned}
&\lim_{\eta\to 0}\lim_{\xi\to 0}\liminf_{k \to \infty} \frac{1}{n_k}  \mfR_{\OO}^{(\sqrt{\eta}+\sqrt{\epsilon_{n_k}}+\sqrt{\delta_{n_k}})^2}(\T^{\otimes \floor{dn_k}} )\\
&\geq\lim_{\eta\to 0}\lim_{\xi\to 0}\liminf_{n \to \infty} \frac{1}{n}  \mfR_{\OO}^{(\sqrt{\eta}+2\sqrt{\xi})^2}(\T^{\otimes \floor{d{n}} })\\
&=\lim_{\eta\to 0}\lim_{\xi\to 0}\liminf_{n\to\infty} \frac{\floor{d n}}{n} \frac{1}{\floor{d n}} \mfR_{\OO}^{(\sqrt{\eta}+2\sqrt{\xi})^2}(\T^{\otimes \floor{d{n}} })\\
    &\geq d \, \underline{\mfR}^\infty_{\OO}(\T),
\end{aligned}\end{equation}
resulting in the desired inequality in \eqref{eq:cost-dist-asymp-with-R_O^infty}.
\end{proof}

This extends and complements similar relations known for state transformations in settings such as entanglement~\cite{Donald2002uniqueness,Watrous2018theory} and a general class of resources~\cite{Brandao2015reversible,Kuroiwa2020general}. 
In particular, this reduces to the intuitive bound $c\geq d$ when $\overline{\mfR}^\infty_\OO(\T)=\underline{\mfR}^\infty_\OO(\T)>0$. 
Note, however, that the regularized resource measure may take 0 for all channels in some cases~\cite{Gour2009measuring}.


\section{Strong converse property of distillable resource}\label{app:strong converse}

Ref.~\cite{Brandao2010generalization} discussed a generalization of quantum Stein's lemma for resource theories satisfying mild assumptions.
Using this, Ref.~\cite{Brandao2010reversible} characterized the asymptotic distillable entanglement under the set of non-entangling operations with the regularized relative entropy of entanglement. 
In fact, their argument shows more than that --- the regularized relative entropy of entanglement also serves as a strong converse distillation rate.

Here, we extend this strong converse property to general resource theories by combining the results in Refs.~\cite{Brandao2010reversible} and \cite{Regula2020benchmarking} in the case when the generalized quantum Stein's lemma holds.
This can then turn the one-sided strong converse inequality in Corollary~\ref{cor:strong converse} to a double-sided inequality, making both quantities --- yield and cost --- strong converse rates for each other. 

Before stating the result, we recall the following characterization of the fidelity of distillation.  

\begin{lemma}[\cite{Regula2020benchmarking}]\label{lem:fidelity distillation}

For an arbitrary convex and closed set $\FF$, if $D_{\min,\FF}(\Phi)=D_{s,\FF}(\Phi)= r$, then
\bal
 \sup_{\E\in\OO} F(\E(\rho),\Phi) = G_\FF(\rho; 2^r),
\eal
 where
\begin{multline}
 G_\FF(\rho;K) \coloneqq \\
 \sup\lset\Tr[W \rho]\sbar 0\leq W\leq \id,\ \Tr[W\sigma]\leq \frac{1}{K}\ \forall \sigma\in\FF\rset.
 \label{eq:G def}
\end{multline}
\end{lemma}

Then, we get the following relations. 

\begin{proposition}
 Let $\FF$ be a set of free states which:
 \begin{enumerate}
     \item is convex and closed,
     \item contains a full-rank state,
     \item is closed under partial trace and composition of free states.
 \end{enumerate}
 Also, let $\Phi$ be a state that satisfies $D_{\min,\FF}(\Phi^{\otimes m})= D_{s,\FF}(\Phi^{\otimes m}) = m D_{\min,\FF}(\Phi) \; \forall m$. Then,
 \bal
  \tilde d_{\OO_{\max}}^\infty(\rho,\Phi)=d_{\OO_{\max}}^\infty(\rho,\Phi).
 \eal
 
 In particular, for every $\OO\subseteq\OO_{\rm max}$ we have that 
 \bal
  \tilde d_\OO^\infty(\rho,\Phi)\leq \tilde c_\OO^\infty(\rho,\Phi) .
 \eal
\end{proposition}
\begin{proof}

Let $F_\OO(\rho\to\phi)$ be the fidelity of distillation from $\rho$ to $\phi$ under free operations $\OO$, defined as 
\bal
 F_\OO(\rho\to\phi)\coloneqq \sup_{\E\in\OO} F(\E(\rho),\phi).
\eal
Lemma~\ref{lem:fidelity distillation} and the assumed additivity of $D_{\min,\FF}$ gives
\bal
 F_{\OO_{\max}}(\rho^{\otimes n}\to\Phi^{\otimes ny})=G_\FF(\rho^{\otimes n};2^{rny}),
 \label{eq:fidelity distillation upper bound}
\eal
where $r = D_{\min,\FF}(\Phi)=D_{s,\FF}(\Phi)$.

Note that $G_\FF(\rho;K)$ is a convex optimization program, and one can obtain its dual program by following standard techniques in convex optimization theory~\cite{Boyd2004convex} (cf.~\cite{Brandao2010generalization}).
For operators $W\geq 0$, $Y\geq 0$, $Z\in\cone(\FF)$, consider the Lagrangian 
\bal 
 \L(\rho,W;Y,Z)&\coloneqq \Tr[W\rho]+\Tr\left[(\id-W)Y\right]+\Tr\left[(\id-KW)Z\right]\\
 &= \Tr[Y]+\Tr[Z] + \Tr[W(\rho-Y-KZ)].
\eal
This form leads to a dual program:
\begin{equation}\begin{aligned}
&\inf\lset\Tr[Y]+\Tr[Z]\sbar Y\geq 0, Y\geq\rho-KZ, Z\in\cone(\FF)\rset \\
&=\inf\lset\Tr[\rho-KZ]_++\Tr[Z]\sbar Z\in\cone(\FF)\rset
\end{aligned}\end{equation}
where we defined $\Tr[A]_+$ to be the trace over the positive part of the operator $A$.
Since Slater's condition~\cite[Sec.\ 5.2.3]{Boyd2004convex} is satisfied, which can be confirmed by taking $W=\id/(K+1)$ in \eqref{eq:G def}, we get
\bal
 G_\FF(\rho;K)=\inf_{\tilde\sigma\in\cone(\FF)}\left(\Tr[\rho-\tilde\sigma]_++\frac{1}{K}\Tr[\tilde\sigma]\right).
\eal
Taking $\tilde\sigma = 2^{rnb} \sigma$ for some $b \in \RR$ and $\sigma \in \FF$ allows us to write
\bal
 G_\FF(\rho^{\otimes n};2^{rny})=\inf_{\sigma\in\FF,b\in\RR}\left(\Tr[\rho-2^{rnb}\sigma]_++2^{-r(y-b)n}\right).
\eal
for all $n$ and $y$.
Let $R_{{\rm rel}, \FF}^\infty(\rho)$ be the regularized relative entropy resource measure~\cite{Brandao2010reversible,Brandao2015reversible} defined as 
\bal
 R_{{\rm rel}, \FF}^\infty (\rho)\coloneqq \lim_{n\to\infty}\frac{1}{n} \inf_{\sigma\in\FF}D(\rho^{\otimes n}\|\sigma),
 \eal
where $D(\rho\|\sigma)$ is the relative entropy defined for arbitrary two states $\rho$ and $\sigma$ taking $\Tr[\rho\log\rho]-\Tr[\rho\log\sigma]$ if $\supp(\rho)\subseteq\supp(\sigma)$ and $+\infty$ otherwise.
Let us take $y=\frac{1}{r} R_{{\rm rel}, \FF}^\infty (\rho)+\epsilon$ for some $\epsilon > 0$. Then, for each $n$, we can take $b=\frac{1}{r}R_{{\rm rel}, \FF}^\infty (\rho)+ \frac{\epsilon}{2}$ to get

\begin{equation}\begin{aligned}
&G_\FF(\rho^{\otimes n}; 2^{rn(R_{{\rm rel}, \FF}^\infty (\rho)+\epsilon)})\\
&\quad \leq \inf_{\sigma\in\FF}\left(\Tr[\rho^{\otimes n}-2^{rn(\frac{1}{r}R_{{\rm rel}, \FF}^\infty (\rho)+\frac{\epsilon}{2})}\sigma]_+\right) + 2^{-\frac{rn \epsilon}{2}}\\
&\quad = \inf_{\sigma\in\FF}\left(\Tr[\rho^{\otimes n}-2^{n(R_{{\rm rel}, \FF}^\infty (\rho)+\frac{r\epsilon}{2})}\sigma]_+\right) + 2^{-\frac{rn \epsilon}{2}}.
\end{aligned}\end{equation}
Ref.~\cite[Prop.\ III.1]{Brandao2010generalization} showed that the first term approaches 0 in the limit of $n\to\infty$ for every $\epsilon>0$.
Therefore, combining it with \eqref{eq:fidelity distillation upper bound}, we get that for all $\epsilon>0$,
\bal
 \lim_{n\to\infty}F_{\OO_{\max}}\left(\rho^{\otimes n}\to\Phi^{\otimes n(R_{{\rm rel},\FF}^\infty(\rho)/ r+\epsilon)}\right)=0.
 \label{eq:distillation fidelity converse}
\eal

On the other hand, if we take $y=\frac{1}{r}R_{{\rm rel}, \FF}^\infty (\rho)-\epsilon$, the optimum $b$ for each $n$ needs to satisfy $b<y$ because otherwise $G_\FF(\rho^{\otimes n}, r^{ny})$ would diverge as $n\to\infty$. 
Therefore, 
\bal
 G_\FF(\rho^{\otimes n};r^{ny})\geq \inf_{\sigma\in\FF}\Tr[\rho-2^{n(R_{{\rm rel},\FF}^\infty(\rho)-\epsilon r)}\sigma]_+.
\eal
The right-hand side approaches 1 for every $\epsilon>0$~\cite{Brandao2010generalization} and thus leads to 
\bal
 \lim_{n\to\infty}F_{\OO_{\max}}\left(\rho^{\otimes n}\to\Phi^{\otimes n(R_{{\rm rel},\FF}^\infty(\rho)/r-\epsilon)}\right)=1
\label{eq:distillation fidelity achievable}
\eal
for every $\epsilon>0$.
\eqref{eq:distillation fidelity converse} and \eqref{eq:distillation fidelity achievable} imply 
\bal
 d_{\OO_{\max}}^\infty(\rho,\Phi) = \tilde d_{\OO_{\max}}^\infty(\rho,\Phi) = \frac{R_{{\rm rel},\FF}^\infty(\rho)}{r} = \frac{R_{{\rm rel},\FF}^\infty(\rho)}{D_{\min,\FF}(\Phi)}.
\eal

Finally, combining this with Corollary~\ref{cor:strong converse}, we get 
\bal
\tilde d_{\OO_{\max}}(\rho,\Phi)\leq \tilde c_{\OO_{\max}}(\rho,\Phi).
\eal
Noting that  
\bal
  \tilde d_{\OO}(\rho,\Phi)\leq\tilde d_{\OO_{\max}}(\rho,\Phi),\quad \tilde c_{\OO_{\max}}(\rho,\Phi)\leq\tilde c_{\OO}(\rho,\Phi)
\eal
for every $\OO\subseteq\OO_{\max}$ immediately leads to
\bal
 \tilde d_{\OO}(\rho,\Phi)\leq\tilde c_{\OO}(\rho,\Phi),
\eal
which concludes the proof. 
\end{proof}

We note that the fact that $\frac{1}{r} \tilde D_{\max,\OO}^\infty(\E)$ is a strong converse rate for distillation in general resource theories of channels was previously shown in Ref.~\cite{Regula2021fundamental}. However, the relation
\bal
\tilde D_{\max,\FF}^\infty(\rho) = \tilde D_{H,\FF}^\infty(\rho) = R_{{\rm rel}, \FF}^\infty(\rho) 
\eal
and in particular the fact that $\frac{1}{r} R_{{\rm rel}, \FF}^\infty(\rho)$ constitutes an \emph{achievable} rate of distillation under $\OO_{\mathrm{max}}$ is a very non-trivial result established in Ref.~\cite{Brandao2010generalization}, applicable only to resources of quantum states. It is an open question whether an extension of this result to channel theories can be obtained (cf.~\cite{Gour2019how}). It would also be interesting to understand whether the double-sided strong converse bound described in this section can be shown without relying on the generalized quantum Stein's lemma of Ref.~\cite{Brandao2010generalization}.

\section{Proof of Lemma~\ref{lem:resource to twirling}, Proposition~\ref{pro:smooth measures} and Proposition~\ref{pro:isotropic exact}}\label{app:isotropic}

We first show Lemma~\ref{lem:resource to twirling}.

\medskip

\noindent
\textit{Proof of Lemma~\ref{lem:resource to twirling}.}
Suppose $D_{\min,\FF}(\Phi)=D_{s,\FF}(\Phi)$.
By definition of the standard robustness and the closedness of the set $\FF$, there exist states $\tau,\sigma^\star\in\FF$ such that
\bal
 \tau=\frac{\Phi  + (2^{D_{s,\FF}(\Phi)}-1)\sigma^\star}{2^{D_{s,\FF}(\Phi)}}.
 \label{eq:standard robustness optimal free}
\eal
Also, by definition of $D_{\min,\FF}(\Phi)$ and noting that $D_{\min,\FF}(\Phi)=D_{H,\FF}^{\epsilon=0}(\Phi)$, one can confirm that an operator $P^\star=\Pi_\Phi$  satisfies $0\leq P^\star\leq \id$, $\Tr[P^\star\Phi]=1$, and $\Tr[P^\star\eta]\leq 2^{-D_{\min,\FF}(\Phi)},\forall \eta\in\FF$.
Using this operator, define the following channel
\bal
 \Lambda(\cdot)\coloneqq \Tr[P^\star\cdot]\Phi + \Tr[(\id-P^\star)\cdot]\sigma^\star.
\label{eq:measure prepare superchannel}
\eal

One can check that $\Lambda\in\OO_{\max}$ as follows. 
For every $\eta\in\FF$, we get \bal
\Lambda(\eta)=\Tr[P^\star\eta]\Phi+\Tr[(\id-P^\star)\eta]]\sigma^\star
\eal
with
\bal
\Tr[P^\star\eta]\leq 2^{-D_{\min,\FF}(\Phi)}=2^{-D_{s,\FF}(\Phi)}.
\label{eq:superchannl free}
\eal
The convexity of $\FF$ and the form of \eqref{eq:standard robustness optimal free} implies that all states of the form
\bal
\alpha\Phi+(1-\alpha)\sigma^\star,\quad 0\leq \alpha \leq 2^{-D_{s,\FF}(\Phi)}
\label{eq:isotropic free}
\eal
are free states. 
Therefore, \eqref{eq:superchannl free} ensures $\Lambda(\eta)\in\FF,\ \forall\eta\in\FF$, implying $\Lambda\in\OO_{\max}$. 

Next, we show that $\Tr[P^\star\sigma^\star]=\Tr[\Phi\sigma^\star]=0$.
If we apply $\Lambda$ to $\tau$ in \eqref{eq:standard robustness optimal free}, we get 
\bal
 \Lambda(\tau) &=\frac{\Lambda(\Phi)  + (2^{D_{s,\FF}(\Phi)}-1)\Lambda(\sigma^\star)}{2^{D_{s,\FF}(\Phi)}}\\
 &= \frac{\Phi  + (2^{D_{s,\FF}(\Phi)}-1)\left(\Tr[P^\star\sigma^\star]\Phi+\Tr[(\id-P^\star)\sigma^\star]\sigma^\star\right)}{2^{D_{s,\FF}(\Phi)}}\\
 &= \left[2^{-D_{s,\FF}(\Phi)}+ \Tr[P^\star\sigma^\star]\left(1-2^{-D_{s,\FF}(\Phi)}\right)\right]\Phi \\ 
 &\quad+  \left(1-2^{-D_{s,\FF}(\Phi)}\right)\left(1-\Tr[P^\star\sigma^\star]\right) \sigma^\star.
\eal
Since $\tau\in\FF$ and $\Lambda\in\OO_{\max}$, we have $\Lambda(\tau)\in\FF$. 
Therefore, the definition of $D_{s,\FF}(\Phi)$ (or in other words, Eq.~\eqref{eq:isotropic free}) forces $\Tr[P^\star\sigma^\star]\left(1-2^{-D_{s,\FF}(\Phi)}\right)\leq 0$. Since $2^{-D_{s,\FF}(\Phi)}\leq 1$, we must have that $\Tr[P^\star\sigma^\star]=0$.
Combining the fact that $P^\star\geq \Phi$ because of $\Tr[P^\star\Phi]=1$, we also have $0=\Tr[P^\star\sigma^\star]\geq\Tr[\Phi\sigma^\star]\geq 0$, leading to $\Tr[\Phi\sigma^\star]=0$. 

The proof for the case of $D_{\min,\aff(\FF)}(\Phi)=D_{\max,\FF}(\Phi)$ goes analogously. 
The only difference is that the inequality in \eqref{eq:superchannl free} becomes an equality, and $P^\star$ does not necessarily coincide with $\Pi_\Phi$.
By definition of the generalized robustness, there exist a free state $\tau\in\FF$ and some state $\sigma^\star\in\DD$ such that
\bal
 \tau=\frac{\Phi  + (2^{D_{\max,\FF}(\Phi)}-1)\sigma^\star}{2^{D_{\max,\FF}(\Phi)}}.
 \label{eq:generalized robustness optimal free}
\eal
Also, by definition of $D_{\min,\aff(\FF)}(\Phi)\coloneqq D_{H,\aff(\FF)}^{\epsilon=0}$, there exists an operator $P^\star$ that satisfies $0\leq P^\star\leq \id$, $\Tr[P^\star\Phi]=1$, and $\Tr[P^\star\eta]\leq 2^{-D_{\min,\aff(\FF)}(\Phi)},\forall \eta\in\aff(\FF)$.
In fact, these conditions impose a strong constraint
\bal
\Tr[P^\star\eta]=2^{-D_{\min,\aff(\FF)}(\Phi)},\ \forall\eta\in\FF.
\label{eq:d aff min constant}
\eal
To see this, observe first that $D_{\min,\aff(\FF)}(\Phi)\geq 0$ because in \eqref{eq:hypothesis-testing def}, the choice of $P=\id$ ensures $D_H^\epsilon(\rho\|\sigma)\geq 0$ for all $\sigma\in\aff(\FF)$ and all $\epsilon$, resulting in $D_{H,\aff(\FF)}^\epsilon(\rho)\geq 0$ for an arbitrary state $\rho$.
This particularly ensures that 
\bal
 \Tr[P^\star\eta]\leq 1,\quad \forall\eta\in\aff(\FF).
\label{eq:overlap upper bound}
\eal

Then, suppose that there exist two free states $\eta_1,\eta_2\in\FF$ such that $\Tr[P^\star\eta_1]\neq \Tr[P^\star\eta_2]$, where we assume $\Tr[P^\star\eta_1] - \Tr[P^\star\eta_2]=:\Delta>0$ without loss of generality. 
Define an affine combination $\eta(c)\coloneqq c\eta_1-(1-c)\eta_2\in\aff(\FF)$ for an arbitrary real number $c$.
This operator realizes $\Tr[P^\star\eta(c)]=c\Delta-\Tr[P^\star\eta_2]$.
However, since $\Delta>0$, one could violate \eqref{eq:overlap upper bound} by taking sufficiently large $c$, which is a contradiction. 
Thus, we must have $\Tr[P^\star \eta]={\rm const},\ \forall\eta\in\FF$.
Combining this with the definition of $D_{\min,\aff(\FF)}$ leads to the condition \eqref{eq:d aff min constant}.

Using this operator, define the following channel
\bal
 \Lambda(\cdot)\coloneqq \Tr[P^\star\cdot]\Phi + \Tr[(\id-P^\star)\cdot]\sigma^\star.
\eal
It is now easy to check that $\Lambda\in\OO_{\max}$ using \eqref{eq:d aff min constant} and \eqref{eq:generalized robustness optimal free}.

Next, we show that $\Tr[P^\star\sigma^\star]=\Tr[\Phi\sigma^\star]=0$.
If we apply $\Lambda$ to $\tau$ in \eqref{eq:standard robustness optimal free}, we get 
\begin{equation}\begin{aligned}
 &\Lambda(\tau)\\
 &\ =\frac{\Lambda(\Phi)  + (2^{D_{\max,\FF}(\Phi)}-1)\Lambda(\sigma^\star)}{2^{D_{\max,\FF}(\Phi)}}\\
 &\ = \frac{\Phi  + (2^{D_{\max,\FF}(\Phi)}-1)\left(\Tr[P^\star\sigma^\star]\Phi+\Tr[(\id-P^\star)\sigma^\star]\sigma^\star\right)}{2^{D_{\max,\FF}(\Phi)}}\\
 &\ = \left[2^{-D_{\max,\FF}(\Phi)}+ \Tr[P^\star\sigma^\star]\left(1-2^{-D_{\max,\FF}(\Phi)}\right)\right]\Phi \\ 
 &\ \quad+  \left(1-2^{-D_{\max,\FF}(\Phi)}\right)\left(1-\Tr[P^\star\sigma^\star]\right) \sigma^\star.
\end{aligned}\end{equation}
Since $\tau\in\FF$ and $\Lambda\in\OO_{\max}$, we have $\Lambda(\tau)\in\FF$. 
The definition of $D_{\max,\FF}(\Phi)$ states that $2^{-D_{\max,\FF}(\Phi)}$ is the maximum coefficient in front of $\Phi$ such that a mixture with another state becomes a free state. 
This forces $\Tr[P^\star\sigma^\star]\left(1-2^{-D_{\max,\FF}(\Phi)}\right)\leq 0$. 
Since $2^{-D_{\max,\FF}(\Phi)}\leq 1$, we must have that $\Tr[P^\star\sigma^\star]=0$.
Combining the fact that $P^\star\geq \Phi$ because of $\Tr[P^\star\Phi]=1$, we also have $0=\Tr[P^\star\sigma^\star]\geq\Tr[\Phi\sigma^\star]\geq 0$, leading to $\Tr[\Phi\sigma^\star]=0$. 

\qed
\medskip

Using Lemma~\ref{lem:resource to twirling}, we obtain the following simplification of the evaluation of resource measures:

\begin{lemma}\label{lem:measure restriction}
Let $\mfR$ be a resource measure defined by 
\bal
 \mfR(\rho)=\inf_{\sigma\in\FF} D(\rho,\sigma)
 \label{eq:resource full optimization}
\eal
where $D$ is a contractive measure under free operations, i.e., $D(\rho,\sigma)\geq D(\Lambda(\rho),\Lambda(\sigma)),\ \forall \rho,\sigma$ for an arbitrary free channel $\Lambda\in\OO_{\max}$.
Suppose $D_{s,\FF}(\Phi)=D_{\min,\FF}(\Phi)=: r$ for some state $\Phi$, and let $\sigma^\star$ be the state that appears in \eqref{eq:resource to twirling}.
Also, let $\tilde\DD$ and $\tilde\FF$ be the sets of states defined as
\bal
 \tilde\DD&\coloneqq\lset\kappa\Phi+(1-\kappa)\sigma^\star\sbar0\leq\kappa\leq 1\rset,\\
 \tilde\FF & \coloneqq \lset\alpha \Phi + (1-\alpha) \sigma^\star \sbar 0\leq \alpha \leq 2^{-r}\rset.
 \label{eq:generallized isotropic def}
\eal
Then, for every $\rho\in\tilde\DD$, we can restrict the optimization in \eqref{eq:resource full optimization} as  
\bal
 \mfR(\rho)& =\inf_{\sigma\in\tilde\FF} D(\rho,\sigma).
\label{eq:measure optimization restricted}
\eal

On the other hand, if $D_{\max,\FF}(\Phi)=D_{\min,\aff(\FF)}(\Phi)=:r$, then every $\rho\in\tilde\DD$ satisfies 
\bal
 \mfR(\rho)= D(\rho,\tilde\sigma),\quad  \tilde\sigma\coloneqq 2^{-r} \Phi + (1-2^{-r}) \sigma^\star.
 \label{eq:measure optimization restricted reduced}
\eal

\end{lemma}

\begin{proof}
 When $D_{s,\FF}(\Phi)=D_{\min,\FF}(\Phi)=r$, 
 Lemma~\ref{lem:resource to twirling} ensures the existence of a channel $\Lambda\in\OO_{\max}$ of the form \eqref{eq:resource to twirling}.
 Crucially, all states in $\tilde\DD$ are invariant under $\Lambda$. 
 Thus, every $\rho\in\tilde\DD$ satisfies 
 \bal
  \mfR(\rho) &= \inf_{\sigma\in\FF}D(\rho,\sigma)\\
  &\geq \inf_{\sigma\in\FF}D(\Lambda(\rho),\Lambda(\sigma))\\
  &= \inf_{\sigma\in\FF}D(\rho,\Lambda(\sigma))\\
  &\geq \inf_{\sigma\in\tilde\FF}D(\rho,\sigma)\\
  &\geq \inf_{\sigma\in\FF}D(\rho,\sigma)\\
  &= \mfR(\rho)
  \label{eq:measure restricted inequality}
 \eal
 where we used the contractivity of $D$ in the second line, the fact that $\Lambda(\rho)=\rho$ in the third line, $\Lambda(\sigma)\in\tilde\FF\ \ \forall \sigma\in\FF$ due to Lemma~\ref{lem:resource to twirling} and the definition of the robustness measures in the fourth line, $\tilde \FF\subseteq \FF$ in the fifth line, leading to \eqref{eq:measure optimization restricted}.
 
 When $D_{\max,\FF}(\rho)=D_{\min,\aff(\FF)}(\rho)=r$, the operator $P^\star$ in \eqref{eq:resource to twirling} satisfies $\Tr[P^\star\sigma]=2^{-r},\forall \sigma\in\FF$ as in \eqref{eq:d aff min constant}. 
 Thus, \eqref{eq:measure optimization restricted reduced} is obtained by replacing $\tilde\FF$ in \eqref{eq:measure restricted inequality} with $\{\tilde\sigma\}$.
\end{proof}

We are now ready to prove Propositions~\ref{pro:smooth measures} and \ref{pro:isotropic exact}. We will in fact show a more general result, which immediately implies both of the Propositions and allows for the computation of the smoothed entropic measures for all isotropic-like states $\Phi_\kappa$.

\begin{proposition}\label{pro:smooth measures full}
Suppose a state $\Phi$ satisfies $D_{\min,\FF}(\Phi)=D_{s,\FF}(\Phi)=:r$, and let $\sigma^\star$ be the state in \eqref{eq:resource to twirling}. 
Then, every state $\Phi_\kappa=\kappa\Phi+(1-\kappa)\sigma^\star$ with $0\leq\kappa\leq 1$ satisfies  
\bal
 D_{H,\FF}^\epsilon(\Phi_\kappa)=\begin{cases}
 0 & \epsilon=0,\ 0\leq\kappa<1 \\
 r + \log\frac{1}{1-\epsilon} & 0\leq\epsilon<1,\ \kappa=1.
 \end{cases}
\eal

Also, let $\eta_{\min}^\epsilon$ and $\eta_{\max}^\epsilon$ be the minimum and maximum $\eta$ that satisfy $F_{\rm cl}((\eta,1-\eta),(\kappa,1-\kappa))\geq 1-\epsilon$, where $F_{\rm cl}(p,q)\coloneqq \left(\sum_i\sqrt{p_iq_i}\right)^2$ is the fidelity for two classical distributions. 
Then,
 \bal
 D_{\max,\FF}^\epsilon(\Phi_\kappa) = D_{s,\FF}^\epsilon(\Phi_\kappa) = \max\left\{r - \log\frac{1}{\eta_{\min}^\epsilon},0\right\}
 \eal
for all $\epsilon\in[0,1)$.

Similarly, if a state $\Phi$ satisfies $D_{\min,\aff(\FF)}(\Phi)=D_{\max,\FF}(\Phi)=:r$, then 
\bal
 D_{H,\aff(\FF)}^\epsilon(\Phi_\kappa)&=D_{H,\FF}^\epsilon(\Phi_\kappa)\\
 &=\begin{cases}
 0 & \epsilon=0,\ 0<\kappa<1\\
 \frac{1}{1-2^{-r}} &\epsilon=0,\ \kappa=0\\
 r+\log\frac{1}{1-\epsilon} & 0\leq \epsilon<1,\ \kappa=1,
 \end{cases}
\eal
and 
 \bal
D_{\max,\FF}^\epsilon(\Phi_\kappa)=\begin{cases}
r-\log\frac{1}{\eta_{\min}^\epsilon} & \eta_{\min}^\epsilon\geq 2^{-r}\\ 
\log\frac{1-\eta_{\max}^\epsilon}{1-2^{-r}} & \eta_{\max}^\epsilon\leq 2^{-r}\\
0 & \eta_{\min}^\epsilon\leq 2^{-r}\leq \eta_{\max}^\epsilon
\end{cases}
\eal
for all $\epsilon\in[0,1)$.
\end{proposition}
  
We remark that one can obtain an analogous result for smoothing with different distance measures. In particular, the trace-distance smoothing leads to simple expressions with $\eta_{\min}^\epsilon=\kappa-\epsilon$ and $\eta_{\max}^\epsilon=\kappa+\epsilon$.

\begin{proof}
Let us first consider the case when $D_{\min,\FF}(\Phi)=D_{s,\FF}(\Phi)=r$. 
Since $D_H^\epsilon(\cdot\|\cdot)$ satisfies the data-processing inequality, we can use Lemma~\ref{lem:measure restriction} to get 
\bal
D_{H,\FF}^\epsilon(\Phi_\kappa) &= \log\min_{\tau\in\FF}\max_{\substack{0\leq Q \leq \id \\ \Tr[Q\Phi_\kappa]\geq 1-\epsilon}}\Tr[Q\tau]^{-1}\\
 &= \log\min_{\tau\in\tilde\FF}\max_{\substack{0\leq Q \leq \id \\ \Tr[Q\Phi_\kappa]\geq 1-\epsilon}}\Tr[Q\tau]^{-1}
 \label{eq:hypothesis restricted}
\eal
where $\tilde\FF$ is the set of free states defined in \eqref{eq:generallized isotropic def}. 
Let $\Lambda$ and $P^\star$ be the channel and operator in \eqref{eq:resource to twirling}.
Then, since $\Lambda(\Phi_\kappa)=\Phi_\kappa$ and $\Lambda(\tau)=\tau, \forall \tau\in\tilde\FF$, if $Q$ is a feasible solution in \eqref{eq:hypothesis restricted}, $\Lambda^{\dagger}(Q)=\Tr[Q\Phi]P^\star+\Tr[Q\sigma](\id-P^\star)$ is also a feasible solution giving the same objective function, i.e., $\Tr[\Lambda^\dagger(Q)\tau]^{-1}=\Tr[Q\tau]^{-1}$.
Thus, it suffices to take the optimization over  operators of the form $Q=\eta P^\star + \lambda (\id-P^\star)$. 
The conditions $0\leq Q \leq \id$, $\Tr[Q\Phi_\kappa]\geq 1-\epsilon$ are equivalent to 
\bal
0\leq \lambda\leq 1,\ 0\leq \eta \leq 1,\ \eta\kappa+\lambda(1-\kappa) \geq 1-\epsilon.
\label{eq:hypothesis condition}
\eal

Let $\mathcal{A}$ be the set of $(\lambda,\eta)$ that satisfies \eqref{eq:hypothesis condition}.
Then, we can compute $D_{H,\FF}^\epsilon(\Phi_\kappa)$ as
\begin{align}\begin{aligned}
 &\log \min_{0\leq \alpha \leq 2^{-r}}\max_{(\lambda,\eta)\in\mathcal{A}}\Tr[\{\eta P^\star+\lambda(\id-P^\star)\} \{\alpha\Phi+(1-\alpha)\sigma^\star\}]^{-1}\\
 &\quad= - \log \max_{0\leq \alpha \leq 2^{-r}}\min_{(\lambda,\eta)\in\mathcal{A}}[\eta\alpha+\lambda(1-\alpha)]
 \label{eq:DH isotropic intermediate}
\end{aligned}\end{align} 
where we used $\Tr[P^\star\Phi]=1$ and $\Tr[P^\star\sigma^\star]=0$.

Let us first consider the case $\kappa=1$.
The condition for $\mathcal{A}$ in this case turns to 
\bal
 0\leq\lambda\leq 1,\quad 1-\epsilon\leq\eta\leq 1.
 \label{eq:condition target}
\eal
Then, \eqref{eq:DH isotropic intermediate} can be further computed as 
\begin{align}\begin{aligned}
  &- \log \max_{0\leq \alpha \leq 2^{-r}}(1-\epsilon)\alpha = r + \log\frac{1}{1-\epsilon}.
  \label{eq:DH target}
\end{aligned}\end{align} 

On the other hand, suppose $\epsilon=0$ and $0\leq \kappa \leq 1$.
Then, the condition for $\mathcal{A}$ takes the form 
\bal
 0\leq \lambda\leq 1,\ 0\leq\eta\leq 1,\ \eta\kappa+\lambda(1-\kappa)=1.
 \label{eq:condition zero error}
\eal
From this, it is clear that for the case $0<\kappa<1$, we must have $\lambda=\eta=1$, which makes the quantity in \eqref{eq:DH isotropic intermediate} equal to 0. 
In the case of $\kappa=0$, which forces $\lambda=1$, $0\leq\eta\leq 1$, the optimum in \eqref{eq:DH isotropic intermediate} is achieved at $\alpha=\eta=0$, $\lambda=1$, resulting in 0 also. 
The case of $\kappa=1$ is included in \eqref{eq:DH target}; we get $r$ by setting $\epsilon=0$.

To summarize, we showed
\bal
 D_{H,\FF}^\epsilon(\Phi_\kappa)=\begin{cases}
 0 & \epsilon=0,\ 0\leq\kappa<1 \\
 r + \log\frac{1}{1-\epsilon} & 0\leq\epsilon<1,\ \kappa=1.
 \end{cases}
\eal

To show the expression for $D_{\max,\FF}^\epsilon$, note that 
\bal
 D_{\max,\FF}^\epsilon(\Phi_\kappa) &= \inf_{\sigma\in\FF}D_{\max}^\epsilon(\Phi_\kappa\|\sigma)\\
 &=\inf\lset \log s \sbar \rho\leq s\sigma,\ \sigma\in\tilde\FF, F(\rho,\Phi_\kappa)\geq 1-\epsilon \rset 
 \label{eq:dmax optimization free state restricted}
\eal
where in the second line we used Lemma~\ref{lem:measure restriction} because $D_{\max}^\epsilon(\cdot\|\cdot)$ satisfies the data-processing inequality for all quantum channels.
Consider again the map $\Lambda$ in \eqref{eq:resource to twirling}, which maps an arbitrary state to a state in $\tilde\DD$ (c.f., \eqref{eq:generallized isotropic def}), while stabilizing every state $\Phi_\kappa\in\tilde\DD$ and $\sigma\in\tilde\FF$ as $\Lambda(\Phi_\kappa)=\Phi_\kappa$ and $\Lambda(\sigma)=\sigma$.
Then, $\rho\leq s\sigma$ implies $\Lambda(\rho)\leq s\Lambda(\sigma)=s\sigma, \forall \sigma\in\tilde\FF$. 
We also have 
\bal
F(\Lambda(\rho),\Phi_\kappa)=F\left(\Lambda(\rho),\Lambda(\Phi_\kappa)\right)\geq F(\rho,\Phi_\kappa)\geq 1-\epsilon.
\eal
Thus, the optimization in \eqref{eq:dmax optimization free state restricted} is achieved by the states of the form $\rho=\eta \Phi + (1-\eta)\sigma^\star$, $\sigma=\alpha \Phi + (1-\alpha)\sigma^\star$ with constraints $0\leq\eta\leq 1$, $F(\rho,\Phi_\kappa)\geq1-\epsilon$, and $0\leq \alpha\leq 2^{-r}$.
Noting that $\Tr[\Phi\sigma^\star]=0$ and thus $\rho=\eta\Phi\oplus(1-\eta)\sigma^\star$, the condition on $\eta$ can equivalently be written as a condition for two classical distributions $(\eta,1-\eta)$ and $(\kappa,1-\kappa)$:
\bal
 F_{\rm cl}((\eta,1-\eta),(\kappa,1-\kappa))\geq 1-\epsilon.
\label{eq:fidelity dmax}
\eal
Let $\eta_{\min}^\epsilon$ and $\eta_{\max}^\epsilon$ be the minimum and maximum $\eta$ that satisfies \eqref{eq:fidelity dmax}.
Then, we can compute $D_{\max,\FF}^\epsilon(\Phi_\kappa)$ as
\begin{equation}\begin{aligned}
 &\inf_{\substack{\alpha\in[0,2^{-r}]\\\eta\in[\eta_{\min}^\epsilon,\eta_{\max}^\epsilon]}}\lset \log s \sbar \eta\Phi+(1-\eta)\sigma^\star\leq s(\alpha \Phi + (1-\alpha)\sigma^\star)\rset\\
 &=\inf_{\substack{\alpha\in[0,2^{-r}]\\\eta\in[\eta_{\min}^\epsilon,\eta_{\max}^\epsilon]}}\lset \log s \sbar s\alpha-\eta\geq 0,\ s(1-\alpha)-(1-\eta)\geq 0\rset\\
 &=\log\inf_{\substack{\alpha\in[0,2^{-r}]\\\eta\in[\eta_{\min}^\epsilon,\eta_{\max}^\epsilon]}}\max\left\{\frac{\eta}{\alpha},\frac{1-\eta}{1-\alpha}\right\}.
 \label{eq:dmax smooth proof intermediate}
\end{aligned}\end{equation}

Note that $\max\left\{\frac{\eta}{\alpha},\frac{1-\eta}{1-\alpha}\right\}$ is clearly lower bounded by 1, and it is achieved when $\eta=\alpha$. Thus,
when $\eta_{\min}^\epsilon\leq 2^{-r}$, we immediately get $D_{\max,\FF}^\epsilon(\Phi_\kappa)=0$. 
On the other hand, when $\eta_{\min}^\epsilon\geq 2^{-r}$, we always have $\max\left\{\frac{\eta}{\alpha},\frac{1-\eta}{1-\alpha}\right\}=\frac{\eta}{\alpha}$, and the minimization over $\alpha$ and $\eta$ gives $D_{\max,\FF}^\epsilon(\Phi_\kappa)=r-\log\frac{1}{\eta_{\min}^\epsilon}$. 
These can concisely be written as 
\bal
D_{\max,\FF}^\epsilon(\Phi_\kappa)=\max\left\{r-\log\frac{1}{\eta_{\min}^\epsilon},0\right\},
\label{eq:dmax smoothed final expression proof}
\eal
which concludes the proof for the expression of $D_{\max,\FF}^\epsilon$.

The smooth standard robustness can be computed similarly. 
Note that $D_{s,\FF}^\epsilon(\Phi)=\inf_{\sigma\in\FF}D_{s,\FF}^\epsilon(\Phi\|\sigma)$ where 
\begin{align}\begin{aligned}
 &D_{s,\FF}^\epsilon(\rho_1\|\rho_2)\coloneqq\\
 &\inf\lset\log(1+s)\sbar\frac{\rho_1'+s\tau}{1+s}=\rho_2,\ \tau\in\FF,\ F(\rho_1',\rho_1)\geq 1-\epsilon\rset,
\end{aligned}\end{align}
and $D_{s,\FF}^\epsilon$ is contractive under every free operation $\E\in\OO_{\max}$, as for all states $\rho$ and $\sigma$,  
\bal
 D_{s,\FF}^\epsilon(\rho\|\sigma)&= D_{s,\FF}(\tilde\rho\|\sigma)\\
 &\geq D_{s,\FF}(\E(\tilde\rho)\|\E(\sigma))\\
 &\geq D_{s,\FF}^\epsilon(\E(\rho)\|\E(\sigma))
 \eal
where we set $\tilde\rho$ as the optimal state realizing the standard robustness, and in the third line we used that $F(\E(\tilde\rho),\E(\rho))\geq F(\tilde\rho,\rho)\geq 1-\epsilon$.
Thus, we can use Lemma~\ref{lem:measure restriction} to compute $D_{s,\FF}^\epsilon(\Phi_\kappa)$ as
\begin{equation}\begin{aligned}
 &\inf_{\substack{\alpha\in[0,2^{-r}]\\\eta\in[\eta_{\min}^\epsilon,\eta_{\max}^\epsilon]}}\lset \log s \sbar \eta\Phi+(1-\eta)\sigma^\star\leq_\FF s(\alpha \Phi + (1-\alpha)\sigma^\star)\rset\\
 &=\inf_{\substack{\alpha\in[0,2^{-r}]\\\eta\in[\eta_{\min}^\epsilon,\eta_{\max}^\epsilon]}}\lset \log s \sbar s\alpha-\eta\geq 0,\ s(1-\alpha)-(1-\eta)\geq 0,\right.\\
 &\hspace{4cm}\left.\frac{s(1-\alpha)-(1-\eta)}{s\alpha-\eta}\geq 2^r-1\rset\\
 &=\log\inf_{\substack{\alpha\in[0,2^{-r}]\\\eta\in[\eta_{\min}^\epsilon,\eta_{\max}^\epsilon]}}\max\left\{\frac{\eta}{\alpha},\frac{1-\eta}{1-\alpha},\frac{1-2^r\eta}{1-2^r\alpha}\right\}
\end{aligned}\end{equation}
where in the first line, we used the notation $A\leq_\FF B \iff B-A\in \cone(\FF)$. 
When $\eta_{\min}^\epsilon\leq 2^{-r}$, we immediately get $D_{s,\FF}^\epsilon(\Phi_\kappa)=0$, which is achieved at $\alpha=\eta$. 
On the other hand, when $\eta_{\min}^\epsilon\geq 2^{-r}$, 
we have $\frac{\eta}{\alpha}>\frac{1-\eta}{1-\alpha}$ and $\frac{\eta}{\alpha}>\frac{1-2^r\eta}{1-2^r\alpha}$, in which the minimization over $\alpha$ and $\eta$ results in 
$D_{s,\FF}^\epsilon(\Phi_\kappa)=r-\log\frac{1}{\eta_{\min}}$. 
These two can be combined as 
\bal
D_{s,\FF}^\epsilon(\Phi_\kappa)=\max\left\{r-\log\frac{1}{\eta_{\min}^\epsilon},0\right\}.
\label{eq:ds smoothed final expression proof}
\eal

The proof for the case when $D_{\min,\aff(\FF)}(\Phi)=D_{\max,\FF}(\Phi)$ goes analogously, where we basically change the region of optimization over free states from $\tilde\FF$ to $\{\tilde\sigma\}$ defined in \eqref{eq:measure optimization restricted reduced}.
Since $D_H^\epsilon(\rho\|\sigma)$ satisfies the data-processing inequality even if $\sigma$ is not a positive operator~\cite{Regula2020oneshot}, we can employ the same argument as that in Lemma~\ref{lem:measure restriction} to get 
\bal
D_{H,\aff(\FF)}^\epsilon(\Phi_\kappa) = \log\sup_{\substack{0\leq Q \leq \id \\ \Tr[Q\Phi_\kappa]\geq 1-\epsilon}}\Tr[Q\tilde\sigma]^{-1}
\eal
where we used that the map $\Lambda$ in \eqref{eq:resource to twirling} transforms all $\sigma\in\aff(\FF)$ to $\tilde\sigma$.
We can then restrict the optimization to $Q=\eta P^\star+\lambda(\id-P^\star)$ with the condition \eqref{eq:hypothesis condition}.
Recall that $\mathcal{A}$ is the set of $(\lambda,\eta)$ that satisfies \eqref{eq:hypothesis condition}.
Then, we can compute $D_{H,\aff(\FF)}^\epsilon(\Phi_\kappa)$ as
\begin{align}\begin{aligned}
 &\log \sup_{(\lambda,\eta)\in\mathcal{A}}\Tr[\{\eta P^\star+\lambda(\id-P^\star)\} \{2^{-r}\Phi+(1-2^{-r})\sigma^\star\}]^{-1}\\
 &\quad= - \log \inf_{(\lambda,\eta)\in\mathcal{A}}[\eta\, 2^{-r}+\lambda(1-2^{-r})].
\end{aligned}\end{align} 

When $\kappa=1$, the condition on $\mathcal{A}$ becomes \eqref{eq:condition target}. 
Then, we evaluate this as 
\begin{align}\begin{aligned} 
 -\log\left[(1-\epsilon)2^{-r}\right]=r + \log\frac{1}{1-\epsilon}.
 \end{aligned}\end{align}

On the other hand, when $\epsilon=0$ and $0<\kappa<1$, we are constrained to $\lambda=\eta=1$, leading to value 0. 
When $\kappa=0$, we have $\lambda=1$, $0\leq\eta\leq1$, giving $\log \frac{1}{1-2^{-r}}$.
When $\kappa=1$, we have $0\leq\lambda\leq 1$, $\eta=1$, giving $r$.

Noting that the above argument gives the same conclusion for $D_{H,\FF}^\epsilon$ as well, we reorganize the above form to reach
\bal
 D_{H,\aff(\FF)}^\epsilon(\Phi_\kappa)&=D_{H,\FF}^\epsilon(\Phi_\kappa)\\
 &=\begin{cases}
 0 & \epsilon=0,\ 0<\kappa<1\\
 \log\frac{1}{1-2^{-r}} &\epsilon=0,\ \kappa=0\\
 r+\log\frac{1}{1-\epsilon} & 0\leq \epsilon<1,\ \kappa=1
 \end{cases}
\eal

As for $D_{\max,\FF}^\epsilon$, we follow the same argument up to \eqref{eq:dmax smooth proof intermediate} to get 
\begin{equation}\begin{aligned}
 D_{\max,\FF}^\epsilon(\Phi_\kappa)=\log\inf_{\eta\in[\eta_{\min}^\epsilon,\eta_{\max}^\epsilon]}\max\left\{\frac{\eta}{2^{-r}},\frac{1-\eta}{1-2^{-r}}\right\}.
\end{aligned}\end{equation}

If $\eta_{\min}^\epsilon\geq 2^{-r}$, we always have $\frac{\eta}{2^{-r}}\geq \frac{1-\eta}{1-2^{-r}}$ and thus $D_{\max,\FF}^\epsilon(\Phi_\kappa)=\log\frac{\eta_{\min}^\epsilon}{2^{-r}}$.
If $\eta_{\max}^\epsilon\leq 2^{-r}$, we always have $\frac{\eta}{2^{-r}}\leq \frac{1-\eta}{1-2^{-r}}$ and thus $D_{\max,\FF}^\epsilon(\Phi_\kappa)=\log\frac{1-\eta_{\max}^\epsilon}{1-2^{-r}}$.
If $\eta_{\min}^\epsilon\leq 2^{-r}\leq \eta_{\max}^\epsilon$, then 
$D_{\max,\FF}^\epsilon(\Phi_\kappa)=0$, which is achieved at $\eta=2^{-r}$.

Summarizing, we get
\bal
D_{\max,\FF}^\epsilon(\Phi_\kappa)=\begin{cases}
r-\log\frac{1}{\eta_{\min}^\epsilon} & \eta_{\min}^\epsilon\geq 2^{-r}\\ 
\log\frac{1-\eta_{\max}^\epsilon}{1-2^{-r}} & \eta_{\max}^\epsilon\leq 2^{-r}\\
0 & \eta_{\min}^\epsilon\leq 2^{-r}\leq \eta_{\max}^\epsilon
\end{cases},
\eal
concluding the proof.

\end{proof}

The generality of Lemma~\ref{lem:measure restriction} has wide applicability beyond the above measures. 
As an example, it allows us to provide an exact evaluation for a measure based on the trace distance. 

\begin{proposition}
 Define the trace-distance measure $R_{\tr,\FF}(\rho)\coloneqq \min_{\sigma\in\FF} \frac{1}{2}\|\rho-\sigma\|_1$.
 If $D_{\min,\FF}(\Phi)=D_{s,\FF}(\Phi)=:r$ or $D_{\min,\aff(\FF)}(\Phi)=D_{\max,\FF}(\Phi)=:r$ then
 \bal
  R_{\tr,\FF}(\Phi) = 1-2^{-r}.
 \eal
\end{proposition}
\begin{proof}
For the case when $D_{\min,\FF}(\Phi)=D_{s,\FF}(\Phi)=:r$, 
\bal
 R_{\tr,\FF}(\Phi)&=\min_{\sigma\in\tilde\FF}\frac{1}{2}\|\Phi-\sigma\|_1\\
 &=\min_{0\leq\alpha\leq 2^{-r}} \frac{1}{2}\|\Phi - \left[\alpha\Phi + (1-\alpha)\sigma^\star\right]\|_1\\
 &=\min_{0\leq\alpha\leq 2^{-r}} (1-\alpha)\\
 &=1-2^{-r},
\eal
where in the first line we used Lemma~\ref{lem:measure restriction}, and in the third line we used that $\Tr[\Phi\sigma]=0$. 
The case for $D_{\min,\aff(\FF)}(\Phi)=D_{\max,\FF}(\Phi)$ can be shown analogously. 
\end{proof}

\section{Proof of Theorem~\ref{thm:second law general state}}\label{app:second law general state}

\begin{proof}
The proof combines Proposition~\ref{pro:smooth measures} and the argument in Ref.~\cite{Wilde2021second}.
For two states $\Phi_1,\Phi_2\in\TT$, consider the following sequence of transformations: 
\bal
 \Phi_2\xrightarrow[\OO]{\epsilon_2}\rho \xrightarrow[\OO]{\epsilon_1}\Phi_1.
\eal

Let $\Lambda_2$ and $\Lambda_1$ be free transformations corresponding to the first and the second transformation above.
Define the purified distance $P(\rho,\sigma)\coloneqq \sqrt{1-F(\rho,\sigma)}$, which satisfies the following triangle inequality~\cite{Tomamichel2015quantum} 
\bal
 P(\rho,\tau)\leq P(\rho,\sigma)\sqrt{F(\sigma,\tau)}+P(\sigma,\tau)\sqrt{F(\rho,\sigma)}
\eal
for every set of three states $\rho,\sigma,\tau$ such that $P(\rho,\sigma)^2+P(\sigma,\tau)^2\leq 1$. 
Applying this to our setting, we get that $\Lambda_1\circ\Lambda_2\in\OO$ achieves the transformation $\Phi_2\xrightarrow[\OO]{\epsilon'}\Phi_1$ with
\bal
\epsilon'\coloneqq \left(\sqrt{\epsilon_2(1-\epsilon_1)}+\sqrt{\epsilon_1(1-\epsilon_2)}\right)^2
\label{eq:error state}
\eal
whenever $\epsilon_1+\epsilon_2\leq 1$ is satisfied.
Moreover, $\epsilon_1+\epsilon_2< 1$ ensures $\epsilon'< 1$. This can be shown as follows. 
Direct calculation gives
\bal
 \epsilon'=\epsilon_1+\epsilon_2 +2\sqrt{\epsilon_1\epsilon_2}\left(\sqrt{(1-\epsilon_1)(1-\epsilon_2)}- \sqrt{\epsilon_1\epsilon_2}\right).
\eal
Let us change parameters as $\Delta\coloneqq 1-(\epsilon_1+\epsilon_2)$ and $p\coloneqq \epsilon_1\epsilon_2$, which behave as independent variables under the condition that $(1-\Delta)^2-4p\geq 0$, $0<\Delta\leq 1$, $p\geq 0$, noting that $\epsilon_1,\epsilon_2$ are solutions of the quadratic equation $x^2-(1-\Delta) x + p=0$ while satisfying $\epsilon_1,\epsilon_2
\geq 0$ and $\epsilon_1+\epsilon_2<1$.
Then, we get 
\bal
 1-\epsilon' = \Delta - 2\sqrt{p}\left(\sqrt{\Delta+p}-\sqrt{p}\right).
\eal
This gives
\bal
 \frac{\partial}{\partial \Delta} (1-\epsilon')=1 - \frac{\sqrt{p}}{\sqrt{\Delta+p}}>0,
\eal
implying that $1-\epsilon'$ is a strictly increasing function of $\Delta$ for an arbitrary fixed $p$. 
Since $(1-\epsilon')|_{\Delta=0} = 0$ for every $p$, we get that $1-\epsilon'>0$ for all $\epsilon_1$ and $\epsilon_2$ satisfying $\epsilon_1+\epsilon_2<1$. 

Therefore, for every dilution operation $\Lambda_2$ with error $\epsilon_2$ and distillation operation $\Lambda_1$ with error $\epsilon_1$, we get
\bal
  D_{\min,\FF}(\Phi_1)
  &\leq D_{H,\FF}^{\epsilon'}(\Phi_2)\\
 &= D_{\min,\FF}(\Phi_2) + \log\frac{1}{1-\epsilon'},
\eal
where the first line follows from Lemma~\ref{lem:converses} and the second line from Proposition~\ref{pro:smooth measures}. Optimizing over all feasible $\Lambda_1$ and $\Lambda_2$ and using the definition of distillable resource and resource cost gives the desired relation
\bal
d_{\OO}^{\epsilon_1}(\rho) \leq c_{\OO}^{\epsilon_2}(\rho) + \log\frac{1}{1-\epsilon'}.
\eal
\end{proof}

\section{Comparing the bounds in Theorem~\ref{thm:second law general} and Theorem~\ref{thm:second law general state}}

\label{app:comparison bounds}

Here, we show that the bound in Theorem~\ref{thm:second law general state} is always tighter than the bound in Theorem~\ref{thm:second law general}.

Let us first compare the expression in Theorem~\ref{thm:second law general} valid for $\epsilon_1+\epsilon_2<1$ and the bound in Theorem~\ref{thm:second law general state}.
Comparing the denominators of the expressions inside the logarithm, we get
\begin{align}\begin{aligned}
 &\left(1-\epsilon'\right)-\left(\sqrt{1-\epsilon_2}-\sqrt{\epsilon_1}\right)^2 \\&= -2\epsilon_1+2\epsilon_1\epsilon_2-2\sqrt{\epsilon_1\epsilon_2(1-\epsilon_2)(1-\epsilon_1)}+2\sqrt{\epsilon_1(1-\epsilon_2)}\\
 &=2\sqrt{\epsilon_1(1-\epsilon_2)}\left(1-\sqrt{\epsilon_2(1-\epsilon_1)}\right) - 2\epsilon_1(1-\epsilon_2)\\
 &=2\sqrt{\epsilon_1(1-\epsilon_2)}\left(1-\sqrt{\epsilon_2(1-\epsilon_1)}- \sqrt{\epsilon_1(1-\epsilon_2)}\right)\\
 &=2\sqrt{\epsilon_1(1-\epsilon_2)}\left(1-\sqrt{\epsilon'}\right)
\end{aligned}\end{align}
where $\epsilon'$ is the one introduced in \eqref{eq:error state}.
As shown in Appendix~\ref{app:second law general state}, we have $\epsilon'<1$ (and thus $1-\sqrt{\epsilon'}>0$) for all $\epsilon_1$ and $\epsilon_2$ satisfying $\epsilon_1+\epsilon_2<1$. 
This implies $\left(1-\epsilon'\right)-\left(\sqrt{1-\epsilon_2}-\sqrt{\epsilon_1}\right)^2>0$, and in particular 
\bal
 \log\frac{1}{1-\epsilon'}< \log\left(\sqrt{1-\epsilon_2}-\sqrt{\epsilon_1}\right)^{-2}
\eal
for $\epsilon_1+\epsilon_2<1$. 

We next compare the expression in Theorem~\ref{thm:second law general} to the bound in Theorem~\ref{thm:second law general state} for $\epsilon_1+\sqrt{\epsilon_2}<1$. 
Direct calculation gives
\begin{align}\begin{aligned}
 &\left(1-\epsilon'\right)-\left(1-\epsilon_1-\sqrt{\epsilon_2}\right) \\&= \epsilon_1+\sqrt{\epsilon_2}-\left(\sqrt{\epsilon_1(1-\epsilon_2)}+\sqrt{\epsilon_2(1-\epsilon_1)}\right)^2\\
 &=\sqrt{\epsilon_2}-\epsilon_2+2\epsilon_1\epsilon_2-2\sqrt{\epsilon_1\epsilon_2(1-\epsilon_1)(1-\epsilon_2)}.
\end{aligned}\end{align}

Let us define $g(\epsilon_1,\epsilon_2)\coloneqq\left(1-\epsilon'\right)-\left(1-\epsilon_1-\sqrt{\epsilon_2}\right)$ and consider varying $\epsilon_1$ for a fixed $\epsilon_2$. 
Although we are interested in the region $0\leq \epsilon_1 <1-\sqrt{\epsilon_2}$, let us consider an extended region $0\leq \epsilon_1 <1-\epsilon_2$ as a domain of $g(\epsilon_1,\epsilon_2)$. 
Since $1-\sqrt{\epsilon_2}\leq 1-\epsilon_2$, if we can show that $g(\epsilon_1,\epsilon_2)\geq 0$ for $0\leq \epsilon_1 \leq 1-\epsilon_2$, then $g(\epsilon_1,\epsilon_2)\geq 0$ for $0\leq \epsilon_1 \leq 1-\sqrt{\epsilon_2}$ automatically follows.

Since $g(0,\epsilon_2)=g(1-\epsilon_2,\epsilon_2)=\sqrt{\epsilon_2}-\epsilon_2\geq 0$, it suffices to show that $g(\cdot,\epsilon_2)\geq 0$ at local minima. 
Since 
\bal
 \frac{\partial}{\partial\epsilon_1}g(\epsilon_1,\epsilon_2) &= 2\epsilon_2 - \frac{\epsilon_2(1-\epsilon_2)(1-2\epsilon_1)}{\sqrt{\epsilon_1\epsilon_2(1-\epsilon_1)(1-\epsilon_2)}}\\
 &\propto 2\sqrt{\epsilon_1\epsilon_2(1-\epsilon_1)(1-\epsilon_2)}-(1-\epsilon_2)(1-2\epsilon_1),
\eal
the local minima occur at $\epsilon_1^\star$ such that 
\bal
 2\sqrt{\epsilon_1^\star\epsilon_2(1-\epsilon_1^\star)(1-\epsilon_2)}-(1-\epsilon_2)(1-2\epsilon_1^\star) = 0,
 \label{eq:local minima}
\eal
for which we get a simplified form of $g(\epsilon_1^\star,\epsilon_2)$ as 
\bal
 g(\epsilon_1^\star,\epsilon_2)&=\sqrt{\epsilon_2}-\epsilon_2+2\epsilon_1^\star\epsilon_2-(1-\epsilon_2)(1-2\epsilon_1^\star)\\
 &=\sqrt{\epsilon_2}-1+2\epsilon_1^\star.
 \label{eq:local minimum evaluated}
\eal

Eq.~\eqref{eq:local minima} can be alternatively written as 
\bal
 -1+ \epsilon_1^\star +\left(\sqrt{\epsilon_1^\star(1-\epsilon_2)}+\sqrt{\epsilon_2(1-\epsilon_1^\star)}\right)^2 = 0,
\eal
leading to 
\bal
  \sqrt{\epsilon_1^\star(1-\epsilon_2)} = (1-\sqrt{\epsilon_2})\sqrt{1-\epsilon_1^\star},
\eal
for the region $\epsilon_1^\star \leq 1$.
This gives the expression of the local minimum as 
\bal
 \epsilon_1^\star = \frac{1-\sqrt{\epsilon_2}}{2}.
\eal
Plugging this into \eqref{eq:local minimum evaluated} gives
\bal
 g(\epsilon_1^\star,\epsilon_2)&= 0.
\eal
This concludes the proof that $\left(1-\epsilon'\right)-\left(1-\epsilon_1-\sqrt{\epsilon_2}\right)\geq 0$ for all $\epsilon_1$ and $\epsilon_2$ satisfying $\epsilon_1+\sqrt{\epsilon_2}<1$, and in particular,
\bal
 \log \frac{1}{1-\epsilon'} \leq \log \frac{1}{1-\epsilon_1-\sqrt{\epsilon_2}}.
\eal


\section{Proof of Lemma~\ref{lem:twirling to resource}}

\begin{proof}
Let us write $\Lambda\in\OO_{\max}\cap\S(\Phi)$ as 
\bal
 \Lambda(\cdot)=\Tr[P\cdot]\Phi + \tilde\Lambda(\cdot)
\eal
where $\Tr[P\Phi]=1$ and $\tilde\Lambda\in{\rm CP}$.
Let $\sigma\in\FF$ be an arbitrary free state. Since $\Lambda\in\OO_{\max}$, we get 
\bal
 \Tr[P\sigma]\Phi +\left(1-\Tr[P\sigma]\right)\tau\in\FF
\eal
where we defined a state $\tau\coloneqq \tilde\Lambda(\sigma)/\Tr[\tilde\Lambda(\sigma)]$ and used that $\Tr[\tilde\Lambda(\sigma)]=1-\Tr[P\sigma]$ because $\Lambda$ is trace preserving. 
Since this holds for every free state $\sigma$, the first expression of $D_{\max,\FF}$ in \eqref{eq:robustness def} implies that
\bal
\Tr[P\sigma]^{-1}\geq 2^{D_{\max,\FF}(\Phi)},
\forall \sigma\in\FF.
\eal
Note that 
\bal
2^{D_{\min,\FF}(\Phi)}=\min_{\sigma\in\FF}\max_{\substack{0\leq Q\leq \id\\\Tr[Q\Phi]=1}}\Tr[Q\sigma]^{-1}\\
\geq \max_{\substack{0\leq Q\leq \id\\\Tr[Q\Phi]=1}}\min_{\sigma\in\FF}\Tr[Q\sigma]^{-1} 
\eal
where the inequality is due to the max-min inequality. (This can be made into an equality by further using the convexity of $\FF$, the linearity of the trace, and  Sion's minimax theorem, which however we do not need here.)
These give
\bal
 2^{D_{\min,\FF}(\Phi)} & \geq \max_{\substack{0\leq Q\leq \id\\\Tr[Q\Phi]=1}}\min_{\sigma\in\FF}\Tr[Q\sigma]^{-1} \\
 & \geq \min_{\sigma\in\FF}\Tr[P\sigma]^{-1} \\
 & \geq 2^{D_{\max,\FF}(\Phi)},
\eal
implying $D_{\min,\FF}(\Phi)\geq D_{\max,\FF}(\Phi)$.
On the other hand, note that $D_{\min,\FF}(\Phi)\leq D_{\max,\FF}(\Phi)$ for every $\Phi$ because 
\bal
D_{\min,\FF}(\Phi) & =\inf_{\sigma\in\FF}D_{\min}(\Phi\|\sigma),\\
D_{\max,\FF}(\Phi) & =\inf_{\sigma\in\FF}D_{\max}(\Phi\|\sigma),\\
D_{\min}(\Phi\|\sigma) & \leq D_{\max}(\Phi\|\sigma)\quad  \forall \Phi,\sigma.
\eal
These result in $D_{\min,\FF}(\Phi)=D_{\max,\FF}(\Phi)$.
The proof for $D_{s,\FF}$ goes analogously.

Furthermore, $\Lambda^{\otimes n}$ takes the form 
\bal
\Lambda^{\otimes n}= \Tr[P^{\otimes n}\cdot]\Phi^{\otimes n} + \Lambda'(\cdot)
\eal
where $\Lambda'$ is another completely-positive map.
Noting that $0\leq P^{\otimes n}\leq \id$, $\Tr[P^{\otimes n}\Phi^{\otimes n}]=1$, and 
\begin{multline}
\Lambda^{\otimes n} = (\Lambda\otimes\idc\otimes\dots\otimes\idc)\circ(\idc\otimes\Lambda\otimes\dots\otimes\idc)\\
\circ\dots\circ(\idc\dots\otimes\idc\otimes\Lambda)
\end{multline}
is in $\OO_{\max}$ because $\Lambda$ is completely free, we get $\Lambda^{\otimes n}\in\OO_{\max}\cap\S(\Phi^{\otimes n})$. 
Thus, by using the same argument, we get $D_{\min,\FF}(\Phi^{\otimes n})=D_{\max,\FF}(\Phi^{\otimes n})$ for every positive integer $n$. 
\end{proof}


\section{Proof of Proposition~\ref{pro:group twirling}}

\begin{proof}
Let us define the group twirling operation 
\bal
 \Xi(\cdot)\coloneqq \int_G \mathrm{d}g\, U(g) \cdot U(g)^\dagger
\eal
where the integral is taken over the Haar measure of the group $G$.
We first show a general form of how the twirling operation acts on an arbitrary state $\rho$ (in particular, Eq.~\eqref{eq:twirling general}).
Although Eq.~\eqref{eq:twirling general} was already presented in Ref.~\cite[Sec.~2.6]{Chiribella2006optimal}, we here provide a self-contained explanation for completeness.

An arbitrary unitary representation $\{U(g)\}_{g\in G}$ of a finite or a compact Lie group $G$ can be decomposed into a direct sum of irreducible representations~\cite{Fulton2013representation}. 
Accordingly, the Hilbert space that each $U(g)$ acts on admits the following decomposition:
\bal
 \H = \bigoplus_\mu \H_\mu^{(r)}\otimes \H_\mu^{(m)}
\eal
where $\mu$ labels the irreducible representations, $\H_\mu^{(r)}$ denotes the subspace on which each irreducible representation acts nontrivially, and $\H_\mu^{(m)}$ denotes the multiplicity subspace.
We write the dimensions of $\H_{\mu}^{(r)}$ and $\H_{\mu}^{(m)}$ as $d_\mu^{(r)}$ and $d_\mu^{(m)}$ respectively. 

Each $U(g)$ can be written as 
\bal
 U(g) = \bigoplus_\mu U_\mu^{(r)}(g)\otimes\id^{(m)}_\mu.
 \label{eq:unitary decomposition}
\eal
Also, Schur's lemma~\cite{Fulton2013representation} imposes a structure to every \emph{symmetric state} $\sigma$ satisfying $U(g)\sigma U(g)^\dagger=\sigma\ \forall g$ as 
\bal
 \sigma = \bigoplus_\mu q_\mu \frac{\id_\mu^{(r)}}{d_\mu^{(r)}}\otimes\eta_\mu^{(m)}
 \label{eq:symmetric states}
\eal
where $\{q_\mu\}_\mu$ is a probability distribution, $\eta_\mu^{(m)}$ is a quantum state, and $\id_{\mu}^{(r)}$ is the projector onto $\H_\mu^{(r)}$.

It is easy to see that every output of the group twirling operation is a symmetric state, because, for every $g\in G$,
\bal
 U_g\Xi(\rho)U_g^\dagger &= \int \mathrm{d}g' U(g)U(g')\rho U(g')^\dagger U(g)^\dagger\\
 &= \int_G \mathrm{d}g'\,  U(gg')\rho U(gg')^\dagger\\
 &= \int_G \mathrm{d}g'\, U(g')\rho U(g')^\dagger\\
 &= \Xi(\rho)
\eal
where in the third line we used the left and right invariance of the Haar measure.

Let us also define the projection onto invariant subspaces as
\bal
 \P(\cdot)\coloneqq \sum_\mu \id_\mu^{(r)}\otimes\id_\mu^{(m)}\cdot \id_\mu^{(r)}\otimes\id_\mu^{(m)}.
\label{eq:projector}
\eal
The projection turns every state $\rho$ into the form
\bal
 \P(\rho) = \bigoplus_\mu p_\mu \sigma_\mu^{(r,m)}
\eal
where $\{p_\mu\}_\mu$ is probability distribution, and $\sigma_\mu^{(r,m)}$ is a quantum state acting on $\H_\mu^{(r)}\otimes\H_\mu^{(m)}$. 
Also, the forms of \eqref{eq:unitary decomposition} and \eqref{eq:projector} imply that $U_g\P(\rho)U_g^\dagger=\P(U_g\rho U_g^\dagger)\ \forall \rho, \forall g\in G$, and hence
\bal
\Xi\circ\P(\rho)&=\P\circ\Xi(\rho)\\
&=\Xi(\rho)
\eal
for every state $\rho$, where in the second line, we used the fact that $\Xi$ maps every state to a symmetric state, which has the form \eqref{eq:symmetric states}.
Then,
\begin{align}\begin{aligned}
&\Xi(\rho) = \Xi(\P(\rho))\\
&= \bigoplus_\mu\int_G \mathrm{d}g\, p_\mu \left(U_\mu^{(r)}(g)\otimes\id_\mu^{(m)}\right)\sigma_\mu^{(r,m)}\left(U_\mu^{(r)}(g)^\dagger\otimes\id_\mu^{(m)}\right).
\label{eq:twirling subspace}
\end{aligned}\end{align}

Since $\Xi(\rho)$ is a symmetric state having a structure of \eqref{eq:symmetric states}, and $U_\mu^{(r)}$ in \eqref{eq:twirling subspace} acts only on $\H_\mu^{(r)}$ nontrivially, we must have 
\begin{align}\begin{aligned}
\int_G \mathrm{d}g\, p_\mu &\left(U_\mu^{(r)}(g)\otimes\id_\mu^{(m)}\right)\sigma_\mu^{(r,m)}\left(U_\mu^{(r)}(g)^\dagger\otimes\id_\mu^{(m)}\right)\\
&\quad=p_\mu\frac{\id_\mu^{(r)}}{d_\mu^{(r)}}\otimes\sigma_\mu^{(m)},\ \forall \mu.
\end{aligned}\end{align}
where $\sigma_\mu^{(m)}\coloneqq\Tr_r[\sigma_\mu^{(r,m)}]$ and $\Tr_r$ denotes the partial trace over the system with $\H_\mu^{(r)}$.
Therefore, we get
\bal
\Xi(\rho) &= \sum_\mu p_\mu\frac{\id_\mu^{(r)}}{d_\mu^{(r)}}\otimes \sigma_\mu^{(m)}\\
&=\sum_\mu \frac{\id_\mu^{(r)}}{d_\mu^{(r)}}\otimes\Tr_r\left[\P(\rho)\id_\mu^{(r)}\otimes\id_\mu^{(m)}\right]\\
&=\sum_\mu \frac{\id_\mu^{(r)}}{d_\mu^{(r)}}\otimes\Tr_r\left[\rho\id_\mu^{(r)}\otimes\id_\mu^{(m)}\right].
\label{eq:twirling general}
\eal

We now show that this group twirling operation satisfies the condition for the operation that appears in Lemma~\ref{lem:twirling to resource}.
By assumption, $\ket{\Phi}$ is an eigenvector of $U_g$ with eigenvalue $e^{i\phi_g}$ for every $g\in G$.
This implies that $\ket{\Phi}$ is in the invariant subspace of $U(g)$ corresponding to the one-dimensional irreducible representation $\{e^{i\phi_g}\}_{g\in G}$.
We assign the label $\mu^\star$ for this representation.
Then, every $U(g)$ has the form 
\bal
 U(g) = \begin{pmatrix}
 e^{i\phi_g}\id_{\mu^\star}^{(m)}&\\
 &\bigoplus_{\mu\neq \mu^\star}U_{\mu}^{(r)}(g)\otimes \id_\mu^{(m)},
 \end{pmatrix}
\eal
implying that there exist $d_{\mu^\star}^{(m)}$ simultaneous eigenvectors of all $U(g)$'s with eigenvalues $\{e^{i\phi_g}\}_g$. 
Using the assumption that $\ket{\Phi}$ is the unique simultaneous eigenvector with eigenvalues $\{e^{i\phi_g}\}_g$, we identify $d_{\mu^\star}^{(m)}=1$. 
Noting $d_{\mu^\star}^{(r)}=1$ and $\id_{\mu^\star}^{(r)}=\Phi$, we get from \eqref{eq:twirling general} that
\bal
 \Xi(\rho) &= \Tr[\rho\Phi]\Phi+\sum_{\mu\neq\mu^\star}\frac{\id_\mu^{(r)}}{d_\mu^{(r)}}\otimes\Tr_r\left[\rho\id_\mu^{(r)}\otimes\id_\mu^{(m)}\right].
 \label{eq:twirling reference}
\eal

The second term is a completely positive map.
Also, the convexity of $\OO$ and the assumption that $U_g\cdot U_g^\dagger\in\OO$ imply that $\Xi\in\OO\subseteq\OO_{\max}$.
Thus, we can apply Lemma~\ref{lem:twirling to resource} to conclude the proof. 

\end{proof}


\section{Proofs of Propositions~\ref{pro:face state}--\ref{pro:Norrell state}}

We first show Proposition~\ref{pro:face state}.

\medskip

\noindent
\textit{Proof of Proposition~\ref{pro:face state}.}
The face state is the $+1$ eigenstate of the Clifford unitary $K\coloneqq SH$, which cycles Pauli operators as $X\to Z \to Y$. 
This implies that one can construct the dephasing map $\Lambda$ with respect to $\ket{F}$ and the $-1$ eigenstate $\ket{\overline{F}}$ of $K$  of the form 
\bal
 \Lambda(\rho)=\frac{1}{2}\rho + \frac{1}{2}K\rho K^{\dagger}=\Tr[F\rho] F + \Tr[(\id-F)\rho]\overline{F}.
\eal
Since this is realized by a probabilistic application of the Clifford gate $K$, it is a stabilizer operation and thus $\Lambda\in\OO_{\rm all}$.
We get the statement by applying Lemma~\ref{lem:twirling to resource}, or alternatively, Proposition~\ref{pro:group twirling}.
\qed
\medskip

This idea can be employed to prove Proposition~\ref{pro:Clifford magic}.

\medskip

\noindent
\textit{Proof of Proposition~\ref{pro:Clifford magic}.}
Consider a $t$-qudit system with local dimension $d$.   
Since $\phi$ is a stabilizer state, there exists a Clifford unitary $U$ such that $\ket{\phi}=U\ket{0}^{\otimes t}$ where $\ket{0}$ is the $+1$ eigenstate of $\{Z^j\}_{j=0}^{d-1}$ with $Z\coloneqq \sum_{j=0}^{d-1}e^{i2\pi j/d}\dm{j}$.
Since $U$ is Clifford, $\phi$ is the $+1$ eigenvalue of the $t(d-1)$ generalized Pauli operators $\left\{U^\dagger Z^j_k U\right\}_{k,j}$ where $Z_k$ is the $Z$ operator that only acts on the $k$\,th qudit and acts trivially on the other qudits. 
Since $V$ is in the 3rd level of the Clifford hierarchy, $\psi$ is the unique $+1$ eigenstate of $t(d-1)$ Clifford unitaries $\left\{V^\dagger U^\dagger Z^j_k U V\right\}_{k,j}$. 
Let $W_{jk}\coloneqq V^\dagger U^\dagger Z^j_k U V$ and $\ket{\psi_{\vec z}}\coloneqq VU\ket{\vec z}$ with $\vec z\in\{0,\dots,d-1\}^t$ be another eigenstate of $W_{jk}$ parameterized by $\vec z$, satisfying $W_{jk}\ket{\psi_{\vec z}}=e^{i2\pi  jz_k/d}\ket{\psi_{\vec z}}$.
The uniformly random application of Clifford unitaries $\left\{\prod_{k=1}^t W_{j_k k}\right\}_{j_1\dots j_t}$ works as a dephasing among $\{\psi_{\vec z}\}_{\vec z}$ because 
\begin{align}\begin{aligned}
&\frac{1}{d^t}\sum_{j_1=0}^{d-1}\dots\sum_{j_t=0}^{d-1}\left(\prod_{k=1}^t W_{j_k k}\right) \ketbra{\psi_{\vec z}}{\psi_{\vec{z'}}} \left(\prod_{k=1}^t W_{j_k k}^\dagger\right) \\&= 
\ketbra{\psi_{\vec z}}{\psi_{\vec{z'}}} \frac{1}{d^t}\sum_{j_1,\dots,j_t} \prod_{k=1}^t e^{i2\pi j_k (z_k-z_k')/d}\\
&= \ketbra{\psi_{\vec z}}{\psi_{\vec{z'}}} \prod_{k=1}^t \left(\frac{1}{d}\sum_{j=0}^{d-1} e^{i2\pi j (z_k-z_k')/d}\right)\\
&= \ketbra{\psi_{\vec z}}{\psi_{\vec{z'}}} \delta_{\vec{z}\vec{z'}}. 
\end{aligned}\end{align}
Thus, it acts as a projector onto $\{\psi_{\vec z}\}_{\vec z}$ as 
\bal
\frac{1}{d^t}\sum_{j_1\dots j_t}^{d-1}\left(\prod_{k=1}^t W_{j_k k}\right) \cdot \left(\prod_{k=1}^t W_{j_k k}^\dagger\right) = \sum_{\vec z} \dm{\psi_{\vec z}}\cdot \dm{\psi_{\vec z}}.
\eal

Importantly, this is a free operation because any $\prod_{k=1}^t W_{j_k k}$ is a Clifford unitary. 
On the other hand, this is clearly in $\S(\psi)$ and thus we can apply Lemma~\ref{lem:twirling to resource}, or alternatively, Proposition~\ref{pro:group twirling}.
\qed
\medskip

Next, we prove Proposition~\ref{pro:Hoggar state}.

\medskip

\noindent
\textit{Proof of Proposition~\ref{pro:Hoggar state}.} 
The crucial property of the Hoggar state is that it has a completely flat representation in the Pauli basis, in the sense that
\begin{equation}\begin{aligned}
  \left|\braket{{\rm Hog} | P | {\rm Hog}}\right| = \frac{1}{3}
  \label{eq:Hoggar overlap}
\end{aligned}\end{equation}
for every non-trivial Pauli operator $P \neq \id$.

On the one hand, we know that ${\rm Hog}$ is a maximizer of $D_{s,\STAB}$ with $D_{s,\STAB}({\rm Hog}) =\log\frac{12}{5}$~\cite{Howard2017application}. 
On the other hand, we will consider the stab-norm~\cite{Campbell2011catalysis}, defined for every $n$-qubit operator $A$ as
\begin{equation}\begin{aligned}
\| A \|_{\rm st} \coloneqq \frac{1}{2^n} \sum_{P \in \mathcal{P}}  \left|\Tr(A\, P)\right| ,
\end{aligned}\end{equation}
where $\mathcal{P}$ denotes all $n$-qubit Pauli operators. Crucially, since any stabilizer state satisfies $\| \sigma \|_{\rm st} \leq 1$, we can define the set
\begin{equation}\begin{aligned}
  \FF_{\mathcal{P}} \coloneqq \lset \sigma \sbar \sigma \geq 0,\; \Tr \sigma = 1,\; \| \sigma \|_{\rm st} \leq 1 \rset \supseteq \STAB.
\end{aligned}\end{equation}
Then $D_{\min,\STAB}(\rho) \geq D_{\min,\FF_{\mathcal{P}}}(\rho)$ for every state $\rho$.

Take then the Hoggar state, which can be written as
\begin{equation}\begin{aligned}
  \proj{\rm Hog} = \frac{1}{2^3} \left[ \id + \sum_{P \in \mathcal{P} \setminus \{\id\}} \frac{1}{3} e^{i \phi_P} P \right] ,
\end{aligned}\end{equation}
and consider any state $\sigma \in \FF_{\mathcal{P}}$, which can always be written as
\begin{equation}\begin{aligned}
  \sigma = \frac{1}{2^3} \left[ \id + \sum_{P \in \mathcal{P} \setminus \{\id\}} z_P P \right] 
\end{aligned}\end{equation}
with $\sum_P |z_P| \leq 2^3 - 1$ since $\| \sigma \|_{\rm st} \leq 1$. 
Then
\begin{equation}\begin{aligned}
  2^{-D_{\min,\FF_{\mathcal{P}}}({\rm Hog})} &\leq \braket{ \rm Hog | \sigma | \rm Hog}\\ 
  &\leq \frac{1}{2^3} \left[ 1 + \sum_{P \in \mathcal{P} \setminus \{\id\}} \frac{1}{3} e^{i \phi_P} z_P \right]\\
  &\leq \frac{1}{2^3} \left[ 1 + \frac{1}{3} \left(2^3 - 1\right)  \right]\\
  &= \frac{1}{8} + \frac{7}{24}\\
  &= \frac{5}{12}.
\end{aligned}\end{equation}
Using the fact that $D_{s,\STAB}({\rm Hog}) \geq D_{\max,\STAB}({\rm Hog}) \geq D_{\min,\STAB}({\rm Hog}) \geq D_{\min,\FF_{\mathcal{P}}}({\rm Hog})$ concludes the proof of the first part of the result.

To show the second part of the result, note that the Hoggar state $\ket{\rm Hog}$ is the $+1$ eigenvector of all unitaries in the group generated by two unitaries with order 7 and order 12~\cite{Zhu2012quantum}, defined as
\bal
\tilde U_7 \coloneqq \frac{\omega^5}{\sqrt{2}}\begin{pmatrix}
0 & 0 & 1 & 0 & -i & 0 & 0 & 0 \\
0 & 0 & i & 0 & -1 & 0 & 0 & 0 \\
0 & 0 & 0 & -i & 0 & -1 & 0 & 0 \\
0 & 0 & 0 & -1 & 0 & -i & 0 & 0 \\
1 & 0 & 0 & 0 & 0 & 0 & -i & 0 \\
-i & 0 & 0 & 0 & 0 & 0 & 1 & 0 \\
0 & -i & 0 & 0 & 0 & 0 & 0 & -1 \\
0 & 1 & 0 & 0 & 0 & 0 & 0 & i 
\end{pmatrix}
\eal
\bal
\tilde U_{12} \coloneqq \frac{\omega^3}{\sqrt{2}}\begin{pmatrix}
0 & 0 & 0 & 0 & 1 & i & 0 & 0 \\
0 & 0 & 0 & 0 & -1 & i & 0 & 0 \\
1 & -i & 0 & 0 & 0 & 0 & 0 & 0 \\
-1 & -i & 0 & 0 & 0 & 0 & 0 & 0 \\
0 & 0 & 1 & i & 0 & 0 & 0 & 0 \\
0 & 0 & -1 & i & 0 & 0 & 0 & 0 \\
0 & 0 & 0 & 0 & 0 & 0 & 1 & -i \\
0 & 0 & 0 & 0 & 0 & 0 & -1 & -i 
\end{pmatrix}
\eal
where $\omega=e^{2\pi i/8}$.
Let us call this group $G$ and consider the representation $\{U_g\}_{g\in G}$ that has the identical form to the element of the unitary group $G$.  As one can explicitly check, $\tilde U_{12}$ 
has a one-dimensional $+1$-eigenspace with the unique $+1$-eigenvector $\ket{\rm Hog}$, meaning that the Hoggar state is the unique simultaneous $+1$-eigenvector of all unitaries $\{U_g\}_{g\in G}$.

Since $\tilde U_7$ and $\tilde U_{12}$ are Clifford unitaries, $G$ is a subgroup of the three-qubit Clifford group, and hence $U_g\in\OO_{\rm STAB}\ \forall g\in G$. Noting also that $\OO_{\rm STAB}$ is completely free~\cite{Seddon2019quantifying}, we can apply Proposition~\ref{pro:group twirling} to conclude the proof. 
\qed
\medskip

We now focus on qutrit states. 
First, we show Proposition~\ref{pro:strange state}. 

\medskip

\noindent
\textit{Proof of Proposition~\ref{pro:strange state}.}
The Strange state is a fiducial state for a SIC-POVM for dimension 3.
Recall that a set $\{\dm{\psi_j}\}_{j=0}^{d^2-1}$ of projectors  is called a SIC-POVM if $\frac{1}{d}\sum_j \dm{\psi_j}=\id$ and $|\braket{\psi_i|\psi_j}|^2=\frac{d\delta_{ij}+1}{d+1}$.   
Moreover, a SIC-POVM is called covariant with group $G$ if any state $\ket{\psi_i}$ in the SIC-POVM can be constructed by applying some unitary representation $U_{g_i}, g_i\in G$ to a fiducial state $\ket{\psi_0}$ as $U_{g_i}\ket{\psi_0}=\ket{\psi_i}$.
The SIC-POVM generated from the Strange state is covariant with respect to the Heisenberg-Weyl group, and the Strange state is stabilized by a subset of the Clifford group that is isomorphic to the special linear group $\rm SL(2,\mathbb{Z}_3)$, which is the set of matrices with unit determinant whose entries are over $\mathbb{Z}_3$~\cite{Appleby2005symmmetric}. 
It was shown that there is a one-to-one correspondence between any $F\in {\rm SL}(2,\mathbb{Z}_3)$ and a Clifford unitary $U_F$ (up to global phase) acting as 
\bal
 U_FD_{\bf k}U_F^\dagger = D_{F{\bf k}} ,
\label{eq:Clifford}
\eal
where $D_{\bf k}=D_{k_1,k_2}=-e^{i\pi/d}X^{k_1}Z^{k_2}$~\cite{Appleby2005symmmetric,Zhu2010SIC}. 
Noting that the projectors in the SIC-POVM generated from the Strange state can be parameterized by ${\bf k}$ as $\dm{\psi_{\bf k}}=D_{\bf k}\dm{S}D_{\bf k}^\dagger$, and
\bal
U_F\dm{\psi_{\bf k}}U_F &= U_F D_{\bf k} U_F^\dagger \dm{S} U_F D_{\bf k}^\dagger U_F^\dagger\\
&= \dm{\psi_{F{\bf k}}}
\eal
where we used that $U_F\ket{S}=\ket{S}$, we get that for ${\bf k\neq 0}$
\bal
 \Lambda(\psi_{\bf k})&\coloneqq \frac{1}{|{\rm SL(2,\mathbb{Z}_3)}|}\sum_{F\in {\rm SL(2,\mathbb{Z}_3)}} U_F \psi_{\bf k} U_F^\dagger \\&=  \frac{1}{8}\sum_{\bf k'\neq 0} \psi_{\bf k'} \\
 &= \frac{3\id-\dm{S}}{8}.
 \label{eq:strange proof equations}
\eal
In the third equality, we used the fact that $\{\frac{1}{3}\psi_{\bf k}\}_{\bf k}$ constitutes a POVM and thus $\sum_{\bf k}\psi_{\bf k}= 3\id$.
The second equality follows from the following observation. 
The set of Clifford unitaries of the form \eqref{eq:Clifford} is the collection of all possible mappings from a Pauli operator to another Pauli operator. 
This means that a certain given non-identity Pauli operator $D_{\bf k}$ can be mapped to an arbitrary non-identity Pauli operator $D_{\bf k'}$ by some $U_F$. 
Let $\{U_{{\bf k}\to {\bf k'}}^{(j)}\}_{j=1}^{N_{\bf k'}}$ be the set of Clifford unitaries that map $D_{\bf k}$ to $D_{\bf k'}$ as in \eqref{eq:Clifford}, where $N_{\bf k'}$ is the number of such Clifford unitaries, and we set that $U_{{\bf k}\to {\bf k'}}^{(j_1)}\neq U_{{\bf k}\to {\bf k'}}^{(j_2)}$ for $j_1\neq j_2$.
The second equality of \eqref{eq:strange proof equations} is equivalent to showing that $N_{\bf k'}$ takes the same value for all ${\bf k'}\neq {\bf 0}$.
Suppose, to the contrary, that there exists ${\bf k'}$ such that $N_{\bf k'}>N_{\bf k''}$ for all ${\bf k''}\neq {\bf k'}$. 
Pick an arbitrary ${\bf k''}$ with ${\bf k''}\neq {\bf k'}$, and let $\{U_{{\bf k'}\to {\bf k''}}\}_{\bf k''\neq {\bf 0}}$ be a fixed set of unitaries that satisfies 
\bal
U_{{\bf k'}\to{\bf k''}}  D_{\bf k'}  U_{{\bf k'}\to{\bf k''}}^\dagger = D_{\bf k''}.
\eal
This implies that for every $j$, $U_{{\bf k'}\to{\bf k''}}U_{{\bf k}\to {\bf k'}}^{(j)}$ maps $D_{\bf k}$ to $D_{\bf k''}$ and thus is a member of the set $\{U_{{\bf k}\to{\bf k''}}^{(l)}\}_{l=1}^{N_{\bf k''}}$. 
Moreover, we have 
\bal
 U_{{\bf k'}\to{\bf k''}}U_{{\bf k}\to {\bf k'}}^{(j_1)}\neq U_{{\bf k'}\to{\bf k''}}U_{{\bf k}\to {\bf k'}}^{(j_2)},\quad \forall j_1\neq j_2
\eal
because otherwise it would result in $U_{{\bf k}\to {\bf k'}}^{(j_1)}= U_{{\bf k}\to {\bf k'}}^{(j_2)}$, violating the assumption. 
This implies that for every $U_{{\bf k}\to{\bf k'}}^{(j)}$, we can construct a corresponding element in $\{U_{{\bf k}\to{\bf k''}}^{(l)}\}_{l=1}^{N_{\bf k''}}$, which are distinct from each other for different $j$'s.
This shows that we must at least have $N_{\bf k''}\geq N_{\bf k'}$, but this violates the assumption that $N_{\bf k'}>N_{\bf k''}$. 
This concludes the proof that $N_{\bf k'}$ takes the same value for all ${\bf k'}\neq {\bf 0}$, leading to the second equality in \eqref{eq:strange proof equations}.

On the other hand, for ${\bf k=0}$, we get $\Lambda(\psi_{\bf 0})=\Lambda(S)=S$.
Since $\{\frac{1}{3}\psi_{\bf k}\}_{\bf k}$ is informationally complete, any state $\rho$ can be expanded as 
\bal
 \rho = \sum_{\bf k} c_{\bf k} \psi_{\bf k}.
\eal
This gives an action of $\Lambda$ for an arbitrary state $\rho$ as 
\bal
 \Lambda(\rho)&=c_{\bf 0} S + (1-c_{\bf 0})\frac{3\id-S}{8}\\
 &= (1-\epsilon_\rho)S + \epsilon_\rho \frac{\id-S}{2}
\eal
where we set $\epsilon_\rho\coloneqq \frac{3(1-c_{\bf 0})}{4}$.
In order for $\Lambda$ to be linear in $\rho$, $\epsilon_\rho$ must have the form $\epsilon_\rho=\Tr[H\rho]$ for some Hermitian operator $H$.
Combining this with the conditions $\Lambda(S)=S$ and $\Lambda\left(\frac{\id-S}{2}\right)=\frac{\id-S}{2}$,  we get
\bal
\Lambda(\rho)=\Tr[S\rho ]S+\Tr[(\id-S)\rho]\frac{\id-S}{2},
\eal
showing $\Lambda\in\S(S)$.

Moreover, $\frac{\id-S}{2}$ is a stabilizer state because $\Lambda(\dm{0})=\frac{\id-S}{2}$ while $\ket{0}$ is a stabilizer state and $\Lambda$ is a stabilizer operation. 
Indeed, we explicitly find
\bal
 \frac{\id-S}{2}=\frac{1}{4}(\dm{0}+\dm{x}+\dm{xz}+\dm{xz^2})
\eal
where $\ket{x},\ket{xz},\ket{xz^2}$ are the eigenstates of $X$ with eigenvalue $+1$, $XZ$ with eigenvalue $e^{i2\pi/3}$, and $XZ^2$ with eigenvalue $e^{-i2\pi/3}$ respectively, defined as 
\bal
 \ket{x}&\coloneqq \frac{1}{\sqrt{3}}(\ket{0}+\ket{1}+\ket{2})\\
 \ket{xz}&\coloneqq \frac{1}{\sqrt{3}}(\ket{0}+e^{-i2\pi/3}\ket{1}+e^{-i2\pi/3}\ket{2})\\
 \ket{xz^2}&\coloneqq \frac{1}{\sqrt{3}}(\ket{0}+e^{i2\pi/3}\ket{1}+e^{i2\pi/3}\ket{2}).
\eal

These allow us to apply Lemma~\ref{lem:twirling to resource} to get the statement.
The value for $D_{\min,\FF_{\rm STAB}}(S)$ can be explicitly calculated by computing the overlaps between the Strange state and all the pure stabilizer states and taking the logarithm of the inverse of the maximum overlap.  
\qed
\medskip

Finally, we show Proposition~\ref{pro:Norrell state}.

\medskip

\noindent
\textit{Proof of Proposition~\ref{pro:Norrell state}.}
$D_{s,\STAB}(N)\leq \log\frac{3}{2}$ is obtained by noting
\bal
 \dm{N}+\frac{1}{2}\dm{+}=\frac{1}{2}(\dm{1}+\dm{xz'}+\dm{xzz'})
\eal
where
\bal
 \ket{+}&\coloneqq \frac{1}{\sqrt{3}}(\ket{0}+\ket{1}+\ket{2})\\
 \ket{xz'}&\coloneqq \frac{1}{\sqrt{3}}(\ket{0}+e^{i2\pi/3}\ket{1}+\ket{2})\\
 \ket{xzz'}&\coloneqq \frac{1}{\sqrt{3}}(\ket{0}+e^{-i2\pi/3}\ket{1}+\ket{2})
\eal
are stabilizer states. 
On the other hand, the optimization in $D_{\min,\STAB}$ is achieved by
\bal
 D_{\min,\STAB}(N)=\log \braket{N|xzz'}^{-2}=\log \frac{3}{2}.
\eal
Using the inequalites $D_{\min,\STAB}(N)\leq D_{\max,\STAB}(N)\leq D_{s,\STAB}(N)$ concludes the proof.
\qed

\vspace{2cm}


\bibliographystyle{apsrmp4-2}
\bibliography{bib_ryuji,bib_bartosz}

\end{document}